\documentclass[11pt,oneside,english,reqno,a4paper]{article}
\usepackage[utf8]{inputenc}
\usepackage[pdfusetitle,colorlinks=true,linktoc=all]{hyperref} 
\usepackage{amsmath,amsthm,amsfonts,amssymb}     
\usepackage{lmodern}
\usepackage{textcomp}
\usepackage[T1]{fontenc}        
\usepackage[english]{babel}
\usepackage[numbers]{natbib}    
\usepackage[toc,page]{appendix}
\usepackage{graphicx}
\usepackage{pdfpages}

\usepackage{xcolor}
\usepackage{soulutf8}

\usepackage{caption}
\usepackage{enumitem}
\usepackage{url}

\usepackage{float}
\usepackage{chngcntr}
\counterwithin{figure}{section} 

\usepackage{tikz}
\usetikzlibrary{calc, decorations.pathmorphing}
\usepackage{tikz-cd}
\newdimen\Radius 

\usepackage{mathtools}

\usepackage{dsfont}
\usepackage{braket}

\usepackage{xspace}

\hypersetup{
  pdfkeywords={}, linkcolor=black, citecolor=black, filecolor=black, urlcolor=black,
}

\newenvironment{compactenum}
  { \begin{enumerate}[nolistsep,noitemsep,label=(\roman*)] }
  { \end{enumerate} }

\newtheorem{theorem}{Theorem}[section]
\newtheorem{proposition}[theorem]{Proposition}
\newtheorem{lemma}[theorem]{Lemma}
\newtheorem{corollary}[theorem]{Corollary}

\newtheorem{definition}[theorem]{Definition}
\theoremstyle{definition}
\newtheorem{notation}[theorem]{Notation}
\newtheorem{remark}[theorem]{Remark}
\newtheorem{example}[theorem]{Example}
\theoremstyle{remark}

\newcommand{\G}{\mathcal{G}}
\newcommand{\He}{\mathcal{H}}

\newcommand{\g}{\mathfrak{g}}
\newcommand{\h}{\mathfrak{h}}

\newcommand{\N}{\mathds N}
\newcommand{\F}{\mathds C}
\newcommand{\R}{\mathds R}
\newcommand{\C}{\mathds{C}}
\newcommand{\Z}{\mathds Z}

\newcommand{\FR}{F\!R}

\newcommand{\0}{\circ}
\newcommand{\ox}{\otimes}
\newcommand{\x}{\times}

\newcommand{\id}{\textnormal{id}}

\DeclareMathOperator{\Forall}{\;\forall\;}

\DeclareMathOperator{\End}{End}
\DeclareMathOperator{\Hom}{Hom}
\DeclareMathOperator{\Hol}{Hol}

\DeclareMathOperator{\Ad}{Ad}
\DeclareMathOperator{\ad}{ad}
\DeclareMathOperator{\Lie}{\mathbf{L}}

\begin{document}
\interfootnotelinepenalty=10000
\allowdisplaybreaks

\title{\texorpdfstring{\vspace{-1.5cm}}{}Poisson-geometric analogues of Kitaev models}
\date{February 24, 2020}
\author{Alexander Spies\\~\\ Department Mathematik \\
Friedrich-Alexander-Universität Erlangen-Nürnberg \\
Cauerstraße 11, 91058 Erlangen, Germany \\
}

\maketitle

\begin{abstract}
  We define Poisson-geometric analogues of Kitaev's lattice models.
  They are obtained from a Kitaev model on an embedded graph $\Gamma$ by replacing its Hopf algebraic data with Poisson data for a Poisson-Lie group $G$.

  Each edge is assigned a copy of the Heisenberg double $\He(G)$.
  Each vertex (face) of $\Gamma$ defines a Poisson action of $G$ (of  $G^*$) on the product of these Heisenberg doubles.
  The actions for a vertex and adjacent face form a Poisson action of the double Poisson-Lie group $D(G)$.
  We define Poisson counterparts of vertex and face \emph{operators} and relate them via the Poisson bracket to the vector fields generating the actions of $D(G)$.

  We construct an isomorphism of Poisson $D(G)$-spaces between this Poisson-geometrical Kitaev model and  Fock and Rosly's Poisson structure for the graph $\Gamma$ and the Poisson-Lie group $D(G)$.
  This decouples the latter and represents it as a product of Heisenberg doubles.
  It also relates the Poisson-geometrical Kitaev model to the symplectic structure on the moduli space of flat $D(G)$-bundles on an oriented surface with boundary constructed from $\Gamma$.
\end{abstract}

\textbf{Acknowledgements.}
I am very grateful to my advisor Catherine Meusburger for suggesting this topic and questions to investigate, for our numerous discussions, and for her comments that greatly improved this manuscript.
  Thomas Voß deserves my thanks for our insightful discussions about the similarities and differences between quantum and Poisson-geometrical Kitaev models.
 I also thank Studienstiftung des deutschen Volkes e.\ V.\ for granting me a scholarship, which made this article possible.

This research was supported in part by Perimeter Institute for Theoretical Physics. Research at Perimeter
Institute is supported by the Government of Canada through the Department of Innovation, Science and
Economic Development and by the Province of Ontario through the Ministry of Research and Innovation.

\section{Introduction}
  \textbf{Motivations.} Kitaev's lattice models were initially introduced to model a quantum memory on a surface that is governed by its topological properties and allows for intrinsically fault-tolerant quantum computation \cite{kitaev03}.
  Besides their role in topological quantum computing, Kitaev models have also gained popularity in other areas of current research as many interesting links have been discovered.
  They were related, for instance, to Levin and Wen's string-net models \cite{buerschaperaguado09}, which were introduced in the context of condensed matter physics \cite{levinwen05}.
  They can also be viewed as Hamiltonian analogues of 3d topological quantum field theories \cite{balsamkirillov12} of Turaev-Viro type \cite{turaev1992state, barrett1996invariants}.

  Recently, Meusburger related Kitaev models to the combinatorial quantization of Chern-Simons theories \cite{meusburger16}, that was obtained by Alekseev, Grosse and Schomerus \cite{alekseevgrosseschomerus94-1, alekseevgrosseschomerus94-2, alekseevschomerus96} and Buffenoir and Roche \cite{buffenoirroche95, buffenoirroche96}, and axiomatized as a Hopf algebra gauge theory by Meusburger and Wise \cite{meusburgerwise15}.
  The combinatorial models have been derived by quantizing a Poisson structure introduced by Fock and Rosly \cite{fockrosly98}.
  It describes the canonical symplectic structure on the moduli space of flat $G$-bundles, which is the (reduced) phase space of Chern-Simons theory.
  This indicates that there should be a Poisson-geometrical counterpart to Kitaev models and that it should be related to Fock and Rosly's Poisson structure in a similar way as Kitaev models are related to combinatorial quantization.
However, Kitaev models were defined ad hoc and not obtained by quantizing a Poisson manifold.

In this article we construct a Poisson analogue of Kitaev models and relate it to Fock and Rosly's Poisson structure and to moduli spaces of flat $G$-bundles.
This offers a new, Poisson-geometrical perspective on Kitaev models and provides an additional relation to Chern-Simons theory already on the classical level.
It also allows one to apply insights gained from Kitaev models to moduli spaces of flat $G$-bundles.
For instance, it follows from the contents of this article that the moduli space for a surface with boundary can be viewed as a Poisson counterpart to a Kitaev model with excitations, or quasi-particles, located at the boundary components.
Our results also allow for a decoupling of the symplectic structure on the moduli space by describing it in terms of a product Poisson manifold associated to a Poisson-geometrical Kitaev model.\footnote{
  These points are discussed in Remark \ref{remark:to_theorem_kitaev_moduli_space}.
}

\textbf{Kitaev models.}
Originally, Kitaev's lattice models were based on an embedded graph whose edges were decorated with elements of the group algebra of $\Z_2$ \cite{kitaev03}.
This was generalized first by Bombin and Martin-Delgado \cite{bombin2008family} to group algebras of finite groups, and  finally to finite-dimensional semi-simple Hopf $*$-algebras by Buerschaper et al. \cite{buerschaper2013hierarchy}.

A Kitaev model on an oriented surface $S$ is defined via an embedded graph $\Gamma$ with edge set $E$.
A copy of a finite-dimensional semi-simple Hopf $*$-algebra $H$ over $\C$ is assigned to each edge to obtain the \emph{extended space} $H^{\ox E}$.
It is known that the endomorphism algebra $\End_\C(H)$ is isomorphic to the Heisenberg double $\He(H)$ \cite{montgomery93}, so that the algebra of endomorphisms $\End_\C(H^{\ox E})$ of the extended space is isomorphic to the tensor product $\He(H)^{\ox E}$ of Heisenberg doubles \cite{meusburger16}.

To the vertices and faces of $\Gamma$ one associates vertex and face \emph{operators}.
The operators for a vertex and adjacent face can be combined to obtain a representation of the Drinfeld double $D(H)$ on $H^{\ox E}$ \cite{kitaev03, buerschaper2013hierarchy} and a $D(H)$-module algebra structure on the endomorphism algebra $\End_\C(H^{\ox E})$ \cite{meusburger16}.
The ground state of the Kitaev model, the \emph{protected space}, is the subspace of $H^{\ox E}$ on which all vertex and face operators act trivially.
Its endomorphism algebra is a subalgebra of the invariants of $\End_\C(H^{\ox E})$ with respect to the $D(H)$-module algebra structures \cite{meusburger16}.
  The protected space depends only on the homeomorphism class of the surface $S$ \cite{kitaev03, buerschaper2013hierarchy}.
Moreover, Balsam and Kirillov have shown in \cite{balsamkirillov12} that it is isomorphic to the vector space that topological quantum field theory of Turaev-Viro type \cite{turaev1992state,barrett1996invariants} (for the category $H$-Mod) assigns to $S$.
Meusburger showed that the topological invariant associated to a Kitaev model for the Hopf algebra $H$, the endomorphism algebra of its protected space, is isomorphic to the quantum moduli algebra obtained from combinatorial quantization of Chern-Simons theories for the Drinfeld double $D(H)$.

\textbf{Poisson analogues of Kitaev models.}
Poisson-Lie groups can be regarded as Poisson analogues of Hopf algebras and many structures and constructions for Hopf algebras have Poisson-Lie group counterparts.
Like Hopf algebras, each Poisson-Lie group $G$ possesses a dual Poisson-Lie group $G^*$ and there are Poisson-geometrical notions of Heisenberg and Drinfeld doubles.
Additionally, a counterpart to a module algebra over a Hopf algebra is given by the Poisson algebra $C^\infty(M, \R)$ of functions on a Poisson $G$-space $M$, that is, a Poisson manifold $M$ together with a Poisson action of the Poisson-Lie group $G$.
(See Table \ref{tab:dictionary}.)

\begin{table}
  \centering
  \begin{tabular}{|p{7.5cm}|p{7.5cm}|}
    \hline
    Hopf algebra $H$ & Poisson-Lie group $G$ \\
    \hline
    dual Hopf algebra  $H^*$ & dual Poisson-Lie group $G^*$\\
    \hline
    Drinfeld double $D(H)$ & classical double $D(G)$\\
    \hline
    Heisenberg double $\mathcal H(H)$ & Heisenberg double $\mathcal H(G)$  \\
    \hline
    module algebra over a Hopf algebra & Poisson algebra $C^\infty(M, \R)$ of functions on a Poisson $G$-space $M$ \\
    \hline
  \end{tabular}
  \caption{Hopf algebra structures and their Poisson-Lie group counterparts}
\label{tab:dictionary}
\end{table}

This correspondence suggests that an analogue of Kitaev models can be formulated with data from a Poisson-Lie group.
We define such an analogue, which we call a Poisson-Kitaev model, by replacing the Hopf-algebraic data of a Kitaev model by its Poisson-Lie counterparts.
For this we consider an embedded graph $\Gamma$ with edge set $E$ and assign a copy of the Heisenberg double $\He(G)$ of a Poisson-Lie group $G$ to every edge to obtain the product Poisson manifold $\He(G)^{\x E}$.
The Poisson algebra of functions $C^\infty(\He(G)^{\x E}, \R)$ is the counterpart to the endomorphism algebra of the extended space $\End_\C(H^{\ox E}) \cong \He(H)^{\ox E}$.

\textbf{Holonomies.} In \cite{meusburger16}, Meusburger defined a holonomy functor on the graph groupoid $\G(\Gamma_D)$ of the \emph{thickening} $\Gamma_D$ of $\Gamma$, which one obtains by replacing every edge of $\Gamma$ by a rectangle and each vertex by a polygon.
These holonomies associate to a path in $\G(\Gamma_D)$ endomorphisms of the extended space $H^{\ox E}$.
The endomorphism algebra of $H^{\ox E}$ is isomorphic to $\He(H)^{\ox E}$ and the Heisenberg double $\He(H)$ is isomorphic to $H \ox H^*$ as a vector space.
Vertex operators are obtained from paths in $\G(\Gamma_D)$ around vertices and are associated with elements of $H$, whereas face operators are obtained from paths around faces and are associated with elements of $H^*$ \cite{kitaev03, bombin2008family, buerschaper2013hierarchy, meusburger16}.

We define an analogous holonomy functor that assigns to paths in $\G(\Gamma_D)$ functions on $\He(G)^{\x E}$.
Poisson counterparts to vertex and face operators are obtained from paths around vertices and faces, respectively.
The Heisenberg double $\He(G)$ is locally diffeomorphic to $G \x G^*$.
In analogy to quantum Kitaev models, the holonomy around a face depends only on the $G$-components of the copies of $\He(G)$ associated to adjacent edges, and vertex holonomies depend only on the $G^*$-components.
This holonomy functor also allows us to define local flatness for Poisson-Kitaev models: an element of $\He(G)^{\x E}$ is flat at a vertex or face if the associated holonomy is trivial.

\textbf{Vertex and face actions.}
We define Poisson actions of $G$ and $G^*$ on $\He(G)^{\x E}$ associated to the vertices and faces of $\Gamma$, respectively.
We show that the actions for a pair of a vertex $v$ and adjacent face $f$ can be combined to obtain a Poisson action of the classical double $D(G)$ that corresponds to the module algebra structure of $\He(H)^{\ox E}$ over the Drinfeld double $D(H)$.
We relate the derivation $\left\{ h, - \right\}$ obtained from the Poisson bracket on $\He(G)^{\x E}$ and a vertex or face operator $h$ to the vector fields that generate the respective vertex or face action.
This is in analogy to the Hopf-algebraic picture, where the $D(H)$-module algebra structure is obtained from vertex and face operators \cite{meusburger16}.

\textbf{Relation to Fock and Rosly's Poisson structure.}
Fock and Rosly's Poisson-structure \cite{fockrosly98} is obtained by decorating the edges of an embedded graph $\Gamma$ with copies of a quasi-triangular Poisson-Lie group $D$.
The Poisson bracket on the manifold $D^{\x E}$ is obtained from the $r$-matrix of $D$, which is used to describe the interaction of edges incident at the same vertex.

We prove that the Poisson manifold $\He(G)^{\x E}$ of a Poisson-Kitaev model for a Poisson-Lie group $G$ is Poisson-isomorphic to Fock and Rosly's Poisson manifold $D(G)^{\x E}$ for the classical double $D(G)$.
The Poisson manifold $D(G)^{\x E}$ is equipped with a Poisson action of $D(G)$ for every vertex $v$ of $\Gamma$.
We show that the Poisson isomorphism intertwines this Poisson action with the one associated by the Poisson-Kitaev model to $v$ and an adjacent face.
As a manifold, $D(G)^{\x E}$ is diffeomorphic to $\He(G)^{\x E}$ but $D(G)^{\x E}$ is equipped with a Poisson bracket that involves non-trivial contributions between different copies of $D(G)$ for edges at a given vertex.
As the Poisson manifold associated with the Poisson-Kitaev model is the product Poisson manifold $\He(G)^{\x E}$, this Poisson isomorphism can be interpreted as decoupling the contributions associated with different edges in the Poisson bracket of $D(G)^{\x E}$.

\clearpage
\textbf{Relation to moduli spaces of flat $D(G)$-bundles.}
The natural symplectic structure on the moduli space of flat $D(G)$-bundles on a  compact oriented surface with boundary can be obtained from Fock and Rosly's Poisson manifold $D(G)^{\x E}$ for an embedded graph by taking the quotient with respect to  the $D(G)$-action for every vertex.
We use this fact to relate the Poisson algebra of functions on the moduli space with Poisson-Kitaev models.
We consider a  compact oriented  surface $S$ with at least one boundary component and an embedded graph $\Gamma$ such that $S$ is obtained by gluing an annulus or disk to each face of $\Gamma$.
To $S$  we associate  a submanifold $\He(G)^{\x E}_{flat}$ of the Poisson manifold $\He(G)^{\x E}$ of the Poisson-Kitaev model for $\Gamma$.
 It consists of elements that  have trivial holonomies around vertices and faces except for each face that corresponds to an annulus and for an adjacent vertex of every such face.
From the set of functions that are invariant under vertex and face actions on $\He(G)^{\x E}_{flat}$ we obtain a Poisson algebra that is isomorphic to the Poisson algebra of functions on the moduli space $\Hom(\pi_1(S), D(G))/D(G)$.

This isomorphism of Poisson algebras can be understood as a decoupling transformation.
It represents the symplectic structure of the moduli space on the direct Poisson product $\He(G)^{\x E}$.
It can be viewed as a generalization of a similar transformation constructed by Alekseev and Malkin \cite{alekseevmalkin95}, which represents $\Hom(\pi_1(S), G)/G$ on the product Poisson manifold $(G^*)^n \x \He(G)^g$, where $n$ is the number of boundary components and $g$ the genus of $S$.
The construction in \cite{alekseevmalkin95} is based on a set of generators of the fundamental group $\pi_1(S)$, whereas our construction works for more general graphs.

  \textbf{Further work.}
  Some constructions associated with Kitaev models are still lacking a Poisson counterpart, or a counterpart exists but deserves further investigation.
  For one thing, we only consider Poisson analogues of Kitaev models with excitations and do not construct an analogue of the ground state.
  More specifically, we relate the Poisson algebra of functions on the moduli space $\Hom(\pi_1(S), D(G))/D(G)$ for a surface $S$ with non-empty boundary to the endomorphism algebra of a state with excitations.\footnote{
    See Remarks \ref{remark:excitations} and \ref{remark:to_theorem_kitaev_moduli_space} (ii).
  }
  The excitations correspond to the boundary components of $S$.
  Therefore, it is plausible that the endomorphism algebra of the ground state of a Kitaev model (that is, the state without excitations) corresponds to the Poisson algebra of functions on the moduli space $\Hom(\pi_1(\tilde S), D(G))/D(G)$ for the surface $\tilde S$ obtained by gluing a disk to every boundary component of $S$.
  However, unlike the moduli space $\Hom(\pi_1(S), D(G))/D(G)$ for a surface with boundary, a precise relationship between $\Hom(\pi_1(\tilde S), D(G))/D(G)$ and Poisson-Kitaev models is yet to be established.

  Also, while we construct a holonomy functor analogous to the one introduced by Meusburger \cite{meusburger16}, which generalizes Kitaev's ribbon operators \cite{kitaev03, bombin2008family}, we do not investigate Poisson analogues of ribbon operators in detail.
  Ribbon operators are endomorphisms $H^{\ox E} \to H^{\ox E}$ of the extended space of a Kitaev model and implement quantum computations by creating excitations, moving them on the surface associated to the Kitaev model, and destroying them.
  As states with excitations correspond to moduli spaces $\Hom(\pi_1(S), D(G))/D(G)$ for surfaces $S$ with a boundary component for every excitation, Poisson analogues of ribbon operators might relate moduli spaces for surfaces with different numbers of boundary components.

  Moreover, we do not define a Poisson analogue of the Hamiltonian of a Kitaev model.
  The latter is obtained from vertex and face operators and the Haar integrals of the Hopf algebras $H$ and $H^*$ \cite{kitaev03, bombin2008family, buerschaper2013hierarchy}.
  As there are Poisson counterparts of vertex and face operators, one could attempt to obtain a Hamiltonian from these together with the Haar measures on the Poisson-Lie groups $G$ and $G^*$ associated with a Poisson-Kitaev model (which correspond to the Hopf algebras $H$ and $H^*$), at least for compact $G$ and $G^*$.

\textbf{Structure of the article.}
In Sections \ref{subsection:embedded_graphs} through \ref{subsection:double_poisson_lie_groups} we introduce the relevant background for the mathematical structures used in this article: graphs on surfaces and Poisson structures associated with Poisson-Lie groups and Poisson actions.
Fock and Rosly's description of the canonical symplectic structure on moduli spaces of flat $G$-bundles is summarized in Section \ref{subsection:fock_and_roslys_poisson_structure}.
We review Kitaev's lattice models in Section \ref{subsection:kitaev_models}.

In Section \ref{subsection:kitaev_models_with_poisson_geometric_data} we define Poisson-Kitaev models.
We introduce notions of holonomies and flatness, define analogues of vertex and face operators, and introduce vertex and face actions.
Section \ref{subsection:graph_trafos_kitaev} is dedicated to graph transformations.
They allow us to relate Poisson-Kitaev models on different lattices.
Poisson actions associated with vertices and faces are investigated in Section \ref{subsection:gauge_invariance_flatness} together with the Poisson subalgebra of invariant functions.
We prove that a pair of actions for a vertex and adjacent face can be combined into a Poisson $D(G)$-action.
In Section \ref{subsection:vertex_face_operators} we discuss the Poisson analogues of vertex and face operators and show that they are related to the Poisson-Lie group actions associated with the respective vertices and faces.

In Sections \ref{subsection:poisson_kitaev_models_fock_rosly_spaces} and \ref{subsection:relation_with_moduli_spaces} we prove our main results.
The isomorphism between Poisson-Kitaev models and Fock and Rosly's Poisson structure is shown in Section \ref{subsection:poisson_kitaev_models_fock_rosly_spaces}.
The relation between Poisson-Kitaev models and moduli spaces of flat $D(G)$-bundles is proven in Section \ref{subsection:relation_with_moduli_spaces}.

\section{Background}

\subsection{Embedded graphs}
\label{subsection:embedded_graphs}

In this section, we introduce the notion of ribbon graphs (also called fat graphs, embedded graphs or maps).
For more background on ribbon graphs see, for instance, \cite{ellismonaghanmoffatt13} or \cite{landozvonkin04}.
First we introduce some general terms for directed graphs.

\begin{definition} Let $\Gamma$ be a directed graph, $V$ its set of vertices and $E$ its set of edges.
  \label{def:graph_basics}
  \begin{enumerate}[noitemsep,nolistsep,label=(\roman*)]
    \item
      For an edge $e \in E$ we write $s(e), t(e) \in V$ for the source and target vertices of $e$, respectively.
      We set $s(e^{\mp}) = t(e^{\pm})$ for the edge $e^{-1}$ with reversed orientation.
      We say that $e$ is \textbf{incoming at} (\textbf{outgoing from}) $v$ if $t(e) =v$ and $s(e) \neq v$ ($t(e) \neq v$ and $s(e)=v$).
    \item
      We denote by $\G(\Gamma)$ the \textbf{graph groupoid} associated to $\Gamma$.
      This is the free groupoid generated by the edges $e \in E$ which we interpret as morphisms $e : s(e) \to t(e)$.
      A morphism $p : v_1 \to v_2$ is called a \textbf{path} from the vertex $v_1 \in V$ to $v_2 \in V$.
    \item
      The \textbf{graph subdivision} of $\Gamma$ is the directed graph $\Gamma_{\0}$ which is obtained by introducing a bivalent vertex on every edge $e \in E$, thus splitting $e$ into the two edges $i_s(e), i_t(e)$.
      The edges $i_s(e), i_t(e)$ are called the \textbf{edge ends} of $e$ and inherit its orientation, so that $s(i_s(e)) = s(e)$ and $t(i_t(e)) = t(e)$.
    \item
      A vertex $v \in V$ is called \textbf{$n$-valent}, if it has $n$ incident edge ends, and \textbf{univalent}, \textbf{bivalent}, \textbf{trivalent}, etc. if $n = 1, 2, 3, \dots$.
  \end{enumerate}
\end{definition}

\begin{definition}~
  \begin{enumerate}[noitemsep,nolistsep,label=(\roman*)]
    \item
      A \textbf{ribbon graph} is a finite directed graph $\Gamma$ together with a cyclic ordering of the edge ends at each vertex.
      If there is a \emph{linear} ordering at all vertices, we call $\Gamma$ a \textbf{ciliated ribbon graph}. 
    \item
      A path $p: v \to v$ in a ribbon graph is called a \textbf{face path}, if it traverses each edge $e \in E$ at most once in each direction and turns maximally right at every vertex.
      That is, it enters each vertex through an edge end $i$ and leaves it through the edge end $i'$ that is directly after $i$ with respect to the cyclic ordering, and the first edge end of the path comes directly after the last one.
    \item
      A \textbf{face} $f$ of the ribbon graph $\Gamma$ is an equivalence class of face paths under cyclic permutations.
      An edge $e$ of $f$ is \textbf{oriented clockwise} (\textbf{oriented counterclockwise}) with regard to $f$ if there is a face path $p = e_n^{\varepsilon_n} \0 \dots \0 e_1^{\varepsilon_1}$ of $f$ with $e = e_j$ and $\varepsilon_j = 1$ ($\varepsilon_j =-1$) for some $j$ and $p$ traverses $e$ exactly once.
  \end{enumerate}
\end{definition}

\begin{notation}
  In the following $\Gamma$ always denotes a ribbon graph.
  We write $V, E$ and $F$ for its sets of vertices, edges and faces, respectively.
  When we depict a piece of $\Gamma$, we will assume that the edge ends at vertices are ordered counterclockwise.
  This implies that face paths are traversed in a clockwise manner.
\end{notation}

\begin{remark}
  Ribbon graphs are in correspondence with graphs embedded into compact oriented surfaces (with boundary).
  To obtain the associated surface, glue an annulus to $\Gamma$ for each face $f \in F$ along an associated face path.

  Conversely, consider a finite directed graph $\Gamma$ that is embedded into  an oriented  surface $S$.
  Then the orientation of the surface $S$ naturally induces a cyclic ordering of the edge ends at every vertex $v \in V$, turning $\Gamma$ into a ribbon graph.
\end{remark}

The name \emph{ribbon graph} is motivated by the fact that ribbon graphs can be \emph{thickened}, thus turning each edge into a ribbon:

\begin{definition} Let $\Gamma$ be a ribbon graph.
  \label{def:thickening}
  \begin{enumerate}[noitemsep,nolistsep,label=(\roman*)]
    \item 
      The \textbf{thickening} of $\Gamma$ is the ribbon graph $\Gamma_D$ where each edge is replaced by a rectangle and each vertex by a polygon.
      An edge $e$ is replaced by the four edges $r(e), l(e), f(e), b(e)$.
      The edges $r(e), l(e)$ stand for the right and left side of $e$, respectively, and are oriented towards the vertex $t(e)$.
      The edges $f(e), b(e)$ stand for the edge ends of $e$ at $t(e)$ and $s(e)$, respectively, and are oriented towards the face left of $e$.
      This is illustrated in Figure \ref{fig:thickening}.
    \item
      The edges $r(e), l(e)$ are called the \textbf{edge sides} of $e$.
      As the edges $f(e), b(e)$ are in bijection to $i_t(e), i_s(e)$, we also refer to them as \textbf{edge ends}.
  \end{enumerate}
\end{definition}

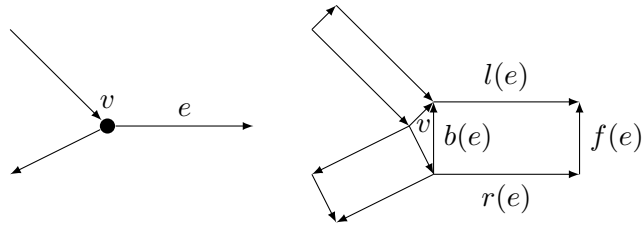
\begin{figure}[h]
\centering
\begin{tikzpicture}[vertex/.style={circle, fill=black, inner sep=0pt, minimum size=2mm}, plain/.style={draw=none, fill=none}, scale=0.8]
  \begin{scope}[scale=0.8]

  \node[vertex, label=90:{$v$}] (m) at (0,0) {};
  \coordinate (1) at (3,0) {};
  \coordinate (2) at (-2,2) {};
  \coordinate (3) at (-2,-1) {};

  \draw [-latex] (m) -- (3);
  \draw [-latex] (2) -- (m);
  \draw [-latex] (m) -- node[above] {$e$} ++ (1);

  \node[plain] (n) at ($ (m) + (6.5,0)$) {$v$};
  \coordinate (n2) at ($ (n) - (0.3,0)$) {};
  \coordinate (n3) at ($ (n2) + (0.5, -1)$) {};
  \coordinate (n1) at ($ (n2) + (0.5, 0.5)$) {};

  \draw [-latex] (n3) -- ($(n3) + (3,0)$) node[midway, below] {$r(e)$};
  \draw [-latex] (n1) -- ($(n1) + (3,0)$) node[midway, above] {$l(e)$};
  \draw [-latex] ($(n3) + (3,0)$) -- ($(n1) + (3,0)$) node[midway, right] {$f(e)$};
  \draw [-latex] (n3) -- (n1) node[midway, right] {$b(e)$};

  \draw [-latex] ($(n1) + (-2,2)$) -- (n1);
  \draw [-latex] ($(n2) + (-2,2)$) -- (n2);
  \draw [-latex] (n2) -- (n1);
  \draw [-latex] ($(n2) + (-2,2)$) -- ($(n1) + (-2,2)$);

  \draw [-latex] (n2) -- ($(n2) + (-2,-1)$);
  \draw [-latex] (n3) -- ($(n3) + (-2,-1)$);
  \draw [-latex] (n2) -- (n3);
  \draw [-latex] ($(n2) + (-2,-1)$) -- ($(n3) + (-2,-1)$);
    
  \end{scope}
\end{tikzpicture}
\caption{An edge $e$ of the ciliated ribbon graph $\Gamma$ and the corresponding edges $r(e), l(e), f(e), b(e)$ of the thickened graph $\Gamma_D$}
\label{fig:thickening}
\end{figure}

Note that the thickening $\Gamma_D$ obtains its ribbon graph structure naturally from that of $\Gamma$.
A face of $\Gamma_D$ corresponds to either a vertex, edge or face of $\Gamma$.
We do not distinguish between vertices, edges and faces of $\Gamma$ and the associated polygons of $\Gamma_D$.

Consider the dual $\Gamma^*$ of a ribbon graph $\Gamma$.
(The vertices of $\Gamma^*$ correspond to the faces of $\Gamma$.
For every edge separating the faces $f, f'$ of $\Gamma$ the graph $\Gamma^*$ has an edge connecting $f$ and $f'$.)
It is also a ribbon graph in a natural way:
the edge ends of $\Gamma^*$ correspond to the edge sides of $\Gamma$ and the cyclic ordering at a vertex of $\Gamma^*$ is obtained from any face path of the corresponding face in $\Gamma$.
The thickening of $\Gamma^*$ coincides with that of $\Gamma$ (up to edge orientation) if we switch the roles of edge ends and sides.
The dual of a \emph{ciliated} ribbon graph $\Gamma$, however, does not inherit a ciliated ribbon graph structure. 
For this reason, we consider ciliated ribbon graphs with the following additional structure.

\begin{definition}
  \label{def:doubly_ciliated_ribbon_graph}
  A \textbf{doubly ciliated ribbon graph} is a ciliated ribbon graph $\Gamma$ together with a choice of a face path for every face.
\end{definition}

The chosen face path for a face of $\Gamma$ equips it with a linear ordering of the adjacent edge  sides in the thickening $\Gamma_D$. 
Graphically, the orderings of a doubly ciliated ribbon graph $\Gamma$ can be expressed in $\Gamma_D$ by adding cilia to the polygons that represent vertices and faces, as is shown in Figure \ref{fig:ordering_and_site}.
\begin{definition} Let $\Gamma$ be a doubly ciliated ribbon graph.
  \label{def:double_graph_basics}
  \begin{enumerate}[nolistsep,noitemsep,label=(\roman*)]
    \item
      Let $v \in V$ and let $e \in E$ be the edge that contains the first edge end $i_v$ of $v$.
      The \textbf{face associated with $v$} is the face of $\Gamma$ incident to the right side $r(e)$ of $e$ if $i_v = b(e)$ (incident to the left side $l(e)$ of $e$ if $i_v=f(e)$).
      It is denoted by $f(v)$.
    \item
      Let $f \in F$ and let $e \in E$ be the edge that contains the first edge side $i_f$ of $f$.
      The \textbf{vertex associated with $f$} is the vertex  $s(e)$ if $i_f = r(e)$ ($t(e)$ if $i_f = l(e)$) and is denoted by $v(f)$.
    \item
      Let $v \in V, f \in F$ with $v(f) = v$ and $f(v) = f$.
      We say that the pair $(v,f)$ is a \textbf{site} if the first edge end and edge side of $v$ and $f$, respectively, are either $b(e)$ and $r(e)$ or $f(e)$ and $l(e)$ for some edge $e$.
    \item
      A doubly ciliated ribbon graph $\Gamma$ is called \textbf{paired} if
      \begin{enumerate}[nolistsep,noitemsep]
        \item 
          for each vertex $v \in V$ and face $f \in F$ the pairs $(v, f(v)), (v(f), f)$ are sites,
        \item
          it has no loops and
        \item
          for each edge $e \in E$ the face to the left of $e$ is different from the face to the right.
      \end{enumerate}
  \end{enumerate}
\end{definition}

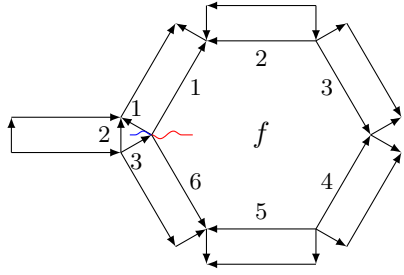
\begin{figure}
\centering
\begin{tikzpicture}[vertex/.style={circle, fill=black, inner sep=0pt, minimum size=2mm}, plain/.style={draw=none, fill=none}, scale=0.8]
  \begin{scope}[scale=0.9]

  \node at (0,0) {$f$};

  \coordinate (n) at (2,0);

  \coordinate (o) at (60:2cm);
  \draw [-latex] (o) -- (n) node[midway,left] {\footnotesize 3};
  \draw [-latex] (o) -- ($(o)!0.65cm!90:(n)$);
  \draw [-latex] (n) -- ($(n) + (o)!0.65cm!90:(n) - (o)$);
  \draw [-latex] ($(o)!0.65cm!90:(n)$) -- ($(n) + (o)!0.65cm!90:(n) - (o)$);

  \coordinate (p) at (120:2cm);
  \draw [-latex] (o) -- (p) node[midway,below] {\footnotesize 2};
  \draw [-latex] ($(o) + (0, 0.65)$) -- (o);
  \draw [-latex] ($(o) + (0, 0.65)$) -- ($(p) + (0, 0.65)$);
  \draw [-latex] ($(p) + (0, 0.65)$) -- (p);

  \coordinate (q) at (180:2cm);
  \draw [-latex] (q) -- (p) node[midway, right] {\footnotesize 1};
  \draw [-latex] (q) -- ($(q) ! 0.65cm ! 90:(p)$) node[midway,above] {\footnotesize $1$};
  \draw [-latex] (p) -- ($(p) + (q)!0.65cm!90:(p) - (q)$);
  \draw [-latex] ($(q)!0.65cm!90:(p)$) -- ($(p) + (q)!0.65cm!90:(p) - (q)$);

\draw [decorate,decoration={snake, amplitude=-0.5mm}, color=blue] (q) -- ($(q)-(0.4,0)$);
  \draw [decorate,decoration={snake, amplitude=-0.5mm}, color=red] (q) -- ($(q)+(0.75,0)$);

  \coordinate (r) at (240:2cm);
  \draw [-latex] (q) -- (r) node[midway, right] {\footnotesize 6};
  \draw [-latex] ($(q)!0.65cm!-90:(r)$) -- (q) node[midway, below] {\footnotesize $3$};
  \draw [-latex] ($(r) + (q)!0.65cm!-90:(r) - (q)$) -- (r);
  \draw [-latex] ($(q)!0.65cm!-90:(r)$) -- ($(r) + (q)!0.65cm!-90:(r) - (q)$);

  \coordinate (s) at (300:2cm);
  \draw [-latex] (s) -- (r) node[midway, above] {\footnotesize 5};
  \draw [-latex] (s) -- ($(s)!0.65cm!90:(r)$);
  \draw [-latex] (r) -- ($(r) + (s)!0.65cm!90:(r) - (s)$);
  \draw [-latex] ($(s)!0.65cm!90:(r)$) -- ($(r) + (s)!0.65cm!90:(r) - (s)$);

  \draw [-latex] (s) -- (n) node[midway, left] {\footnotesize 4};
  \draw [-latex] (s) -- ($(s)!0.65cm!-90:(n)$);
  \draw [-latex] (n) -- ($(n) + (s)!0.65cm!-90:(n) - (s)$);
  \draw [-latex] ($(s)!0.65cm!-90:(n)$) -- ($(n) + (s)!0.65cm!-90:(n) - (s)$);

  \coordinate (q1) at ($(q)!0.65cm!90:(p)$);
  \coordinate (q2) at ($(q)!0.65cm!-90:(r)$);
  \coordinate (t) at ($(q1)!2cm!-90:(q2)$);
  \draw [-latex] (t) -- (q1);
  \draw [-latex] (q2) -- (q1) node[midway,left] {\footnotesize 2};
  \draw [-latex] ($(t) + (q2) - (q1)$) -- (t);
  \draw [-latex] ($(t) + (q2) - (q1)$) -- (q2);

  \end{scope}
\end{tikzpicture}
\caption{ The thickening $\Gamma_D$ of a doubly ciliated ribbon graph $\Gamma$.
  The orderings at the face $f$ of $\Gamma$ and at an adjacent vertex
  are indicated by the cilia (red and blue, respectively) }
\label{fig:ordering_and_site}
\end{figure}

In Figure \ref{fig:ordering_and_site} the polygon that contains the blue cilium corresponds to the vertex $v(f)$ of $\Gamma$.
Moreover, the pair $(v(f), f)$ is a site.
Graphically, the condition that $(v(f), f)$ is a site means that the cilia of $v(f)$ and $f$ are attached to the same vertex in $\Gamma_D$.

\subsection{Poisson-Lie groups and Poisson \texorpdfstring{$G$}{G}-spaces}

In this section we introduce the data for Fock and Rosly's Poisson structure \cite{fockrosly98} and for Poisson-Kitaev models:
Poisson-Lie groups and manifolds with Poisson actions of Poisson-Lie groups.
We mostly follow the presentation in \cite{charipressley94}.

\begin{definition} \cite{drinfeld83}
  \begin{compactenum}
    \item
      A Lie group $G$ that is also a Poisson manifold is called a \textbf{Poisson-Lie group} if the multiplication map $\mu : G \x G \to G$ is Poisson with respect to the product Poisson structure on $G \x G$.
    \item
      A \textbf{homomorphism} $\Phi : G \to H$ of Poisson-Lie groups is a homomorphism of Lie groups that is also a Poisson map.
    \item
      A \textbf{Poisson-Lie subgroup} $H$ of $G$ is a Lie subgroup that is also a Poisson submanifold: $H$ is a Poisson-Lie group together with an injective immersion $\iota : H \to G$ that is a group homomorphism and a Poisson map.
  \end{compactenum}
\end{definition}

(Poisson-)Lie groups and smooth manifolds are always assumed finite dimensional.
We do not require submanifolds to be embedded.

The Poisson structure of a Poisson-Lie group $G$ equips its Lie algebra $\Lie(G)$ with the structure of a \emph{Lie bialgebra}.

\clearpage
\begin{definition}~
  \begin{compactenum}
    \item
    A Lie algebra $(\g, \left[ \, , \right])$ together with a skew symmetric linear map $\delta: \g \to \g \ox \g$, called the \textbf{cocommutator} of $\g$, is called a \textbf{Lie bialgebra} if $\delta^* : \g^* \ox \g^* \to \g^*$ defines a Lie bracket on $\g^*$ and
    $\delta$ is a 1-cocycle of $\g$ with values in $\g \ox \g$:
    \[
      \delta [X,Y] = (\ad_X \ox \, \id_\g + \id_\g \ox \ad_X) \, \delta Y - (\ad_Y \ox \, \id_\g + \id_\g \ox \ad_Y) \, \delta X \qquad \Forall X,Y \in \g \, .
    \]
    \item
      A \textbf{homomorphism of Lie bialgebras} $\varphi : \g \to \h$ is a homomorphism of Lie algebras such that $\varphi^* : \h^* \to \g^*$ is also a homomorphism of Lie algebras.
  \end{compactenum}
\end{definition}

Let $(G, \left\{ \, , \right\})$ be a Poisson-Lie group with associated Lie algebra $\g := \Lie(G)$.
We express the Poisson bracket $\left\{ \, , \right\}: C^\infty(G, \R)^{2} \to C^\infty(G, \R)$ by the corresponding Poisson bivector $w \in \Gamma(\bigwedge^2 TG)$:
\[
  \left\{ f_1, f_2 \right\}(g) = \langle df_1 \ox df_2, w(g) \rangle \quad \Forall g \in G \, .
\]
Denote by $L_g: G \to G$ the left multiplication $h \mapsto g \cdot h$ with an element $g \in G$ and by $R_g : G \to G$ the right multiplication.
We write $TL_g$ and $TR_g$ for the associated tangent maps.

\begin{theorem} \cite{drinfeld83}
  \label{theorem:poisson_lie_group_lie_bialgebra}
  \begin{compactenum}
  \item
    The Lie algebra $\g$ is a Lie bialgebra with cocommutator 
    \[
      \delta := T_1 w^R: \g \to \g \ox \g \quad \text{ with } \quad
      w^R : G \to T_1G^{\ox 2} = \g^{\ox 2} \quad g \mapsto TR_{g^{-1}}^{\ox 2} \, w(g) \, .
    \]
    It is called the \textbf{tangent Lie bialgebra} associated to $G$.

    If $H$ is a connected and simply connected Lie group, then every Lie bialgebra structure $(\h, \delta)$ on its Lie algebra $\h$ equips $H$ with a unique structure of a Poisson-Lie group so that $\h$ is the tangent Lie bialgebra $of H$.

  \item
    If $\Phi: G \to H$ is a homomorphism of Poisson-Lie groups, then the derivative at the unit $T_1 \Phi: \Lie(G) \to \Lie(H)$ is a homomorphism of Lie bialgebras.

    If $G$ is connected and simply connected, then for every homomorphism $\varphi: \Lie(G) \to \Lie(H)$ of Lie bialgebras there is a unique homomorphism $\Phi: G \to H$ of Poisson-Lie groups such that $\varphi = T_1 \Phi$.
  \end{compactenum}
\end{theorem}

For any finite-dimensional Lie bialgebra $(\g, \left[ \, , \right], \delta)$ the dual $\g^*$ has a canonical Lie bialgebra structure with commutator given by $\delta^*: \g^* \ox \g^* \cong (\g \ox \g)^* \to \g^*$ and cocommutator $\left[ \, , \right]^* : \g^* \to \g^* \ox \g^*$.

\begin{definition}~
  \begin{compactenum}
  \item
    The Lie bialgebra $(\g^*, \delta^*, \left[ \, , \right]^*)$ is called the \textbf{dual Lie bialgebra} of $\g$.
  \item
    Let $G$ be a Poisson-Lie group with associated Lie bialgebra $(\g, \left[ \, , \right], \delta)$.
      Any Poisson-Lie group with tangent Lie bialgebra $(\g^*, \delta^*, \left[ \, , \right]^*)$ is called \textbf{dual} to $G$.
      The unique connected and simply connected Poisson-Lie group $G^*$ dual to $G$ is called the \textbf{dual Poisson-Lie group} of $G$.
  \end{compactenum}
\end{definition}

\begin{notation}
  For an element $r \in \g \ox \g$ we write $r = r_{(1)} \ox r_{(2)}$, suppressing the sum.
  Denote by $r_{21} := r_{(2)} \ox r_{(1)}$ the flipped element.
  We write $r_a := \frac{1}{2} (r - r_{21})$ and $r_s := \frac{1}{2}(r + r_{21})$ for the antisymmetric and symmetric components of $r$, respectively.
\end{notation}

\begin{definition}~
  \begin{compactenum}
  \item
    Let $\g$ be a Lie algebra.
    The \textbf{classical Yang-Baxter equation}, or \textbf{CYBE}, for an element $r \in \g \ox \g$ is the equation
    \begin{equation}
      \label{eq:cybe}
      \left[ \left[ r,r \right] \right] := [r_{12}, r_{13}] + [r_{12}, r_{23}] + [r_{13}, r_{23}] = 0 \, ,
    \end{equation}
    where we use the following notation for $r, s \in \g \ox \g$:
    \begin{align*}
      [r_{12}, s_{13}] &:= [r_{(1)}, s_{(1)}] \ox r_{(2)} \ox s_{(2)} \\
      [r_{12}, s_{23}] &:= r_{(1)} \ox [r_{(2)}, s_{(1)}] \ox s_{(2)} \\
      [r_{13}, s_{23}] &:= r_{(1)} \ox s_{(1)} \ox [r_{(2)}, s_{(2)}] \, .
    \end{align*}
    A solution of the CYBE is called a \textbf{classical $r$-matrix}.

  \item
    A Poisson-Lie group $G$ is called \textbf{coboundary} if its Poisson bivector is of the form
    \begin{equation}
      w(g) = (TL_g^{\ox 2} - TR_g^{\ox 2}) \; r
      \label{eq:SklyaninBivector}
    \end{equation}
    for an element $r \in \g \ox \g$.
    It is called \textbf{quasi-triangular} if $r$ is a classical $r$-matrix, and \textbf{triangular} if in addition $r$ is antisymmetric.
  \end{compactenum}
\end{definition}

\begin{theorem} \cite{semenovtianshansky85}
  Equation \eqref{eq:SklyaninBivector} defines the structure of a Poisson-Lie group on $G$ if and only if both the symmetric component $r_s$ and the element $ \left[ \left[ r,r \right] \right] $ are invariant under the adjoint representation of $G$ on $\g \ox \g$ and $\g^{\ox 3}$, respectively.
\end{theorem}

In particular, if $r_s$ is $\Ad$-invariant and $r$ is a classical $r$-matrix, Equation \eqref{eq:SklyaninBivector} defines the structure of a quasi-triangular Poisson-Lie group on $G$.

\begin{definition}
  \label{def:poisson_g_space}
  Let $G$ be a Poisson-Lie group.
  \begin{compactenum}
  \item
    A Poisson manifold $M$ together with a group action $\rhd : G \x M \to M$ that is also a Poisson map is called a \textbf{Poisson $G$-space}.
  \item
    A map $\varphi: M \to N$ between Poisson $G$-spaces is called a \textbf{  homomorphism  of Poisson $G$-spaces} if it is Poisson and intertwines the $G$-actions:
  \[
    \varphi(g \rhd m) = g \rhd \varphi(m) \quad \Forall m \in M, g\in G \, .
  \]
  \end{compactenum}
\end{definition}

Any Poisson-Lie group $G$ together with the multiplication map $\mu: G \x G \to G$ is a Poisson $G$-space.
If $G$ is coboundary, then there is another Poisson $G$-space structure on $G$. Its Poisson bivector is also obtained from the element $r \in \g \ox \g$ from \eqref{eq:SklyaninBivector}: 
  \begin{equation}
    \label{eq:bivector_heisenberg_double}
    w_{\He} : g \mapsto -(TL_g^{\ox 2} + TR_g^{\ox2}) \, r_a \, .
  \end{equation}
\begin{theorem} \cite{semenovtianshansky85,  semenovtianshansky92}
  \label{theorem:heisenberg_double_inversion_poisson_actions}
  Let $G$ be a coboundary Poisson-Lie group. 
  \begin{compactenum}
    \item
      The pair $G_\He := (G, w_\He)$ is a Poisson manifold. 
    \item
      The Poisson manifold  $G_\He$  becomes a Poisson $G$-space when equipped with any of the following Poisson actions:
      \begin{align*}
        \rhd : G \x  G_\He \to G_\He  \quad & \quad (g, h) \mapsto gh \\
        \rhd' : G \x  G_\He \to G_\He  \quad & \quad (g,h) \mapsto h g^{-1} \, .
      \end{align*}
      The inversion map  $\eta : (G_\He, \rhd) \to (G_\He, \rhd'), g \mapsto g^{-1}$  is an isomorphism of Poisson $G$-spaces.
  \end{compactenum}
\end{theorem}

\begin{remark}
  Note that while the right multiplication on the Poisson-Lie group $G$ is Poisson, the inversion map $\eta : G \to G$ is \emph{anti-Poisson}.
  Hence, the action $\rhd' : G \x G \to G$ from Theorem \ref{theorem:heisenberg_double_inversion_poisson_actions} \emph{does not equip} $G$ with a Poisson $G$-space structure.
  It does, however, equip $G$ with a Poisson $(G, -w_G)$-space structure for the negative Poisson bivector $-w_G$ on $G$.
\end{remark}

\subsection{Double Poisson-Lie groups}
\label{subsection:double_poisson_lie_groups}

Similarly to the Drinfeld double of a Hopf algebra, to any Poisson-Lie group $G$ we can associate a quasi-triangular Poisson-Lie group $D(G)$, the \emph{double} of $G$.

We define the double of a Lie bialgebra as in \cite{charipressley94}.
Let $(\g, [,], \delta)$ be a Lie bialgebra with dual Lie bialgebra $(\g^*, [,]_{\g^*}, \delta_{\g^*})$, where $[,]_{\g^*} = \delta^*$ and $\delta_{\g^*} = [,]^*$.
Denote by $\g^{*cop}$ the Lie bialgebra $\g^*$ with opposite cocommutator.
Consider the symmetric bilinear form $(\, ,) : (\g \oplus \g^*)^{\ox 2} \to \R$ given by
\begin{equation}
  \label{eq:symmetric_bilinear_form}
  (x,y) = (\alpha, \beta) = 0 \qquad
  (x,\alpha) = \alpha(x) 
  \qquad \text{ for } \quad x,y \in \g, \alpha,\beta \in \g^* \, .
\end{equation}

\begin{theorem}\cite{drinfeld83}
  \label{theorem:double_lie_bialgebra}
  There is a unique Lie bialgebra structure on the vector space $\g \oplus \g^*$ such that the inclusions $\g \to \g \oplus \g^*, \g^{*cop} \to \g \oplus \g^*$ are homomorphisms of Lie bialgebras and $(\, ,)$ is invariant under the adjoint representation.
  The Lie bracket is given by
  \[
    [x, \alpha] = (\alpha \ox \id_\g) \, \delta(x) + \ad_x^*(\alpha)  \quad \text{ for } x \in \g, \alpha \in \g^* \, ,
  \]
  where $\ad_x^*(\alpha) = - \alpha \0 \ad_x$ denotes the coadjoint representation.
  This Lie bialgebra structure on $\g \oplus \g^*$ is quasi-triangular with classical $r$-matrix $r = \id_\g$, where we view $\id_\g$ as an element of $\g \ox \g^* \subseteq (\g \oplus \g^*)^{\ox 2}$.
\end{theorem}

\begin{definition}~
  \begin{compactenum}
  \item
    The Lie bialgebra structure from Theorem \ref{theorem:double_lie_bialgebra} is called the \textbf{double} of $\g$ and is denoted by $D(\g)$.
    Lie bialgebras of this form are called \textbf{double Lie bialgebras}.
  \item
    A Poisson-Lie group $G$ is called a \textbf{double Poisson-Lie group} if its tangent Lie bialgebra is a double Lie bialgebra.
  \item
    For a Poisson-Lie group $H$ with Lie bialgebra $\h$ we call the unique connected and simply connected Poisson-Lie group with tangent Lie bialgebra $D(\h)$ the \textbf{double} $D(H)$ of $H$.
  \end{compactenum}
\end{definition}

\begin{notation}
  For a double Lie bialgebra $\g = D(\h)$ of a Lie bialgebra $\h$ we denote the Lie subbialgebras $\h, \h^{*cop}$ by $ \g_- := \h $ and $ \g_+ := \h^{*cop} $.
\end{notation}

The Heisenberg double of a Hopf algebra also has a Poisson-geometrical counterpart:
\begin{definition}\cite{alekseevmalkin94}
  \label{def:classical_heisenberg_double}
  Consider a double Poisson-Lie group $G$.
  \begin{compactenum}
  \item
    The Poisson manifold $G_\He = (G, w_{\He})$ with the Poisson bivector $w_\He$ from \eqref{eq:bivector_heisenberg_double} is called a \textbf{Heisenberg double}.
  \item
    If $G = D(H)$ is the double of a Poisson-Lie group $H$, we call $\He(H) := G_\He = D(H)_\He$ the \textbf{Heisenberg double} of $H$.
  \end{compactenum}
\end{definition}

For any quasi-triangular Poisson-Lie group $G$ with tangent Lie bialgebra $\g$ the  classical $r$-matrix $r \in \g \ox \g$ of $G$  defines two linear maps
\[
  \sigma_+, \sigma_- : \g^* \to \g \qquad \sigma_+(\alpha) := (\alpha \ox \id_\g) \, r \qquad \sigma_-(\alpha) := -(\id_\g \ox \alpha) \, r \, .
\]
These maps are used to relate the Poisson-Lie group $G$ with its dual $G^*$ if the Lie bialgebra $\g$  is \emph{factorizable}, that is, if the symmetric component $r_s$ of its $r$-matrix is non-degenerate.
(See, for instance, \cite{reshetikhinsts88} or \cite{weinsteinxu92}.)
This is the case if $G$ is a double Poisson-Lie group.
The CYBE implies that the maps $\sigma_\pm$ are homomorphisms of Lie algebras, where $\g^*$ is equipped with the commutator $\delta^*$.
Thus, they can be integrated to homomorphisms of Lie groups $S_+, S_- : G^* \to G$.

For the remainder of Section \ref{subsection:double_poisson_lie_groups} we assume that $G$ is a double Poisson-Lie group.
Then $\g = D(\h)$ for some Lie bialgebra $\h$ and the images $\sigma_+(\g^*), \sigma_-(\g^*)$ coincide with the Lie subbialgebras $\g_+ = \h^{*cop}$ and $\g_- = \h$, respectively.
The CYBE also implies that $\sigma_\pm$ are anti-Lie coalgebra homomorphisms.
Therefore, the maps $S_\pm : G^* \to G$ are anti-Poisson and the images $G_+ := S_+(G^*), G_- := S_-(G^*)$ are Poisson-Lie subgroups.
 If the double Poisson-Lie group $G$ is connected and simply connected, then it is the double $D(G_-)$ of $G_-$.
If $G_+$ is simply connected, then it is the dual Poisson-Lie group of $G_-$ with negative Poisson structure.

The map $t : \g^* \to \g, \alpha \mapsto \sigma_+(\alpha) - \sigma_-(\alpha)$ is a linear isomorphism as the symmetric component $r_s$ of $r$ is non-degenerate.
Differentiating  the map $S : G^* \to G, \beta \mapsto S_-(\beta)^{-1} S_+(\beta)$ at the unit yields $T_1 S = t$, so that $S$ is a local diffeomorphism in a neighbourhood of $1 \in G^*$.
Computing the pullback of the Poisson bivector of $G^*$ along the local inverse of $S$ allows us to (locally) express the Poisson structure of $G^*$ on $G$.

\begin{lemma}
  \label{lemma:dual_poisson_lie_group}
  The bivector
  \begin{equation}
    \label{eq:bivector_dual_poisson_lie_group}
    w_{G^*} (g) := -(TL_g^{\ox 2} + TR_g^{\ox 2}) \, r_a - TL_g \ox TR_g \, r_{21} + TR_g \ox TL_g \, r
  \end{equation}
  defines a Poisson structure on $G$.
  The smooth map $S: G^* \to (G, w_{G^*})$ is Poisson and a local diffeomorphism in a neighbourhood of the unit $1 \in G^*$.
\end{lemma}

A proof can be found in \cite[Section 14.7]{babelon03}.
In a neighbourhood $U \subseteq G$ of the unit we obtain from $S$ two local projections $\pi_\pm : U \to G_\pm$:
\[
  \pi_- \left( S_-(\beta)^{-1} S_+(\beta) \right) := S_-(\beta)^{-1} \qquad \pi_+ \left( S_-(\beta)^{-1} S_+(\beta) \right) := S_+(\beta) \, .
\]
As shown by Lu and Weinstein \cite[Theorem 3.14]{luweinstein1990}, the induced (local) right action of $G_+$ on $G_-$ defined by $ (g_-, g_+) \mapsto \pi_-(g_+^{-1} g_-) $ coincides with the right dressing action of the Poisson-Lie group $(G_+, -w_{G_+})$ with negative Poisson bivector on $G_-$.
Therefore, the projection $\pi_-: U \to G_-$ is Poisson and the same can be shown for $\pi_+: U \to G_+$.
In terms of the Poisson bivector $w_G$ of $G$ this means $T\pi_\pm^{\ox 2} \, w_G|_U = w_G \0 \pi_\pm$.
A direct computation shows that $T\pi_\pm^{\ox 2} \, w_{\He}|_U $ and $T\pi_\pm^{\ox 2} \, w_{G^*}|_U$ coincide up to a sign with $T\pi_\pm^{\ox 2} \, w_G|_U$ for the Poisson bivectors $w_{\He}$ from \eqref{eq:bivector_heisenberg_double} and $w_{G^*}$ from \eqref{eq:bivector_dual_poisson_lie_group}.
The following lemma summarizes these results.

\begin{lemma}[Poisson properties of $\pi_\pm$]
  \label{lemma:projectionLocallyPoisson}
  There is an open neighbourhood $U \subseteq G$ of the unit and
  \begin{compactenum}
    \item
      all $g \in U$ can be uniquely factorized as $g = g_- g_+$ with $g_\pm \in G_\pm$,
    \item
      the map $\pi_- : U \to G_- \subseteq G, g \mapsto g_-$ is Poisson if $U$ is equipped with the Poisson bivector $w$ from \eqref{eq:SklyaninBivector}, $w_{G^*}$ from \eqref{eq:bivector_dual_poisson_lie_group} or $w_{\He}$ from \eqref{eq:bivector_heisenberg_double},
    \item
      the map $\pi_+ : U \to G_+ \subseteq G, g \mapsto g_+$ is Poisson if $U$ is equipped with $w$ and anti-Poisson for the Poisson bivectors $w_{G^*}$ and $w_{\He}$.
  \end{compactenum}
\end{lemma}

\begin{definition}
  \label{def:global_double_poisson_lie_group}
  A double Poisson-Lie group $G$ is called a \textbf{global double Poisson-Lie group} if there are Poisson-Lie subgroups $G_+, G_- \subseteq G$ with tangent Lie bialgebras $\g_+, \g_- \subseteq \g$, respectively, and the map $\mu : G_+ \x G_- \to G, (g_+, g_-) \mapsto g_- g_+ $ is a diffeomorphism.
\end{definition}

\begin{example}\cite[Theorem 3.7]{luweinstein1990}
  Let $G$ be a connected and simply connected double Poisson-Lie group and $G_+, G_-$ the connected Poisson-Lie subgroups with tangent Lie bialgebra $\g_+$ and $\g_-$, respectively.
  If $G_-$ is compact and $G_+$ is closed in $G$, then the map $\mu: G_+ \x G_- \to G$ from Definition \ref{def:global_double_poisson_lie_group} is a diffeomorphism.
\end{example}

We conclude this section by deriving certain identities for the projections $\pi_\pm$ that will be frequently used in the following.

\begin{lemma}[Computation rules for $\pi_\pm$]
  Let $G$ be a global double Poisson-Lie group. Then:
  \label{lemma:global_double_projections_properties}
  \begin{compactenum}
  \item
    $\displaystyle \qquad\qquad \pi_-(g \, \pi_-(h)) = \pi_-(gh) \quad \text{and} \quad \pi_+(\pi_+(g) \, h) = \pi_+(gh) \Forall g, h \in G $,
  \item
    $\displaystyle \qquad\qquad \pi_-(xg) = x \pi_-(g) \quad \text{and} \quad \pi_+(g\alpha) = \pi_+(g) \alpha \Forall g \in G, x \in G_-, \alpha \in G_+ $.
  \end{compactenum}
\end{lemma}

\begin{proof}
  This follows directly from the unique decomposition $g = g_- g_+$ for all elements $g \in G$.
\end{proof}

\begin{lemma}[Relations of $\pi_\pm$ with $w_{\He}$  and the $r$-matrix]
  Let $G$ be a global double Poisson-Lie group with classical $r$-matrix $r \in \g \ox \g$.
  \label{lemma:global_double_projections_r_matrix}
  \begin{compactenum}
  \item
    For all $d,g \in G$ the following equations hold:
    \begin{align}
      \label{eq:global_double_projections_r_matrix_i_1}
      (\id \ox T(\pi_+ \0 R_d)) \, w_{\He}  (g) &= - (\id \ox T(\pi_+ \0 R_d)) \0 TL_g^{\ox 2} \, r \\
      \label{eq:global_double_projections_r_matrix_i_2}
      (T(\pi_- \0 L_d) \ox \id) \, w_{\He}  (g) &= - (T(\pi_- \0 L_d) \ox \id) \0 TR_g^{\ox 2} \, r \, .
    \end{align}
  \item
    For all $x \in G_-, \alpha \in G_+$ one has:
    \begin{align}
      \label{eq:global_double_projections_r_matrix_ii_1}
      (\id \ox T(\pi_+ \0 R_x)) \, r &= (\Ad_x \ox \id) \, r \\
      \label{eq:global_double_projections_r_matrix_ii_2}
      (T(\pi_- \0 L_\alpha) \ox \id) \, r &= (\id \ox \Ad_{\alpha^{-1}}) \, r \, .
    \end{align}
  \end{compactenum}
\end{lemma}

\begin{proof}
  Both statements can be shown using the $\Ad$-invariance of $r_s = \frac{1}{2}(r + r_{21})$ and the identities
  \[
    (\id \ox T(\pi_- \0 L_d)) \, r = (T(\pi_+ \0 R_d) \ox \id) \, r = 0 \Forall d \in G
  \]
  that follow from $r \in \g_- \ox \g_+$.
\end{proof}

\subsection{Fock and Rosly's Poisson structure}
\label{subsection:fock_and_roslys_poisson_structure}

In \cite{atiyahbott83}, Atiyah and Bott showed that the moduli space $\mathcal M = \Hom(\pi_1(S), G)/G$ of flat $G$-bundles on a compact oriented surface $S$ has a natural symplectic structure if the Lie group $G$ is equipped with a non-degenerate symmetric $\Ad$-invariant bilinear form.
A combinatorial description of this symplectic structure in terms of intersection points of curves has been given by Goldman \cite{goldman1984, goldman1986}.
Fock and Rosly have shown in \cite{fockrosly98} that on a compact oriented surface $S$ with boundary the symplectic structure on $\mathcal M$ can be represented as a graph gauge theory on a ciliated ribbon graph $\Gamma$ embedded into $S$.
In this section we briefly present these results.

Let $G$ be a quasi-triangular Poisson-Lie group with tangent Lie bialgebra $\g$ and $\Gamma$ a ciliated ribbon graph.
Consider the product manifold $\FR := G^{\x E}$ obtained by associating a copy of $G$ to every edge $e \in E$.
Denote by $\pi_e: \FR \to G$ the projection associated with the edge $e \in E$.
Write $I_v$ for the ordered set of edge ends at a vertex $v \in V$ and denote the corresponding edges by $e_i$ for $i \in I_v$.
Define for $v \in V$ and $i, j \in I_v$ the element $s_v^{ij} \in \left\{ 0, \pm 1 \right\}$ by
\[
  s_v^{ij} =
  \begin{cases}
    1 & \text{ if } i > j \\
    0 & \text{ if } i = j \\
    -1 & \text{ if } i < j \, .
  \end{cases}
\]
For $i \in I_v$ let $V_i : \g \to \Gamma(T \FR)$ be the map that associates to the Lie algebra element $X \in \g$ the vector field $V_i(X)$ defined by
\[
  T\pi_e (V_i\, X \; (\gamma)) =
  \begin{cases}
    0 & \text{ if } e \neq e_i  \\
    -TR_{\pi_e(\gamma)} X & \text{ if } e = e_i \text{ and $i$ is incoming at $v$} \\
    TL_{\pi_e(\gamma)} X & \text{ if } e = e_i \text{ and $i$ is outgoing from $v$.}
  \end{cases}
\]
Choose for every vertex $v \in V$ a classical $r$-matrix of the form $r(v) = r_a(v) + r_s$ with \emph{fixed} $\Ad$-invariant symmetric component $r_s \in \g \ox \g$.
Define the bivector
\begin{equation}
  w_{\FR} := \sum_{v \in V} w_v
  \quad \text{with} \quad
  w_v := - \sum_{i,j \in I_v} V_i \ox V_j \, (r_a(v) + s_v^{ij} r_s) \, .
  \label{eq:fockrosly_vertex_bivector}
\end{equation}

To a vertex $v \in V$ we associate the action $\rhd_{v}^{\FR} : G \x \FR \to \FR$ with
\begin{equation}
  \label{eq:fock_rosly_vertex_action}
  \pi_e( h \rhd_v^{\FR} \gamma) =
  \begin{cases}
    h \pi_e(\gamma) h^{-1} & \text{if $e$ is a loop at $v$} \\
    h \pi_e(\gamma) & \text{if $e$ is incoming at $v$} \\
    \pi_e(\gamma) h^{-1} & \text{if $e$ is outgoing from $v$} \\
    \pi_e(\gamma) & \text{otherwise.} \\
  \end{cases}
\end{equation}

Let $G_v$ be the quasi-triangular Poisson-Lie group with the Poisson bivector from \eqref{eq:SklyaninBivector} for the classical $r$-matrix $r(v)$.

\begin{proposition}\cite[Proposition 3]{fockrosly98}
  \label{proposition:fockrosly_poisson_g_space}
  The bivector $w_{\FR}$ \eqref{eq:fockrosly_vertex_bivector} defines a Poisson structure on $\FR$.
  For each $v \in V$ the action $\rhd_v^{\FR}: G_v \x \FR \to \FR$ is Poisson, so that $\FR$ is a Poisson $G_v$-space.\footnote{
    In \cite{fockrosly98}, the authors require the fixed symmetric component $r_s$ to be non-degenerate.
    Proposition \ref{proposition:fockrosly_poisson_g_space} holds without this assumption -- non-degeneracy is needed only later for Proposition \ref{proposition:fockrosly_goldman}.
  }
\end{proposition}

\begin{definition}
  \label{def:fock_rosly_space}
  We call a Poisson manifold of the type $(\FR, w_{\FR})$ a \textbf{Fock-Rosly space} associated to the graph $\Gamma$.
\end{definition}

Consider the compact oriented manifold $S$ that is obtained by gluing annuli to the faces of $\Gamma$ (or, equivalently, by thickening $\Gamma$ as in Definition \ref{def:thickening}) and let
\[
  C^\infty(\FR, \R)^{inv} := \left\{ f \in C^\infty(\FR, \R) \; \middle| \; f(g \rhd^{\FR}_v \gamma) = f(\gamma) \Forall \gamma \in \FR , v \in V, g \in G_v \right\} \, .
\]

\begin{proposition}\cite[Proposition 5]{fockrosly98}
  \label{proposition:fockrosly_goldman}
  The set of invariant functions $C^\infty(\FR, \R)^{inv}$ is a Poisson subalgebra of $C^\infty(\FR, \R)$.
  If the symmetric component $r_s$ is non-degenerate, this Poisson subalgebra is isomorphic to the Poisson algebra of functions on the moduli space $\mathcal M = \Hom(\pi_1 (S), G)/G$ for the non-degenerate symmetric Ad-invariant bilinear form dual to $r_s$.\footnote{
    $G$ is required to be a reductive complex Lie group in \cite{fockrosly98}.
    However, the proof presented there is applicable to any (real) Lie group with a non-degenerate symmetric $\Ad$-invariant bilinear form.
  }
\end{proposition}

To relate Fock-Rosly spaces on different graphs, one uses graph transformations and constructs associated Poisson maps.
In this article we only need a subset of the transformations from \cite{fockrosly98}:
\begin{compactenum}
\item \textbf{ Reversing  an edge.}
  Consider the graph $\Gamma'$ where we replace an edge $h \in E$ with an edge $h'$ with opposite orientation (while keeping the ordering at the vertices intact).
  Let $\FR'$ be the Fock-Rosly space for $\Gamma'$ (with  identical choice of $r$-matrices).
  To the transformation $\Gamma \to \Gamma'$ we associate the map $\phi : \FR \to \FR'$  given by
  \[
    \pi'_e (\phi(\gamma)) =
    \begin{cases}
      \pi_{h}(\gamma)^{-1} & \text{ if } e = h' \\
      \pi_e(\gamma) & \text{ else,}
    \end{cases}
  \]
  where $\pi'_e : \FR' \to G$ is the edge projection map associated to $e \in E'$.
\item \textbf{Gluing two edges along a bivalent vertex.}\footnote{
  This is a special case of the transformation from \cite{fockrosly98} that contracts an edge.}
  Let $h_1, h_2$ be distinct edges incident at the bivalent vertex $v := t(h_1) = s(h_2)$.
  In the transformed graph $\Gamma'$, replace these edges by the edge $h': s(h_1) \to t(h_2)$ and erase $v$.
  The associated map $\phi : \FR \to \FR'$ is defined by:
  \begin{equation}
    \label{eq:fock_rosly_gluing_edges}
    \pi'_e (\phi(\gamma)) =
    \begin{cases}
      \pi_{h_2}(\gamma) \, \pi_{h_1}(\gamma) & \text{ if } e = h' \\
      \pi_e(\gamma) & \text{ else.}
    \end{cases}
  \end{equation}
\item \textbf{Erasing an edge.}
  Choose an edge $h \in E$ and let $\Gamma'$ be the graph obtained by removing $h$ from $\Gamma$.
  Define the associated map $\phi: \FR \to \FR'$ by
  \begin{equation}
    \label{eq:fock_rosly_erasing_edge}
    \pi'_e (\phi(\gamma)) = \pi_e(\gamma) \qquad \Forall e \in E \setminus \left\{ h \right\} \, .
  \end{equation}
\end{compactenum}

Denote the vertex set of the graph $\Gamma'$ obtained by a graph transformation by $V'$.
\begin{proposition}\cite[Proposition 4]{fockrosly98}
  \label{proposition:fockrosly_graph_trafos}
  For all $v \in V'$, the maps associated to these graph transformations are homomorphisms  of Poisson $G_v$-spaces with respect to the actions $\rhd^{\FR}_v : G_v \x \FR \to \FR$ from Equation \eqref{eq:fock_rosly_vertex_action}.
\end{proposition}

There is a natural notion of flatness for Fock and Rosly's graph gauge theory: an element $\gamma \in \FR$ is flat at a face $f$ if the holonomy around $f$ is trivial.
We express this by a more general functor on graph groupoid $\G(\Gamma_D)$ (Definition \ref{def:graph_basics} (ii)) of the thickening $\Gamma_D$ (Definition \ref{def:thickening}).
Consider the set of smooth maps $C^\infty(\FR, G)$ as a groupoid with a single object $*$ so that $\Hom(*,*) = C^\infty(\FR, G)$ and composition is given by pointwise multiplication.

\begin{definition}
  \label{def:fock_rosly_holonomy_functor}
  The functor $\Hol_{FR}: \G(\Gamma_D) \to C^\infty(\FR,G)$ is the unique functor defined by
  \begin{equation}
    \label{eq:fock_rosly_holonomy_functor}
    \Hol_{\FR} ( r(e) ) := \Hol_{\FR} (l(e)) := \pi_{e} \qquad \Hol_{\FR} (f(e)) := \Hol_{\FR} (b(e)) := (\gamma \mapsto 1_G) \, .
  \end{equation}
\end{definition}
This determines the functor $\Hol_{\FR}$ uniquely as $\G(\Gamma_D)$ is freely generated by the edges of $\Gamma_D$.

\subsection{Kitaev models}
\label{subsection:kitaev_models}

We summarize the construction and some important properties of Kitaev models, mostly following the presentation in \cite{meusburger16}.
First we introduce some Hopf-algebraic notions.

\begin{notation}
  Denote the multiplication and unit of $H$ by $m : H \ox H \to H$ and $1_H$, respectively, and the comultiplication, counit and antipode by $\Delta: H \to H \ox H$, $\epsilon : H \to \F$ and $S: H \to H$.
  We use Sweedler notation for the $(n-1)$-fold comultiplication $ \Delta^{(n-1)} (k) = \sum_{(k)} k_{(1)} \ox \dots \ox k_{(n)}$ for $k \in H $.
  The dual Hopf algebra $H^*$ of a finite-dimensional Hopf algebra is equipped with the multiplication $\Delta^* : (H \ox H)^* \cong H^* \ox H^* \to H^*$, unit $\epsilon^*$, comultiplication $m^* : H^* \to H^* \ox H^*$, counit $H^* \to \F, \alpha \mapsto \alpha(1_H)$ and antipode $S^* : H^* \to H^*$.
  By $H^{*cop}$ we denote the Hopf algebra $H^*$ with opposite comultiplication.
\end{notation}

\begin{theorem}
  \label{theorem:drinfeld_double}
  \cite{drinfeld1988}
  Let $H$ be a finite-dimensional Hopf algebra.
  There is a unique quasi-triangular Hopf algebra structure on the vector space $H^* \ox H$ such that $H \cong 1 \ox H$ and $H^{*cop} \cong H^{*cop} \ox 1$ are Hopf subalgebras and the $R$-matrix is given by $R = \id_H$ viewed as an element of $\in H \ox H^* \subseteq (H^* \ox H)^{\ox 2}$.
  Its multiplication is given by
  \[
    (\alpha \ox h) \cdot (\alpha' \ox h') := \sum_{(\alpha')} \sum_{(h)} \alpha'_{(3)} (h_{(1)}) \, \alpha'_{(1)}( S^{-1}(h_{(3)})) \, \alpha \alpha'_{(2)} \ox h_{(2)} h' \, .
  \]
  This Hopf algebra is denoted by $D(H)$ and called the \textbf{Drinfeld double} of $H$.
\end{theorem}

\begin{definition}
  \label{def:heisenberg_double}
  Let $H$ be a Hopf algebra.
  The \textbf{Heisenberg double} $\He(H)$ of $H$ is the associative algebra structure on the vector space $H \ox H^*$ defined by the multiplication
  \[
    (h \ox \alpha) \cdot (h' \ox \alpha') := \sum_{(\alpha)} \sum_{(h')} \alpha_{(1)} (h'_{(2)}) \, h h'_{(1)} \ox \alpha_{(2)} \alpha' \, .
  \]
\end{definition}

Kitaev models are constructed from a doubly ciliated ribbon graph $\Gamma$ and a semi-simple\footnote{
  The authors of \cite{buerschaper2013hierarchy} additionally require $H$ to be a Hopf $*$-algebra.
  However, this is not required to define the structures needed in this article.
}
finite-dimensional Hopf algebra $H$ over $\C$.
To the graph $\Gamma$ one associates the \textbf{extended space} $H^{\ox E}$ by assigning a copy of $H$ to every edge of $\Gamma$.
For $k \in H, \alpha \in H^*$ consider the linear maps $L^k_\pm, T^\alpha_\pm: H \to H$
\begin{equation}
  \label{eq:kitaev_triangle_operators}
\begin{alignedat}{3}
  L^k_+ \, l &= k \cdot l \qquad&\qquad L^k_- \, l &= l \cdot k \\
  T^\alpha_+ \, l &=  (\id_H \ox \alpha) \, \Delta (l) & T^\alpha_- \, l &=  (\alpha \ox \id_H) \, \Delta (l) \, .
\end{alignedat}
\end{equation}
We then define for each edge $e \in E$ the \textbf{triangle operators} $L^k_{e\pm}, T^{\alpha}_{e\pm}: H^{\ox E} \to H^{\ox E}$ for $k \in H, \alpha \in H^*$ by extending $L^k_\pm, T^\alpha_\pm : H \to H$ to linear maps $H^{\ox E} \to H^{\ox E}$ in the obvious way.

For the remainder of this section assume that the graph $\Gamma$ is paired (Definition \ref{def:double_graph_basics} (iv)).

Consider a vertex $v \in V$ with incident edges $e_1 < \dots < e_n$ numbered according to the ordering of their edge ends at $v$ (counterclockwise).
Let $\varepsilon_1, \dots, \varepsilon_n \in \left\{ \pm 1 \right\}$ such that $e_i^{\varepsilon_i}$ is incoming at $v$.
The \textbf{vertex operator} $A^k_v : H^{\ox E} \to H^{\ox E}$ for $k \in H$ is the linear map
\begin{equation}
  \label{eq:quantum_kitaev_vertex_operator}
  A^k_v := \sum_{(k)} L_{e_1 \varepsilon_1}^{S^{\tau_1} (k_{(1)})} \0 \dots \0 L_{e_n \varepsilon_n}^{S^{\tau_n}(k_{(n)})} \qquad \text{with } \tau_i := \frac{1}{2}(1- \varepsilon_i) \, ,
\end{equation}
where we set $S^0 = \id_H$.
Similarly, let $e_1 < \dots < e_n$ be the edges adjacent to a face $f \in F$, numbered according to the ordering of their edge sides at $f$ (clockwise).
Let $\varepsilon_i \in \left\{ \pm 1 \right\}$ such that $e_i^{\varepsilon_i}$ is oriented clockwise for all $i$.
The \textbf{face operator} $B^\alpha_f : H^{\ox E} \to H^{\ox E}$ for $\alpha \in H^*$ is defined as
\begin{equation}
  \label{eq:quantum_kitaev_face_operator}
  B^\alpha_f := \sum_{(\alpha)} T_{e_n \varepsilon_n}^{S^{\tau_n}(\alpha_{(n)})} \0 \dots \0 T_{e_1 \varepsilon_1}^{S^{\tau_1}(\alpha_{(1)})} \qquad \text{with } \tau_i := \frac{1}{2}(1- \varepsilon_i) \, .
\end{equation}

Vertex and face operators are subject to the following commutation relations:

\begin{lemma}[Commutation relations of vertex and face operators] \cite{kitaev03, buerschaper2013hierarchy}
  \label{lemma:quantum_kitaev_commutation_relations}
  \begin{compactenum}
  \item The vertex operators for distinct vertices $v, v' \in V$ commute: $[A^k_v, A^{k'}_{v'}] = 0 \Forall k, k' \in H$.
  \item The face operators for distinct faces $f, f' \in F$ commute: $[B^{\alpha}_f, B^{\alpha'}_{f'}] = 0 \Forall \alpha, \alpha' \in H^*$.
  \item If the pair $(v,f) \in V \x F$ satisfies $v(f) \neq v$ and $f(v) \neq f$ (see Definition \ref{def:double_graph_basics} (i) and (ii)), the associated operators commute: $[A^k_v, B^\alpha_f] = 0 \Forall k \in H, \alpha \in H^*$
  \item For every site $(v,f)$ one obtains an injective algebra homomorphism
    \[
      \tau: D(H) \to \End_\F(H^{\ox E}) \qquad (\alpha \ox k) \mapsto B^\alpha_f \0 A^k_v \, .
    \]
  \end{compactenum}
\end{lemma}

Due to the semi-simplicity of $H$, one can project onto  a subspace of $H^{\otimes E}$ that is invariant under the action of vertex and face operators. 
Denote by $\eta \in H^*, l \in H$ the normalized Haar integrals of $H^*$ and $H$, respectively.
Let $S$ be the oriented  surface obtained by gluing a disk to every face of the graph $\Gamma$.
\begin{proposition}
  \label{proposition:quantum_kitaev_protected_space}
  \cite{kitaev03, buerschaper2013hierarchy}
    The set $\{ A_v^l \, | \, v \in V \} \cup \{ B_f^\eta \, | \, f \in F \}$ consists of commuting projectors that do not depend on the cilia of $\Gamma$.
    The \textbf{protected space}
    \begin{equation}
      \label{eq:quantum_kitaev_protected_space}
      \He_{pr} := \left\{ \gamma \in H^{\ox E} \, \big| \, A^l_v \, \gamma = B^\eta_f \, \gamma = \gamma \Forall v \in V, f \in F \right\} \subseteq H^{\ox E} \, .
    \end{equation}
    of the Kitaev model is a topological invariant:
    It depends only on the homeomorphism class of the oriented surface  $S$.
\end{proposition}

As introduced by Kitaev in \cite{kitaev03}, \textbf{topological excitations} are obtained by choosing a number of disjoint sites $(v_1, f_1) , \dots, (v_n, f_n)$ and imposing invariance under $A_v^l, B_f^\eta$ for all other vertices and faces.
Writing $L := (V \setminus \left\{ v_1, \dots, v_n \right\}) \dot\cup (F \setminus \left\{ f_1, \dots, f_n \right\})$ we obtain the subspace
\begin{equation}
  \label{eq:quantum_kitaev_protected_space_with_excitations}
  \mathcal L_L := \left\{ \gamma \in H^{\ox E} \, \big| \, A^l_v \, \gamma = B^\eta_f \, \gamma = \gamma \Forall v \in V \cap L, f \in F \cap L \right\} \subseteq H^{\ox E} \, .
\end{equation}

By Lemma \ref{lemma:quantum_kitaev_commutation_relations} (iv) there is a $D(H)$-module structure on $\mathcal L_L$ at each of the sites $(v_i, f_i), i= 1, \dots, n$.
As the Hopf algebra $D(H)^{\ox n}$ is semi-simple, we can decompose the $D(H)^{\ox n}$-module $\mathcal L_L$ into irreducible representations.
Denote by $Irr(D(H))$ a set of representatives of equivalence classes of irreducible representations of $D(H)$.
As presented by Balsam and Kirillov in \cite{balsamkirillov12}, one has
\begin{equation}
  \label{eq:quantum_kitaev_excited_space_decomposition}
  \mathcal L_L \cong \bigoplus_{Y_1, \dots, Y_n} (Y_1^* \ox \dots \ox Y_n^*) \ox \mathcal M(Y_1, \dots, Y_n) \, ,
\end{equation}
where the summation is over $Y_i \in Irr(D(H))$ and $Y_i^*$ denotes the dual representation for all $i$.
The vector space $\mathcal M(Y_1, \dots, Y_n)$, on which $D(H)^{\ox n}$ acts trivially, is called a  \textbf{protected space} for the Kitaev model with excitations \cite{kitaev03, balsamkirillov12}.
Let $S'$ be the oriented  surface obtained by gluing annuli to the faces $f_1, \dots, f_n$ of $\Gamma$ and disks to all other faces.

\begin{theorem} \cite[Theorem 6.1]{balsamkirillov12}
  The vector space $\mathcal M(Y_1, \dots, Y_n)$ is isomorphic to the vector space $Z_{TV} (S', Y_1, \dots, Y_n)$ that extended Turaev-Viro TQFT \cite{balsamkirillov2010} for the category $H$-Mod assigns to $S'$, where the boundary component of $S'$ corresponding to $f_i$ is labeled with $Y_i$ for $i = 1, \dots, n$.
\end{theorem}

In particular, the space $\mathcal M(Y_1, \dots, Y_n)$ only depends on the homeomorphism class of the oriented  surface $S'$ and the irreducible representations $Y_1, \dots, Y_n$ of $D(H)$.
Excited states $\gamma \in \mathcal L_L$ can be interpreted as quasi-particles, or anyons, located at the sites $(v_1, f_1), \dots, (v_n, f_n)$ \cite{kitaev03}.

In \cite{meusburger16}, Meusburger characterizes the protected space $\He_{pr}$ by its endomorphism algebra, a subalgebra of the algebra $\End_\C(H^{\ox E})$ of endomorphisms of the extended space.
The latter can be described using the Heisenberg double $\He(H)$, which is known to be isomorphic to $\End_\C(H)$ \cite[Corollary 9.4.3]{montgomery93}:
\begin{lemma}
  \label{lemma:quantum_kitaev_triangle_algebra}
  \cite{meusburger16}
  The map $\phi: \He(H) \to \End_\F(H), k \ox \alpha \mapsto L^k_+ \0 T^\alpha_+$ is an isomorphism of algebras.
  It induces an algebra isomorphism $\Phi: \End_\C(H^{\ox E}) \to \He(H)^{\ox E}$.
\end{lemma}

The representations of $D(H)$ associated to sites of $\Gamma$ (see Lemma \ref{lemma:quantum_kitaev_commutation_relations} (iv)) induce $D(H)$-module algebra structures on $\End_\C(H^{\ox E}) \cong \He(H)^{\ox E}$.
They are given in terms of vertex and face operators, which can be viewed as elements of $\He(H)^{\ox E}$ by Lemma \ref{lemma:quantum_kitaev_triangle_algebra}.

\begin{theorem}[Module algebra structure on $\End_\C(H^{\ox E})$]
  \label{theorem:quantum_kitaev_module_algebra}
  \cite{meusburger16}
  For every site $(v,f)$ the map $\lhd_v : \He(H)^{\ox E} \ox D(H) \to \He(H)^{\ox E}$ given by
  \begin{equation}
    \label{eq:quantum_kitaev_right_action}
    \gamma \lhd_v (\alpha \ox k) = \sum_{(\alpha)} \sum_{(k)} A_v^{S(k_{(2)})} \cdot B_{f}^{S(\alpha_{(1)})} \cdot \gamma \cdot B_{f}^{\alpha_{(2)}} \cdot A_v^{k_{(1)}}
  \end{equation}
  equips $\He(H)^{op \ox E}$ with the structure of a $D(H)$-right module algebra.
\end{theorem}

As the protected space $\He_{pr}$ is a topological invariant, so is the algebra $\End_\C(\He_{pr})$ of its endomorphisms.
Meusburger has shown that it is a subalgebra of the algebra
\[
  \He(H)^{\ox E}_{inv} := \left\{ \gamma \in \He(H)^{\ox E} \, \middle| \, \gamma \lhd_v (\alpha \ox k) = \epsilon(k) \alpha(1_H) \, \gamma \Forall (\alpha \ox k) \in D(H), v \in V \right\} \, ,
\]
which we identify with a subalgebra of $\End_\C(H^{\ox E})$ using Lemma \ref{lemma:quantum_kitaev_triangle_algebra}.
More specifically, for each site $(v, f)$ of $\Gamma$ define the element $G_v := A_v^l \cdot B_{f}^\eta \in \He(H)^{\ox E}$ using the Haar integrals $l \in H, \eta \in H^*$.
Then one has:

\begin{proposition}
  \label{proposition:quantum_kitaev_flatness}
  \cite{meusburger16}
  The map $ Q_{flat} : \He(H)^{\ox E} \to \He(H)^{\ox E}, \gamma \mapsto \gamma \cdot \prod_{v \in V} G_v $ is a projector.
  Its restriction to $\He(H)^{\ox E}_{inv}$ is an algebra homomorphism with image $\End_\C(\He_{pr})$.
\end{proposition}

\section{Poisson analogues of Kitaev models}
\label{section:kitaev_analogues}

In this chapter we define Poisson analogues of Kitaev models by exchanging the data from the representation theory of a Hopf algebra with Poisson-Lie counterparts, that is, certain Poisson $G$-spaces.

Our goals are:
\begin{compactenum}
\item
  to give a Poisson analogue of Kitaev models that has close structural similarities to (quantum) Kitaev models and
\item
  to show that this Poisson analogue is related to the moduli space of flat $G$-bundles for a surface with boundary.
\end{compactenum}

\subsection{Kitaev models with Poisson-geometric data}
\label{subsection:kitaev_models_with_poisson_geometric_data}

To construct an analogue of Kitaev models, we replace the Hopf-algebraic data in the Kitaev model by Poisson-geometric counterparts as outlined  in Table \ref{tab:dictionary}.
We define Poisson analogues  of vertex and face operators.
We also introduce Poisson actions associated with vertices and faces, and a  notion of flatness.

Consider the extended space $H^{\ox E}$ of a Kitaev model for the Hopf algebra $H$.
By Lemma \ref{lemma:quantum_kitaev_triangle_algebra}, the endomorphism algebra of $H^{\ox E}$  is isomorphic to $\He(H)^{\ox E}$ for the Heisenberg double $\He(H)$ of $H$.
A  Poisson analogue should replace the algebra of operators on $H^{\ox E}$ by a suitable Poisson algebra.
Therefore, we associate a copy of a  \emph{ Poisson-geometrical  Heisenberg double} $G_\He$  (Definition \ref{def:classical_heisenberg_double}) to  every edge $e \in E$.
The Poisson algebra of functions $C^\infty(G_\He^{\x E}, \R)$ on the resulting product Poisson manifold $G_\He^{\x E}$ is the analogue of the endomorphism algebra $\He(H)^{\ox E} \cong \End_\C(H^{\ox E})$. 

A (quantum) Heisenberg double admits the factorization $\He(H) \cong H^* \ox H$ (as a vector space).
This has a direct geometrical interpretation in terms of the triangle algebra from \eqref{eq:kitaev_triangle_operators}:
in the thickening $\Gamma_D$ of the graph $\Gamma$ (see Definition \ref{def:thickening}) the triangle operators $L^k_{e \pm}$ for an edge $e$ are parameterized by elements $k \in H$ and are  associated with the edge ends $f(e), b(e)$, whereas the operators $T^\alpha_{e \pm}$ for $\alpha \in H^*$ are associated with the edge sides $r(e), l(e)$.
(See, for instance, \cite[Section 3.1]{buerschaper2013hierarchy}.)

For this reason, we require $G$ to be a global double Poisson-Lie group (Definition \ref{def:global_double_poisson_lie_group}).
  (We briefly discuss the case of a non-global double at the end of Section \ref{subsection:poisson_kitaev_models_fock_rosly_spaces}.)
  Then the Poisson-geometrical Heisenberg double $G_\He$ is diffeomorphic to $G_+ \x G_-$ (as a manifold).
  The Poisson-Lie groups $G_-$ and $(G_+, -w_{G_+})$ are dual to each other, where $G_+$ is equipped with the negative Poisson bivector $-w_{G_+}$.
  They take on the roles of the Hopf algebras $H, H^*$, respectively.
  If $G$ is connected and simply connected, then one has $G = D(G_-)$, the Poisson manifold $G_\He$ is the Heisenberg double $\He(G_-)$, and $(G_+, -w_{G_+}) = G_-^*$ holds.
  Table \ref{tab:dictionary_with_precise_notation} summarizes the correspondence between the Hopf-algebraic data used for a Kitaev model and the Poisson counterparts.

\begin{table}
  \centering
  \begin{tabular}{|p{6cm}|p{6cm}|}
    \hline
    \textbf{Data for a Kitaev model} & \textbf{Poisson counterpart} \\
    \hline
    Hopf algebras $H, H^*$ & Poisson-Lie groups $G_-, (G_+, -w_{G_+})$ \\
    \hline
    Drinfeld double $D(H) \cong H^* \ox H$ & global double $G \cong G_+ \x G_-$ \\
    \hline
    Heisenberg double $\He(H) \cong H^* \ox H$ & Heisenberg double $G_\He \cong G_+ \x G_-$  \\
    \hline
    endomorphism algebra $\He(H)^{\ox E}$ on the extended space $H^{\ox E}$ & Poisson algebra of functions $C^\infty(G_\He^{\x E}, \R)$ \\
    \hline
  \end{tabular}
  \caption{ Data for Hopf-algebraical and Poisson-geometrical Kitaev models }
  \label{tab:dictionary_with_precise_notation}
\end{table}

\begin{definition} Let $G$ be a global double Poisson-Lie group.
  A \textbf{Poisson-Kitaev model} is a pair $(K, \Gamma)$ consisting of
  \begin{itemize}[noitemsep,nolistsep]
    \item a doubly ciliated ribbon graph $\Gamma$ (see Definition \ref{def:doubly_ciliated_ribbon_graph}) and
    \item the product Poisson manifold $K := G_\He^{\x E}$, where a copy of the Heisenberg double $G_\He$ is assigned to every edge $e \in E$ of $\Gamma$.
  \end{itemize}
\end{definition}

We will now define counterparts of the triangle operators and the vertex and face operators as functions on the Poisson manifold $K$ with values in $G$. 
The triangle operators $L^k_\pm$ and $T^\alpha_\pm$ are replaced by the projections $\pi_\pm: G\to G_\pm \subset G$ from Lemma \ref{lemma:projectionLocallyPoisson}.
For this we introduce a holonomy functor on the graph groupoid $\G(\Gamma_D)$ (Definition \ref{def:graph_basics} (ii)) of the thickened graph $\Gamma_D$ (Definition \ref{def:thickening}).
It is defined similarly to the functor $\Hol_{\FR}$ from \eqref{eq:fock_rosly_holonomy_functor}.
Consider the set of smooth maps $C^\infty(K, G)$ as a groupoid with a single object with composition given by pointwise multiplication.

\begin{definition}
  The \textbf{holonomy functor} $\Hol : \G(\Gamma_D) \to C^\infty(K, G)$ is the functor defined by
  \begin{equation}
    \label{eq:kitaev_holonomy_functor}
    \begin{alignedat}{2}
      &\Hol (r(e)) := \pi_+ \0 \pi_e && \qquad \Hol (f(e)) := \pi_- \0 \pi_e \\
      &\Hol (l(e)) := \eta \0 \pi_+ \0 \eta \0 \pi_e && \qquad \Hol (b(e)) := \eta \0 \pi_- \0 \eta \0 \pi_e
    \end{alignedat}
  \end{equation}
  for all edges $e \in E$ , where $\pi_e : K \to G$ stands for the projection on the component associated with the edge $e$, and $\eta : G \to G$ is the inversion map.
\end{definition}

\begin{figure}[h]
\centering
\begin{tikzpicture}[vertex/.style={circle, fill=black, inner sep=0pt, minimum size=2mm}, plain/.style={draw=none, fill=none}, scale=0.8]
  \draw [-latex] (-6,0) -- (-3,0) node[midway, above] {$d$};

  \coordinate (n3) at (0.5, -0.5);
  \coordinate (n1) at (0.5, 0.5);

  \draw [-latex] (n3) -- ($(n3) + (3,0)$) node[midway, below] {$\pi_+(d)$};
  \draw [-latex] (n1) -- ($(n1) + (3,0)$) node[midway, above] {$\pi_+(d^{-1})^{-1}$};
  \draw [-latex] ($(n3) + (3,0)$) -- ($(n1) + (3,0)$) node[midway, right] {$\pi_-(d)$};
  \draw [-latex] (n3) -- (n1) node[midway, left] {$\pi_-(d^{-1})^{-1}$};
\end{tikzpicture}
\caption{An edge of $\Gamma$ labelled with $d \in G$ and the corresponding holonomies on $\Gamma_D$}
\label{fig:kitaev_holonomy}
\end{figure}

Because of the factorization $d = \pi_-(d) \pi_+(d)$, one has $\Hol (f(e) \circ r(e)) = \pi_e$.
As $d = (d^{-1})^{-1} = (\pi_-(d^{-1}) \pi_+(d^{-1}))^{-1}$, the equation $\Hol(l(e) \0 b(e)) = \pi_e$ also holds.
In fact, $\pi_e$ can be computed from any holonomy along an edge end and adjacent edge side of $e$.

The functor $\Hol$ is an analogue of the Hopf algebra valued holonomy functor in \cite{meusburger16}.
The latter is a functor $\G(\Gamma_D) \to \Hom_\C(D(H)^*, D(H)^{* \ox E})$, where $D(H)^*$ is the dual of the Drinfeld double $D(H)$ (see Theorem \ref{theorem:drinfeld_double}) and $\Hom_\C(D(H)^*, D(H)^{* \ox E})$ is a $\C$-linear category with a single object and the convolution product as composition of morphisms.
One has $\He(H) \cong H \ox H^*$ as a vector space (see Definition \ref{def:heisenberg_double}  ) which in turn is isomorphic to $D(H)^*$.
Therefore, the vector space $D(H)^{* \ox E} \cong \He(H)^{\ox E}$ corresponds to the manifold $K \cong  G_\He ^{\x E}$ of a Poisson-Kitaev model.
The single object category $C^\infty(K, G)$ corresponds to the $\C$-linear category $\Hom_\C( D(H)^*, D(H)^{* \ox E})$.

The holonomy functor $\G(\Gamma_D) \to \Hom_\C(D(H)^*, D(H)^{* \ox E})$ generalizes Kitaev's ribbon operators \cite{kitaev03, bombin2008family, buerschaper2013hierarchy}.
It assigns to the edge ends of an edge $e$ the triangle operators $L^k_{e \pm}$ from Equation \eqref{eq:kitaev_triangle_operators}, indexed by $k \in H$.
To its edge sides it assigns the operators $T^\alpha_{e \pm}$ with $\alpha \in H^*$.
In analogy to the quantum case, the holonomy functor of a Poisson-Kitaev model assigns elements of $G_+$ to paths along edge sides and elements of the dual Poisson-Lie group $G_-$ to paths along edge ends.
The edge side $r(e)$ connects vertices of $\Gamma$, whereas the edge end $f(e)$ connects the faces left and right of $e$ or, equivalently, the corresponding vertices of the dual graph $\Gamma^*$.
We can thus view $r(e)$ as an edge of $\Gamma$, that is decorated with an element of $G_+$, and $f(e)$ as the corresponding edge of the dual graph, which is labeled with an element of the dual Poisson-Lie group $G_-$.
The graphs $\Gamma$ and $\Gamma^*$ are combined in the thickened graph $\Gamma_D$.
The interaction of an edge and its dual edge is described by the Heisenberg double.

\begin{remark}[Edge reversal]
  \label{remark:edge_reversal}
  While a Poisson-Kitaev model $(K, \Gamma)$ and the associated holonomy functor $\Hol$ depend on the orientation of the edges of $\Gamma$, this dependence is only formal.
  Consider the Poisson-Kitaev model $(K', \Gamma')$ where the edge $e$ has been replaced by the reversed edge $e'$.
  Define the functor $I_e : \G(\Gamma_D) \to \G(\Gamma_D')$ by
  \begin{align*}
    I_e (r(e)) := l(e')^{-1} \quad I_e (l(e)) := r(e')^{-1} \quad I_e(f(e)) := b(e')^{-1} \quad I_e(b(e)) := f(e')^{-1}
  \end{align*}
  and $I_e (p) = p$ for all other edge sides and edge ends.
  Let $\eta_e^{*} : C^\infty(K, G) \to C^\infty(K', G), \phi \mapsto \phi \0 \eta_e$, where $\eta_e: K' \to K$ is the involution that inverts the group element at the edge $e'$ and leaves the elements at all other edges invariant.
  It satisfies $\Hol' \0 I_e = \eta_e^* \0 \Hol$, where $\Hol'$ is the holonomy functor for $(K', \Gamma')$.
  By Theorem \ref{theorem:heisenberg_double_inversion_poisson_actions}, the map $\eta_e : K' \to K$ is Poisson and intertwines the two different actions $\rhd, \rhd': G \x  G_\He  \to  G_\He $ on the copy of $ G_\He $ at $e'$.
\end{remark}

The functor $I_e$ describes a rotation of the edge $e$ by 180 degrees.
Hence, reversing an edge in $\Gamma$ is the same as rotating an edge by 180\textdegree{} in $\Gamma_D$ and corresponds to inverting the associated element of $G$.
As edge reversal is compatible with the Poisson structures of $K$ and $K'$ and with the Poisson actions $\rhd, \rhd'$, the orientation of edges is irrelevant.

All constructions introduced in the following are compatible with edge reversal.
The compatibility will not be mentioned explicitly unless it is not obvious.

Paths around vertices and faces play an important role in the following.
We define such paths and the associated holonomies.

\begin{definition}~
  \label{def:vertex_face_paths_and_holonomies}
  \begin{compactenum}
    \item
      Consider a vertex $v \in V$ and let $i_1 < \dots < i_n$ be the linearly ordered edge ends incident at $v$.
      For $k=0, \dots, n$ we define the \textbf{partial vertex path}
      \begin{equation}
        \label{eq:partial_vertex_path}
        p_k(v) := i_1^{\varepsilon_1} \0 \dots \0 i_k^{\varepsilon_k} \, ,
      \end{equation}
      where $\varepsilon_j = 1$ if $i_j = f(e)$ for an edge $e \in E$ and $\varepsilon_j = -1$ if $i_j = b(e)$.
      The path $p_0(v)$ is the empty path $s(i_1^{\varepsilon_1}) \to s(i_1^{\varepsilon_1})$.
      We call $p(v) := p_n(v)$ the \textbf{vertex path} around $v$.
    \item
      For a face $f \in F$ with linearly ordered edge sides $i_1 <  \dots < i_n$ define for $k=0,\dots,n$ the \textbf{partial face path}
      \begin{equation}
        \label{eq:partial_face_path}
        p_k(f) := i_k^{\varepsilon_k} \circ \dots \circ i_1^{\varepsilon_1} \, ,
      \end{equation}
      where $\varepsilon_j = 1$ if $i_j = r(e)$ for an edge $e \in E$ and $\varepsilon_j = -1$ for $i_j = l(e)$.
      The path $p_0(f)$ is the empty path $s(i_1^{\varepsilon_1}) \to s(i_1^{\varepsilon_1})$.
      The path $p(f) := p_n(f)$ is called the \textbf{face path} of $f$ in the following.
    \item
      To a vertex $v$ we associate the \textbf{vertex holonomy} $\Hol^v := \Hol (p(v))$ and to a face $f$ the \textbf{ face holonomy} $\Hol^f := \Hol(p(f))$.
      For a site $(v,f)$ we call the product $\Hol^{(v,f)} := \Hol^v \cdot \Hol^f = \Hol(p(v) \0 p(f))$ the associated \textbf{site holonomy}.
  \end{compactenum}
\end{definition}

These paths are illustrated in Figure \ref{fig:vertex_face_paths}.
Note that the face path $p(f)$ of a face $f$ coincides with the selected face path from Definition \ref{def:doubly_ciliated_ribbon_graph}.

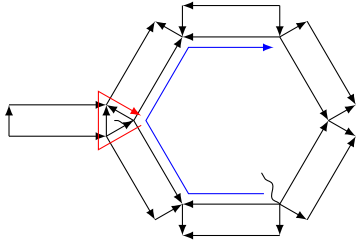
\begin{figure}
\centering
\begin{tikzpicture}[vertex/.style={circle, fill=black, inner sep=0pt, minimum size=2mm}, plain/.style={draw=none, fill=none}, scale=0.8]
  \begin{scope}[scale=0.8]
    
  \coordinate (n) at (2,0);

  \coordinate (o) at (60:2cm);
  \draw [-latex] (o) -- (n);
  \draw [-latex] (o) -- ($(o)!0.65cm!90:(n)$);
  \draw [-latex] (n) -- ($(n) + (o)!0.65cm!90:(n) - (o)$);
  \draw [-latex] ($(o)!0.65cm!90:(n)$) -- ($(n) + (o)!0.65cm!90:(n) - (o)$);

  \coordinate (p) at (120:2cm);
  \draw [-latex] (o) -- (p);
  \draw [-latex] ($(o) + (0, 0.65)$) -- (o);
  \draw [-latex] ($(o) + (0, 0.65)$) -- ($(p) + (0, 0.65)$);
  \draw [-latex] ($(p) + (0, 0.65)$) -- (p);

  \coordinate (q) at (180:2cm);
  \draw [-latex] (q) -- (p);
  \draw [-latex] (q) -- ($(q) ! 0.65cm ! 90:(p)$);
  \draw [-latex] (p) -- ($(p) + (q)!0.65cm!90:(p) - (q)$);
  \draw [-latex] ($(q)!0.65cm!90:(p)$) -- ($(p) + (q)!0.65cm!90:(p) - (q)$);

  \coordinate (r) at (240:2cm);
  \draw [-latex] (q) -- (r);
  \draw [-latex] ($(q)!0.65cm!-90:(r)$) -- (q);
  \draw [-latex] ($(r) + (q)!0.65cm!-90:(r) - (q)$) -- (r);
  \draw [-latex] ($(q)!0.65cm!-90:(r)$) -- ($(r) + (q)!0.65cm!-90:(r) - (q)$);

  \coordinate (s) at (300:2cm);
  \draw [-latex] (s) -- (r);
  \draw [-latex] (s) -- ($(s)!0.65cm!90:(r)$);
  \draw [-latex] (r) -- ($(r) + (s)!0.65cm!90:(r) - (s)$);
  \draw [-latex] ($(s)!0.65cm!90:(r)$) -- ($(r) + (s)!0.65cm!90:(r) - (s)$);

  \draw [-latex] (s) -- (n);
  \draw [-latex] (s) -- ($(s)!0.65cm!-90:(n)$);
  \draw [-latex] (n) -- ($(n) + (s)!0.65cm!-90:(n) - (s)$);
  \draw [-latex] ($(s)!0.65cm!-90:(n)$) -- ($(n) + (s)!0.65cm!-90:(n) - (s)$);

  \coordinate (q1) at ($(q)!0.65cm!90:(p)$);
  \coordinate (q2) at ($(q)!0.65cm!-90:(r)$);
  \coordinate (t) at ($(q1)!2cm!-90:(q2)$);
  \draw [-latex] (t) -- (q1);
  \draw [-latex] (q2) -- (q1);
  \draw [-latex] ($(t) + (q2) - (q1)$) -- (t);
  \draw [-latex] ($(t) + (q2) - (q1)$) -- (q2);

  \draw [decorate,decoration={snake, amplitude=0.5mm}] (q) -- ($(q)-(0.4,0)$);
  \draw [decorate,decoration={snake, amplitude=0.5mm}] (s) -- ($(s)-(300:0.75)$);

  \draw [-latex, color=blue] ($(300:1.75cm) - (0.2cm,0)$) -- (240:1.75cm) -- (180:1.75cm) -- (120:1.75cm) -- (60:1.75cm);

  \coordinate (q123) at ($0.333*(q) + 0.333*(q1) + 0.333*(q2)$);
  \coordinate (q_outer) at ($ (q123)!0.7cm!(q)$);
  \coordinate (q1_outer) at ($ (q123)!0.7cm!(q1)$);
  \coordinate (q2_outer) at ($ (q123)!0.7cm!(q2)$);
  \draw [-latex, color=red] ($(q_outer)!0.2cm!(q2_outer)$) -- (q2_outer) -- (q1_outer) -- ($(q_outer)!0.2cm!(q1_outer)$);

  \end{scope}
\end{tikzpicture}
\caption{A vertex path (red) and a partial face path $p_4(f)$ (blue)}
\label{fig:vertex_face_paths}
\end{figure}

Under the holonomy functor from \cite{meusburger16} the vertex and face operators $A^k_v, B^\alpha_f$ from Equations \eqref{eq:quantum_kitaev_vertex_operator} and \eqref{eq:quantum_kitaev_face_operator} correspond to paths around the vertex $v$ and the face $f$, respectively.
This allows us to define Poisson counterparts by applying our holonomy functor \eqref{eq:kitaev_holonomy_functor} instead.

\begin{definition}
  \label{def:vertex_face_operators}
    A function $a \in C^\infty(K, \R)$ of the form $a = g \0 \Hol^v$ for some $v \in V$ and $g \in C^\infty(G, \R)$ is called a \textbf{vertex operator}.
    Likewise, a function of the form $b = g \0 \Hol^f$ for an $f \in F$ is called a \textbf{face operator}.
\end{definition}

Note that both vertex and face paths turn clockwise around the associated vertex or face.
Vertex and face operators of quantum Kitaev models (Equations \eqref{eq:quantum_kitaev_vertex_operator} and \eqref{eq:quantum_kitaev_face_operator}) are also defined in a clockwise order.

We define the notion of flatness of a vertex or face by requiring the associated holonomies to be trivial.

\begin{definition}
  \label{def:flatness}
  For $\gamma \in K$ we say that $\gamma \in K$ is \textbf{flat at} a vertex $v \in V$ (face $f \in F$) if $\Hol^v (\gamma) = 1_G$ ($\Hol^f (\gamma) = 1_G$).
  Denote the subsets of elements flat at $v$ (respectively $f$) by
  \[
    K_v := \left\{ \gamma \in K \, \middle| \, \gamma \text{ flat at } v \right\} \qquad K_f := \left\{ \gamma \in K \, \middle| \, \gamma \text{ flat at } f \right\}
  \]
  and the set of flat elements on a subset $L \subseteq V \dot \cup F$ by
  \begin{equation}
    \label{eq:flat_elements}
    K_L := \bigcap_{l \in L} K_l \, .
  \end{equation}
  An element $\gamma \in K_L$ for $L \neq \emptyset$ is called \textbf{flat at} $L$.
\end{definition}

Note that flatness at a vertex or face only depends on the respective cyclic ordering because vertex and face paths for different linear orderings are related by cyclic permutations.

Definition \ref{def:flatness} corresponds to the notion of flatness for Kitaev models introduced in \cite{meusburger16}.
The map $Q_{flat}$ from Proposition \ref{proposition:quantum_kitaev_flatness} projects onto the flat elements of $\He(H)^{\ox E}$.
The corresponding subset of $K$ is $K_L$ for $L = V \dot \cup F$.

In this article we relate Poisson-Kitaev models to moduli spaces $\Hom(\pi_1(S), G)/G$ of flat $G$-bundles for compact oriented  surfaces $S$ with boundary.
For this we choose a number of sites $(v_1, f_1), \dots, (v_n, f_n)$ of the graph $\Gamma$ and construct the associated oriented  surface $S$ by gluing annuli to $f_1, \dots, f_n$ and disks to all other faces.
This corresponds to the subset $K_L$ of elements that are flat at all vertices and faces except $v_1, \dots, v_n$ and $f_1, \dots, f_n$.
We will see at the end of Section \ref{subsection:kitaev_models_with_poisson_geometric_data}  that the vertices and faces at which flatness is not imposed correspond to  excitations in Kitaev models.
In this article we only consider the case $n \geq 1$ or, equivalently, require that $S$ has at least one boundary component.

We now introduce a Poisson counterpart of the $D(H)$-right module structure $\lhd_v$ from Equation \eqref{eq:quantum_kitaev_right_action}.
In a quantum  Kitaev model, the vertex and face operators for a site $(v,f)$ give rise to, respectively, an action of the Hopf algebra $H$ and its dual $H^*$ on $\He(H)^{\ox E}$ that combine into an action of the Drinfeld double $D(H)$.
We associate to every vertex a Poisson $G_+$-action and to every face a Poisson $G_-$-action that take the role of, respectively, the $H$- and $H^*$-module algebra structures on $\He(H)^{\ox E}$.
We will show later that the actions for a vertex and adjacent face combine into a Poisson action of the double Poisson-Lie group $G$.

Consider a vertex $v$ with incident edge ends $i_1 < \dots < i_m$.
Set $c_k(\gamma) := \Hol (p_{k-1}(v)) (\gamma)$, where $p_k(v)$ is the $k$-th partial vertex path at $v$ from Equation \eqref{eq:partial_vertex_path}.
For $k = 1,\dots, m$ and $\alpha \in G_+$ define the maps $\phi_k (\alpha) : K \to K$ by
\begin{equation}
  \label{eq:vertex_action_explicit}
  (\pi_e \0 \phi_k (\alpha)) \, (\gamma) =
  \begin{cases}
    \pi_+(\alpha \, c_k(\gamma)) \, \pi_e(\gamma) & \text{ if } i_k = f(e) \\
    \pi_e(\gamma) \, \pi_+(\alpha \, c_k(\gamma))^{-1} & \text{ if } i_k = b(e) \\
    \pi_e(\gamma) & \text{ otherwise,  }
  \end{cases}
\end{equation}
where $\pi_e : K \to G_\He$ is the projection onto the copy of $G_\He$ associated with the edge $e$.
Similarly, for a face $f$ with adjacent edge sides $i'_1 < \dots < i'_n$ set $d_k(\gamma) := \Hol(p_{k-1}(f))(\gamma)$ and define for $x \in G_-$ the maps $\psi_k (x) : K \to K$:
\begin{equation}
  \label{eq:face_action_explicit}
  (\pi_e \0 \psi_k (x)) \, (\gamma) =
  \begin{cases}
    \pi_e(\gamma) \, \pi_- (d_k(\gamma) \, x^{-1}) & \text{ if } i_k' = r(e) \\
    \pi_- (d_k(\gamma) \, x^{-1})^{-1} \, \pi_e(\gamma)  & \text{ if } i_k' = l(e) \\
    \pi_e(\gamma) & \text{ otherwise.}
  \end{cases}
\end{equation}

\begin{definition}~
  \label{def:vertex_face_actions}
  \begin{compactenum}
  \item
    The \textbf{vertex action} associated with the vertex $v$ is the map
    \begin{equation}
    \label{eq:vertex_action}
      \rhd_v : G_+ \x K \to K \qquad \alpha \rhd_v \gamma = (\phi_1 (\alpha) \0 \dots \0 \phi_m (\alpha)) \, (\gamma) \, .
    \end{equation}
  \item
    The \textbf{face action} for the face $f$ is given by
    \begin{equation}
    \label{eq:face_action}
      \rhd_f : G_- \x K \to K \qquad x \rhd_f \gamma = (\psi_1 (x) \0 \dots \0 \psi_n (x)) \, (\gamma) \, .
    \end{equation}
  \end{compactenum}
\end{definition}

\begin{remark}[Relation to dressing actions]
  If the incident edges $e_1 < \dots < e_m$ at a vertex $v$ are all incoming, the action $\rhd_v$ can be written more simply as
  \begin{equation}
    \label{eq:kitaev_vertex_action_easy}
    \pi_e (\alpha \rhd_v \gamma) =
    \begin{cases}
      \pi_+(\alpha c_k(\gamma)) \, \pi_e(\gamma) & \text{ if } e = e_k \\
      \pi_e(\gamma) & \text{ else.}
    \end{cases}
  \end{equation}
  Likewise, for a face $f$ where all adjacent edges $e_1 < \dots < e_n$ are ordered clockwise one has:
  \begin{equation}
    \label{eq:kitaev_face_action_easy}
    \pi_e (x \rhd_f \gamma) =
    \begin{cases}
      \pi_e(\gamma) \, \pi_-(d_k(\gamma) x^{-1}) & \text{ if } e = e_k \\
      \pi_e(\gamma) & \text{ else.}
    \end{cases}
  \end{equation}
  The map $G_- \x G_+ \to G_-, (g_- , g_+) \mapsto \pi_-(g_+^{-1} g_-)$ is the right dressing action of the Poisson-Lie group $(G_+, -w_{G_+})$ with negative Poisson structure on $G_-$ \cite[Theorem 3.14]{luweinstein1990}.
  Therefore, the factor $\pi_-(d_k(\gamma) x^{-1})$ in the first line of \eqref{eq:kitaev_face_action_easy} is given by the dressing action of the partial face holonomy $d_k(\gamma)$ on $x^{-1}$.
  Likewise, Equation \eqref{eq:kitaev_vertex_action_easy} involves the dressing action of $G_-$ on $G_+$.
\end{remark}

\begin{lemma}
  The smooth maps $\rhd_v : G_+ \x K \to K$ and $\rhd_f: G_- \x K \to K$ for $v \in V$ and $f \in F$ are group actions.
\end{lemma}

\begin{proof}
  Direct computation using the computation rules for $\pi_\pm : G \to G_\pm$ from Lemma \ref{lemma:global_double_projections_properties}.
\end{proof}

That $\rhd_v, \rhd_f$ are also Poisson maps will be shown in Proposition \ref{proposition:vertex_face_actions_are_poisson}.

We now define the functions invariant under vertex and face actions.
These are analogues of the elements of $\He(H)^{\ox E}$ that are invariant under the $D(H)$-right module structure from Equation \eqref{eq:quantum_kitaev_right_action}.

\begin{definition}[Invariant functions and flatness]~
  \label{def:functions_on_flat_elements}
  \begin{compactenum}
  \item
    Let $v \in V, f \in F$ and $L \subset V \dot \cup F$.
    We set
    \begin{align*}
      C^\infty(K, \R)^v_L &:= \left\{ g \in C^\infty(K, \R) \, \middle| \, g (\alpha \rhd_v \gamma) = g(\gamma) \Forall \alpha \in G_+, \gamma \in K_L \right\} \\
      C^\infty(K, \R)^f_L &:= \left\{ g \in C^\infty(K, \R) \, \middle| \, g (x \rhd_f \gamma) = g(\gamma) \Forall x \in G_-, \gamma \in K_L \right\} \\
      C^\infty(K, \R)^{inv}_L &:= (\bigcap_{v \in V} C^\infty(K, \R)_L^v) \cap (\bigcap_{f \in F} C^\infty(K, \R)_L^f) \, ,
    \end{align*}
    where $K_L$ is the set of elements flat at $L$ from \eqref{eq:flat_elements}.
    For $L= \emptyset$ we simply write, respectively, $C^\infty(K, \R)^v, C^\infty(K, \R)^f$ and $C^\infty(K, \R)^{inv}$ for $C^\infty(K, \R)^v_L, C^\infty(K, \R)^f_L$ and $C^\infty(K, \R)^{inv}_L$.
  \item
    Denote by
    \[
      \mathcal A(\Gamma, L) := C^\infty(K, \R)^{inv}_L / \sim
    \]
    the set of equivalence classes with respect to the relation $f \sim f' \Longleftrightarrow f|_{K_L} = f'|_{K_L}$.
  \end{compactenum}
\end{definition}

\begin{remark}
  \label{remark:excitations}
  The set $\mathcal A(\Gamma, L)$ is related to Kitaev models with excitations.
  Excited states in a Kitaev model are described by the subspace $\mathcal L_L$ of the extended space $H^{\ox E}$ from Equation \eqref{eq:quantum_kitaev_protected_space_with_excitations}, where $L \subseteq V \dot \cup F$ contains all vertices and faces except for a number of selected sites $(v_1, f_1), \dots, (v_n, f_n)$ that correspond to the excitations.
  According to Meusburger \cite[following Theorem 8.3]{meusburger16}, the endomorphism algebra of $\mathcal L_L$ can be obtained from the endomorphism algebra $\He(H)^{\ox E}$ of the extended space: one has to consider the invariants under the $D(H)$-module algebra structure $\lhd_v$ from \eqref{eq:quantum_kitaev_right_action} for all vertices $v \neq v_1, \dots, v_n$ and project these invariants onto the elements of $\He(H)^{\ox E}$ flat at $L$ using a projector similar to $Q_{flat}$ from Proposition \ref{proposition:quantum_kitaev_flatness}.
  
  The Poisson algebra $C^\infty(K, \R)$ corresponds to the endomorphism algebra $\He(H)^{\ox E}$ and vertex and face actions together are analogues of the $D(H)$-module algebra structures.
  One can thus define an analogue of the endomorphism algebra of $\mathcal L_L$: it is obtained from the set $\bigcap_{l \in L} C^\infty(K, \R)^l_L$ of invariants under the vertex actions $\rhd_v$ for $v \neq v_1, \dots, v_n$ and face actions $\rhd_f$ for  $f \neq f_1, \dots, f_n$.
  The projection onto elements flat at $L$ is implemented by identifying these invariant functions if they coincide on the subspace $K_L \subseteq K$.

  In the definition of the set $\mathcal A(\Gamma, L)$ we instead impose invariance under \emph{all} vertex and face actions before taking the quotient, so that $\mathcal A(\Gamma, L)$ corresponds to a subalgebra of $\End_\C(\mathcal L_L)$.
 It should hence form a Poisson algebra.
We will show this in Proposition \ref{proposition:poisson_subalgebra_reduce_to_flat_subspace} and prove in Theorem \ref{theorem:kitaev_moduli_space} that it is isomorphic to the Poisson algebra of functions on the moduli space of flat $G$-bundles for the oriented  surface obtained by gluing disks to the faces in $L$ and annuli to all other faces.
Elements of $\mathcal A(\Gamma, L)$ appear to correspond to Kitaev's \emph{topological operators} \cite{kitaev03} on $\mathcal L_L$. 
\end{remark}

\subsection{Graph transformations}
\label{subsection:graph_trafos_kitaev}

In this section we present certain transformations between ribbon graphs and associate to them Poisson maps that relate the Poisson-Kitaev model $(K, \Gamma)$ with the model $(K', \Gamma')$ on the transformed graph $\Gamma'$.
These transformations do not change the homeomorphism class of the oriented  surface with boundary associated to $\Gamma$.
The associated Poisson maps induce bijections between the sets $\mathcal A(\Gamma, L)$ and $\mathcal A(\Gamma', L')$.

First we relate graphs that only differ by a cyclic permutation of the linear ordering at a vertex or face.
It turns out that the sets of invariant functions $C^\infty(K, \R)^v_L$ with respect to different orderings at the vertex $v$ coincide if $v \in L$, and an analogous statement holds for faces.

Consider a vertex $v \in V$ with ordered edge ends $i_1 < \dots < i_m$.
Denote by $\rhd_v : G_+ \x K \to K$ the vertex action at $v$ from Equation \eqref{eq:vertex_action} and denote by $\rhd_v' : G_+ \x K \to K$ the vertex action for the cyclically permuted linear ordering $i_2 < \dots < i_m < i_1$.
Similarly, consider the face action $\rhd_f: G_- \x K \to K$ from \eqref{eq:face_action} for a face with adjacent edge sides $i'_1 < \dots < i'_n$ and the action $\rhd_f' : G_- \x K \to K$ for the ordering $i'_2 < \dots < i'_n < i'_1$.

\begin{lemma}[Vertex and face actions for cyclically permuted orderings]~
  \label{lemma:kitaev_trafo_moving_cilium}
  \begin{compactenum}
    \item
      Let $p_1(v)$ be the partial vertex path for the ordering $i_1 < \dots < i_m$.
      One has:
      \[
        \alpha \rhd'_v \gamma = \pi_+(\alpha \, \Hol(p_1(v)) (\gamma)^{-1}) \rhd_v \gamma \qquad \Forall \gamma \in K_v, \alpha \in G_+ \, .
      \]
    \item
      Similarly, let $p_1(f)$ be the partial face path for the ordering $i'_1 < \dots < i'_n$.
      We obtain:
      \[
        x \rhd'_f d = \pi_-( (x \Hol(p_1(f))(\gamma))^{-1} )^{-1} \rhd_f d \qquad \Forall \gamma \in K_f, x \in G_- \, .
      \]
    \item
      For $v\in V \cap L$ or $f\in F \cap L$ the sets $C^\infty(K,R)^v_L$ or $C^\infty(K,R)^f_L$ depend only on the ribbon graph structure of $\Gamma$, but not on the choice of cilia for $v$ or $f$. 
  \end{compactenum}
\end{lemma}

\begin{proof}
  One can show (i) by direct computation using Lemma \ref{lemma:global_double_projections_properties} and the fact that
  \[
    \Hol(p_1(v)) (\gamma)^{-1} = \Hol(p'_{m-1}(v))(\gamma) \qquad \Forall \gamma \in K_v \, ,
  \]
  where $p'_{m-1}(v)$ is the partial vertex path from \eqref{eq:partial_vertex_path} in the transformed ordering.
  The proof for Statement (ii) is analogous and (iii) is a direct consequence of (i) and (ii).
\end{proof}

In the following, denote by $E', V', F'$ the sets of edges, vertices and faces, respectively, of the graph $\Gamma'$ obtained from a graph transformation.
Let $K' =   G_\He ^{\x E'}$ be the associated product Poisson manifold and $\Hol'$ the holonomy functor for the Poisson-Kitaev model $(K', \Gamma')$.

Next, we transform $\Gamma$ by gluing two distinct edges $e_1, e_2$ along a bivalent vertex $v_m$.
We require that there are no faces $f \in F$ with $v(f) = v_m$.
Orient $e_1, e_2$ so that $t(e_1) = v_m = s(e_2)$ and 
suppose that the ordering at $v_m$ is given by $f(e_1) < b(e_2)$, as shown in Figure \ref{fig:kitaev_trafo_glue_vertex}.

In the transformed graph $\Gamma'$ replace the edges $e_1, e_2$ by an edge $e'$ with $s(e') = s(e_1)$ and $t(e') = t(e_2)$.
The ordering of the transformed faces left and right of $e'$ is obtained by replacing the consecutive edge sides $r(e_1) < r(e_2)$ by $r(e')$ and $l(e_2) < l(e_1)$ by $l(e')$.

We define the map $\psi_{v_m}: K \to K'$ associated to this graph transformation by assigning to $e'$ the holonomy along $f(e_2) \0 r(e_2) \0 r(e_1)$:
\begin{equation}
  \label{eq:map_gluing_vertex}
  (\pi_e \0 \psi_{v_m}) (\gamma) :=
  \begin{cases}
    \pi_{e_2}(\gamma) \, \pi_+ (\pi_{e_1}(\gamma)) & \text{for } e = e'\\
    \pi_{e}(\gamma) & \text{otherwise.}
  \end{cases}
\end{equation}

\begin{figure}[h]
\centering
\begin{tikzpicture}[vertex/.style={circle, fill=black, inner sep=0pt, minimum size=2mm}, plain/.style={draw=none, fill=none}, scale=0.8]
  \coordinate (m) at (0,0);
  \draw [latex-] (m) -- ($(m) + (4,0)$);
  \draw [latex-] (m) -- ($(m) + (0,1)$);
  \draw [latex-] ($(m) + (0,1)$) -- ($(m) + (4,1)$);
  \draw [latex-] ($(m) + (4,0)$) -- ($(m) + (4,1)$);
  \node [plain] at ($(m) + (2,0.5)$) {$e'$};

  \draw [latex-] ($(m) + (5,0.5)$) -- ($(m) + (6,0.5)$);

  \coordinate (n) at ($(m) + (7,0)$);
  \draw [latex-] (n) -- ($(n) + (2,0)$);
  \draw [latex-] (n) -- ($(n) + (0,1)$);
  \draw [latex-] ($(n) + (0,1)$) -- ($(n) + (2,1)$);
  \draw [latex-] ($(n) + (2,0)$) to [bend left=45] ($(n) + (2,1)$);
  \node [plain] at ($(n) + (1,0.5)$) {$e_2$};

  \coordinate (o) at ($(n) + (2,0)$);
  \draw [latex-] (o) -- ($(o) + (2,0)$);
  \draw [latex-] (o) to [bend right=45] ($(o) + (0,1)$);
  \draw [latex-] ($(o) + (0,1)$) -- ($(o) + (2,1)$);
  \draw [latex-] ($(o) + (2,0)$) -- ($(o) + (2,1)$);
  \node [plain] at ($(o) + (1,0.5)$) {$e_1$};

  \draw [decorate,decoration={snake, amplitude=0.5mm}] (o) -- ($(o) + (0,0.4)$);

  \draw [-latex, color=red] ($(o) + (2, 1.25)$) -- ($(n) + (-0.25, 1.25)$) -- ($(n) + (-0.25,0)$);
\end{tikzpicture}
\caption{The map $\psi_{v_m}$  glues together the edges $e_1, e_2$ and assigns the value of the holonomy along $f(e_2) \0 r(e_2) \0 r(e_1)$ (red) to the resulting edge $e'$}
\label{fig:kitaev_trafo_glue_vertex}
\end{figure}
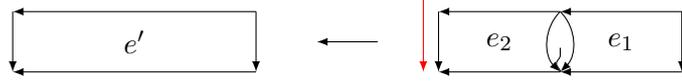

There is an analogous gluing transformation for faces.
Let $f_m \in F$ be a face with exactly two distinct adjacent edges $e_1, e_2$.
We require that there is no vertex $v \in V$ with $f(v) = f_m$.
Choose the cilium of $f_m$ such that the associated face path is given by $p(f_m) = l(e_2)^{-1} \0 r(e_1)$.
This implies that $t(e_1) = t(e_2), s(e_1) = s(e_2)$ and that the edge sides are ordered $r(e_1) < l(e_2)$ as in Figure \ref{fig:kitaev_trafo_glue_face}.

In the transformed graph $\Gamma'$ replace the edges $e_1, e_2$ by a single edge $e' : s(e_1) \to t(e_1)$.
Choose the natural ordering at the faces and vertices adjacent to $e'$ by replacing edge ends and sides of $e_1, e_2$ by the edge ends and sides of $e'$.

The associated map $\psi_{f_m}: K \to K'$ assigns to $e'$ the holonomy along $f(e_1) \0 f(e_2) \0 r(e_2)$:
\begin{equation}
  \label{eq:map_gluing_face}
  (\pi_e \0 \psi_{f_m}) (\gamma) :=
  \begin{cases}
    \pi_-(\pi_{e_1}(\gamma)) \, \pi_{e_2}(\gamma) & \text{for } e = e'\\
    \pi_{e}(\gamma) & \text{otherwise.}
  \end{cases}
\end{equation}

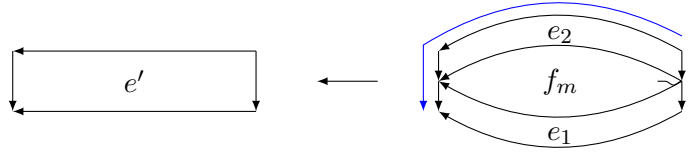
\begin{figure}[h]
\centering
\begin{tikzpicture}[vertex/.style={circle, fill=black, inner sep=0pt, minimum size=2mm}, plain/.style={draw=none, fill=none}, scale=0.8]
  \coordinate (m) at (0,0);
  \draw [latex-] (m) -- ($(m) + (4,0)$);
  \draw [latex-] (m) -- ($(m) + (0,1)$);
  \draw [latex-] ($(m) + (0,1)$) -- ($(m) + (4,1)$);
  \draw [latex-] ($(m) + (4,0)$) -- ($(m) + (4,1)$);
  \node [plain] at ($(m) + (2,0.5)$) {$e'$};

  \draw [latex-] ($(m) + (5,0.5)$) -- ($(m) + (6,0.5)$);

  \coordinate (n) at ($(m) + (7,0)$);
  \draw [latex-] (n) to [bend right=30] node [midway, above=-0.1cm] {$e_1$} ($(n) + (4,0)$);
  \draw [latex-] (n) -- ($(n) + (0,0.5)$);
  \draw [latex-] ($(n) + (0,0.5)$) to [bend right=30] ($(n) + (4,0.5)$);
  \draw [latex-] ($(n) + (4,0)$) -- ($(n) + (4,0.5)$);

  \coordinate (o) at ($(m) + (7,0.5)$);
  \draw [latex-] (o) to [bend left=30] node [midway, above=-0.1cm] {$e_2$} ($(o) + (4,0)$);
  \draw [latex-] (o) -- ($(o) + (0,0.5)$);
  \draw [latex-] ($(o) + (0,0.5)$) to [bend left=30] ($(o) + (4,0.5)$);
  \draw [latex-] ($(o) + (4,0)$) -- ($(o) + (4,0.5)$);

  \draw [decorate,decoration={snake, amplitude=0.5mm}] ($(o) + (4,0)$) -- ($(o) + (3.6,0)$);

  \draw [latex-, color=blue] ($(n) + (-0.25, 0)$) -- ($(o) + (-0.25,0.6)$) to [bend left=30] ($(o) + (4,0.75)$);

  \node at ($(o) + (2,0)$) {$f_m$};
\end{tikzpicture}
\caption{The edges $e_1, e_2$ and are glued together. The resulting edge $e'$ is decorated with the holonomy along $f(e_1) \0 f(e_2) \0 r(e_2)$ (blue)}
\label{fig:kitaev_trafo_glue_face}
\end{figure}

As a face with just two edge sides corresponds to a bivalent vertex in the dual graph $\Gamma^*$, this transformation can be viewed as gluing two edges along a bivalent vertex in $\Gamma^*$ (up to exchanging the roles of $G_-$ and $G_+$).

\begin{lemma}[Properties of the gluing transformations] Let $l \in \left\{ v_m, f_m \right\}$.
  \label{lemma:kitaev_gluing_trafos}
  \begin{compactenum}
  \item
    The map $\psi_{l}$ is Poisson.
  \item
    The map $\psi_{v_m}$ ($\psi_{f_m}$) is invariant under the vertex action at $v_m$ (face action at $f_m$):
    \[
      \psi_{v_m} (\alpha \rhd_{v_m} \gamma) = \psi_{v_m}(\gamma) 
      \qquad
      \psi_{f_m} (x \rhd_{f_m} \gamma) = \psi_{f_m}(\gamma)
      \qquad
      \Forall \gamma \in K, x \in G_-, \alpha \in G_+ \, .
    \]
  \item
    For $\gamma \in K_l$ the map $\psi_l$ intertwines vertex and face actions:
    \begin{equation}
      \label{eq:kitaev_gluing_trafos_compatible_with_actions}
      \psi_{l} (x \rhd_f \gamma) = x \rhd_f \psi_{l}(\gamma) \quad \psi_{l} (\alpha \rhd_v \gamma) = \alpha \rhd_v \psi_{l}(\gamma) \quad \Forall f \in F' \, , x \in G_- \, , v \in V' \, , \alpha \in G_+ \, .
    \end{equation}
  \item
    For $\gamma \in K_l$ the map $\psi_l$ intertwines vertex and face holonomies:
    \begin{equation}
      \label{eq:kitaev_gluing_trafos_compatible_with_holonomies}
      \Hol^v (\gamma) = \Hol'^v(\psi_l(\gamma)) \qquad \Hol^f (\gamma) = \Hol'^f(\psi_l(\gamma))
      \qquad \Forall v \in V', f \in F' \, .
    \end{equation}
  \item
    The map $\psi_l$ preserves flatness:
    one has $\psi_l(K_L) = K'_{L'}$ for all $L' \subseteq V' \dot \cup F'$ and $L := L' \cup \left\{ l \right\}$.
  \end{compactenum}
\end{lemma}

\begin{proof}
  Statement (i) follows from Lemma \ref{lemma:projectionLocallyPoisson} and Theorem \ref{theorem:heisenberg_double_inversion_poisson_actions} (ii).
  Statement (ii) is a direct computation using Lemma \ref{lemma:global_double_projections_properties}.

  We prove (iii) and (iv) for the map $\psi_{v_m}$ as the proof for $\psi_{f_m}$ is completely analogous.
  Note that the vertex holonomy at $t(e_2) = t(e')$ and the face holonomy to the right of $e'$ are invariant under $\psi_{v_m}$ per construction.
  To show (iv), it remains to show that the holonomy functor $\Hol'$ for $\Gamma_D'$ satisfies
  \begin{equation}
    \label{eq:kitaev_gluing_vertex_proof0}
    \Hol' (l(e') \0 b(e')) (\psi_{v_m}(\gamma)) = \Hol (l(e_2) \0 l(e_1) \0 b(e_1)) (\gamma) \, ,
  \end{equation}
  as this is equivalent to saying that the vertex holonomy at $s(e') = s(e_1)$ and the face holonomy to the left of $e_1$ and $e_2$ are invariant under $\psi_{v_m}$.
  Direct computation using Lemma \ref{lemma:global_double_projections_properties} shows that Equation \eqref{eq:kitaev_gluing_vertex_proof0} is equivalent to $\gamma \in K_{v_m}$, which proves (iv).

  Equation \eqref{eq:kitaev_gluing_vertex_proof0} in particular implies that for $\gamma \in K_{v_m}$ the partial vertex and face paths from Equations \eqref{eq:partial_vertex_path} and \eqref{eq:partial_face_path} are invariant under $\psi_{v_m}$.
  Therefore, so are the elements $c_k(\gamma)$ and $d_k(\gamma)$ in Equations \eqref{eq:vertex_action_explicit} and \eqref{eq:face_action_explicit}, that are used to define vertex and face actions.
  This implies (iii).

  It remains to prove Statement (v).
  That $\psi_{l} (K_L) \subseteq K'_{L'}$ follows from Statement (iv).
  To show that $\psi_{l}(K_L) \supseteq K'_{L'}$, we construct a right inverse $\phi_l: K' \to K$ to $\psi_l$.
  First let $l = v_m$.
  As illustrated in Figure \ref{fig:kitaev_trafo_new_vertex}, the map $\phi_{v_m}$ is defined by:
  \begin{equation}
    \label{eq:trafo_new_vertex}
    (\pi_e \0 \phi_{v_m}) (\gamma') =
    \begin{cases}
      \pi_{e'} (\gamma') & \text{if } e = e_1 \\
      \pi_-(\pi_{e'} (\gamma')) & \text{if } e = e_2 \\
      \pi_{e} (\gamma') & \text{otherwise.}
    \end{cases}
  \end{equation}

  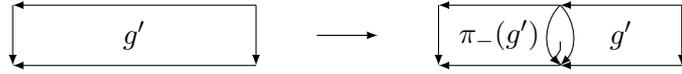
\begin{figure}[h]
  \centering
  \begin{tikzpicture}[vertex/.style={circle, fill=black, inner sep=0pt, minimum size=2mm}, plain/.style={draw=none, fill=none}, scale=0.8]
    \coordinate (m) at (0,0);
    \draw [latex-] (m) -- ($(m) + (4,0)$);
    \draw [latex-] (m) -- ($(m) + (0,1)$);
    \draw [latex-] ($(m) + (0,1)$) -- ($(m) + (4,1)$);
    \draw [latex-] ($(m) + (4,0)$) -- ($(m) + (4,1)$);
    \node [plain] at ($(m) + (2,0.5)$) {$g'$};

    \draw [-latex] ($(m) + (5,0.5)$) -- ($(m) + (6,0.5)$);

    \coordinate (n) at ($(m) + (7,0)$);
    \draw [latex-] (n) -- ($(n) + (2,0)$);
    \draw [latex-] (n) -- ($(n) + (0,1)$);
    \draw [latex-] ($(n) + (0,1)$) -- ($(n) + (2,1)$);
    \draw [latex-] ($(n) + (2,0)$) to [bend left=45] ($(n) + (2,1)$);
    \node [plain] at ($(n) + (1,0.5)$) {$\pi_-(g')$};

    \coordinate (o) at ($(n) + (2,0)$);
    \draw [latex-] (o) -- ($(o) + (2,0)$);
    \draw [latex-] (o) to [bend right=45] ($(o) + (0,1)$);
    \draw [latex-] ($(o) + (0,1)$) -- ($(o) + (2,1)$);
    \draw [latex-] ($(o) + (2,0)$) -- ($(o) + (2,1)$);
    \node [plain] at ($(o) + (1,0.5)$) {$g'$};

    \draw [decorate,decoration={snake, amplitude=0.5mm}] (o) -- ($(o) + (0,0.4)$);
  \end{tikzpicture}
  \caption{The map $\phi_{v_m}$ corresponds to creating a flat bivalent vertex}
  \label{fig:kitaev_trafo_new_vertex}
  \end{figure}
  Direct computations using Lemma \ref{lemma:global_double_projections_properties} show that $\psi_{v_m} \0 \phi_{v_m} = \id_{K'}$ and that for $\gamma' \in K'_{L'}$ one has $\phi_{v_m} (\gamma') \in K_L$.
  This implies $\psi_{v_m}(K_L) \supseteq K'_{L'}$.

  For $l = f_m$ one proceeds analogously, but the right inverse $\phi_{f_m}: K' \to K$ is given by:
  \begin{equation}
    \label{eq:trafo_new_face}
    (\pi_e \0 \phi_{f_m}) (\gamma') =
    \begin{cases}
      \pi_{e'} (\gamma') & \text{if } e = e_1 \\
      \pi_+(\pi_{e'} (\gamma')) & \text{if } e = e_2 \\
      \pi_{e} (\gamma') & \text{otherwise.}
    \end{cases}
  \end{equation}

  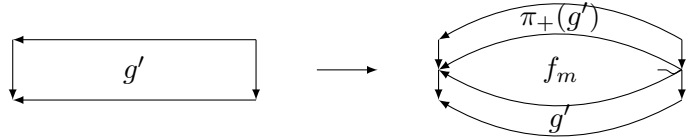
\begin{figure}[h]
  \centering
  \begin{tikzpicture}[vertex/.style={circle, fill=black, inner sep=0pt, minimum size=2mm}, plain/.style={draw=none, fill=none}, scale=0.8]
    \coordinate (m) at (0,0);
    \draw [latex-] (m) -- ($(m) + (4,0)$);
    \draw [latex-] (m) -- ($(m) + (0,1)$);
    \draw [latex-] ($(m) + (0,1)$) -- ($(m) + (4,1)$);
    \draw [latex-] ($(m) + (4,0)$) -- ($(m) + (4,1)$);
    \node [plain] at ($(m) + (2,0.5)$) {$g'$};

    \draw [-latex] ($(m) + (5,0.5)$) -- ($(m) + (6,0.5)$);

    \coordinate (n) at ($(m) + (7,0)$);
    \draw [latex-] (n) to [bend right=30] node [midway, above=-0.1cm] {$g'$} ($(n) + (4,0)$);
    \draw [latex-] (n) -- ($(n) + (0,0.5)$);
    \draw [latex-] ($(n) + (0,0.5)$) to [bend right=30] ($(n) + (4,0.5)$);
    \draw [latex-] ($(n) + (4,0)$) -- ($(n) + (4,0.5)$);

    \coordinate (o) at ($(m) + (7,0.5)$);
    \draw [latex-] (o) to [bend left=30] node [midway, above=-0.1cm] {$\pi_+(g')$} ($(o) + (4,0)$);
    \draw [latex-] (o) -- ($(o) + (0,0.5)$);
    \draw [latex-] ($(o) + (0,0.5)$) to [bend left=30] ($(o) + (4,0.5)$);
    \draw [latex-] ($(o) + (4,0)$) -- ($(o) + (4,0.5)$);

    \draw [decorate,decoration={snake, amplitude=0.5mm}] ($(o) + (4,0)$) -- ($(o) + (3.6,0)$);

    \node at ($(o) + (2,0)$) {$f_m$};
  \end{tikzpicture}
  \caption{The map $\phi_{f_m}$ corresponds to creating a new flat face}
  \label{fig:kitaev_trafo_new_face}
  \end{figure}
\end{proof}

We express the graph transformations from Lemmas \ref{lemma:kitaev_trafo_moving_cilium} and \ref{lemma:kitaev_gluing_trafos} in terms of the set $\mathcal A(\Gamma, L)$ of invariant functions on $K_L$ from Definition \ref{def:functions_on_flat_elements} (ii).
Let $l \in \left\{ v_m, f_m \right\}$ and consider the map $\psi_l$ from Equation \eqref{eq:map_gluing_vertex} or \eqref{eq:map_gluing_face} that corresponds to a graph transformation $\Gamma \to \Gamma'$ that glues two edges together along a bivalent vertex $v_m$ or a face $f_m$ consisting of two edges.

\begin{proposition}[Graph transformations, invariance and flatness]~
\label{proposition:graph_trafos}
  \begin{compactenum}
  \item
    The set $\mathcal A(\Gamma, L)$ does not depend on the choice of cilium for the vertices and faces in $L$.
  \item
    The map $\psi_l^* : C^\infty(K', \R) \to C^\infty(K, \R), h \mapsto h \0 \psi_l$ is an injective homomorphism of Poisson algebras.
    For $L' \subseteq V' \dot \cup F'$ and $L:= L' \cup \left\{ l \right\}$ it satisfies:
    \[
      \psi_l^* (C^\infty(K', \R)_{L'}^{inv}) \subseteq C^\infty(K, \R)_{L}^{inv} \, .
    \]
  \item
    It induces a bijection $\psi^*_{l /\sim} : \mathcal A(\Gamma', L') \to \mathcal A(\Gamma, L), [h'] \mapsto [\psi_l^*(h')]$.
  \end{compactenum}
\end{proposition}

\begin{proof}
  Statement (i) follows directly from Lemma \ref{lemma:kitaev_trafo_moving_cilium} and (ii) from Lemma \ref{lemma:kitaev_gluing_trafos}.

  In the proof of (iii), we first consider the case $l = v_m$.
  By Lemma \ref{lemma:kitaev_gluing_trafos} (v), one has $\psi_{v_m}(K_L) \subseteq K'_{L'}$.
  Together with Statement (ii) this implies that the map $[h'] \mapsto [\psi_{v_m}^*(h')]$ is well-defined.

  To see that it is bijective, we consider the right inverse $\phi_{v_m}: K' \to K$ of $\psi_{v_m}$ from Equation \eqref{eq:trafo_new_vertex}.
  A direct computation using Lemma \ref{lemma:global_double_projections_properties} shows that $\phi_{v_m}(K'_{L'}) \subseteq K_L$ and therefore the map
  \[
    \phi^*_{v_m/ \sim} : \quad \psi^*_{v_m /\sim} (\mathcal A(\Gamma', L')) \to \mathcal A(\Gamma', L') \qquad [h] \mapsto [h \0 \phi_{v_m}]
  \]
  is a well-defined left inverse of $\psi^*_{v_m /\sim}$.
  In particular, $\psi^*_{v_m /\sim}$ is injective.

  To prove that $\psi^*_{v_m /\sim}$ is also surjective, we have to show that for all $h \in C^\infty(K, \R)^{inv}_L$ there is a function $h' \in C^\infty(K', \R)^{inv}_{L'}$ such that
  \begin{equation}
    \label{eq:graph_trafos_proof0}
    \psi_{v_m}^* (h')(\gamma) = h (\gamma) \Forall \gamma \in K_L \, .
  \end{equation}

  Set $h' := h \0 \phi_{v_m}$.
  The map $\phi_{v_m}$ is no left inverse of $\psi_{v_m}$.
  However, for $\gamma \in K_{v_m}$ the map $\phi_{v_m} \0 \psi_{v_m}$ can be expressed by a vertex action at $v_m$:
  \begin{equation}
    \label{eq:graph_trafos_proof1}
    (\phi_{v_m} \0 \psi_{v_m}) (\gamma) = \alpha(\gamma) \rhd_{v_m} \gamma \Forall \gamma \in K_{v_m} \quad \text{ with } \quad \alpha(\gamma) := \pi_+(\pi_{e_2}(\gamma) \, \pi_-(\pi_{e_1}(\gamma))^{-1}) \, .
  \end{equation}
  We therefore obtain for $\gamma \in K_L$
  \[
    \psi_{v_m}^*(h') (\gamma) = (h \0 \phi_{v_m} \0 \psi_{v_m})(\gamma) = h( \alpha(\gamma) \rhd_{v_m} \gamma) = h(\gamma) \, ,
  \]
  which is \eqref{eq:graph_trafos_proof0}.
  It remains to show that $h' \in C^\infty(K', \R)^{inv}_{L'}$.
  A direct computation using Lemma \ref{lemma:global_double_projections_properties} shows that
  \[
    \phi_{v_m} \0 \rhd_f = \rhd_f \0 (\id_{G_-} \x \phi_{v_m}) \Forall f \in F'
    \quad \text{ and } \quad \phi_{v_m} \0 \rhd_v = \rhd_v \0 (\id_{G_+} \x \phi_{v_m}) \Forall v \in V' \setminus \left\{ t(e') \right\} \, .
  \]
  Together with the fact that $\phi_{v_m}(K'_{L'}) \subseteq K_L$ this implies $h' = h \0 \phi_{v_m} \in C^\infty(K', \R)_{L'}^f$ for all $f \in F'$ and $h' \in C^\infty(K', \R)_{L'}^v$ for all $v \in V' \setminus \left\{ t(e') \right\}$.

  To conclude the proof, we have to show that $h' \in C^\infty(K', \R)^{v}_{L'}$ for $v = t(e') = t(e_2)$.
  Let $\beta \in G_+, \gamma' \in K'_{L'}$.
  By Lemma \ref{lemma:kitaev_gluing_trafos} (v) there is a $\gamma \in K_L$ with $\psi_{v_m}(\gamma) = \gamma'$.
  We compute
  \[
    h'(\beta \rhd_v \gamma') =  h \0 \phi_{v_m}(\beta \rhd_v \psi_{v_m}(\gamma)) \stackrel{\eqref{eq:kitaev_gluing_trafos_compatible_with_actions}}= h \0 \phi_{v_m} \0 \psi_{v_m} (\beta \rhd_v \gamma) \, .
  \]
  It is easy to see that $\beta \rhd_v \gamma \in K_{v_m}$, so that Equation \eqref{eq:graph_trafos_proof1} can be applied to obtain:
  \[
    h'(\beta \rhd_v \gamma') = h (\alpha(\beta \rhd_v \gamma) \rhd_{v_m} (\beta \rhd_v \gamma)) \, .
  \]
  One can check directly using Lemma \ref{lemma:global_double_projections_properties} that the actions $\rhd_v$ and $\rhd_{v_m}$ commute.
  This implies:
  \[
    h'(\beta \rhd_v \gamma') = h (\beta \rhd_v (\alpha(\beta \rhd_v \gamma) \rhd_{v_m} \gamma)) \, .
  \]
  A computation using the properties of $\pi_\pm$ from Lemma \ref{lemma:global_double_projections_properties} shows that $\gamma \in K_L$ implies $\alpha \rhd_{v_m} \gamma \in K_L$ for all $\alpha \in G_+$.
  We use the invariance of $h \in C^\infty(K, \R)^{inv}_L$ under $\rhd_v$ and $\rhd_{v_m}$ on $K_L$ to obtain
  \[
    h'(\beta \rhd_v \gamma') = h (\alpha(\beta \rhd_v \gamma) \rhd_{v_m} \gamma) = h (\gamma) = h (\alpha(\gamma) \rhd_{v_m} \gamma) 
    \stackrel{\eqref{eq:graph_trafos_proof1}}=
    h \0 \phi_{v_m} \0 \psi_{v_m} (\gamma) = h' (\gamma') \, .
  \]

  For the gluing transformation of a face $l = f_m$ one proceeds in the same way using the right inverse $\phi_{f_m}$ from Equation \eqref{eq:trafo_new_face}.
  Again, $\phi_{f_m}$ is not a left inverse to $\psi_{f_m}$, but for $\gamma \in K_{f_m}$ it satisfies $(\phi_{f_m} \0 \psi_{f_m}) (\gamma) = x(\gamma) \rhd_{f_m} \gamma$ with $ x(\gamma) = \pi_-(\pi_+(\pi_{e_1}(\gamma))^{-1} \, \pi_{e_2}(\gamma))^{-1} $.
\end{proof}

Proposition \ref{proposition:graph_trafos} allows us to relate a Poisson-Kitaev model with one on a paired doubly ciliated ribbon graph (Definition \ref{def:double_graph_basics} (iv)).

\begin{corollary}[Transformation into a paired graph]
  \label{corollary:kitaev_transform_into_paired_graph}
  Let $(K, \Gamma)$ be a Poisson-Kitaev model and $(v_1, f_1), \dots, (v_n, f_n)$ pairwise distinct sites of $\Gamma$.
  There is a Poisson-Kitaev model $(K', \Gamma')$ on a paired doubly-ciliated ribbon graph $\Gamma'$ and pairwise distinct sites $(v_1', f_1'), \dots, (v_n', f_n')$ such that:
  \begin{compactenum}
  \item
    There is a Poisson map $\psi: K' \to K$.
  \item
    The map $\psi$ induces a bijection $\psi^*_{/\sim} : \mathcal A(\Gamma, L) \to \mathcal A(\Gamma', L') \, , [h] \mapsto [h \0 \psi]$ for the sets $L, L'$ given by $L = (V \setminus \left\{ v_1, \dots, v_n \right\}) \dot \cup (F \setminus \left\{ f_1, \dots, f_n \right\})$ and $L' = (V' \setminus \left\{ v'_1, \dots, v'_n \right\}) \dot \cup (F' \setminus \left\{ f'_1, \dots, f'_n \right\})$.
  \item
    The oriented  surfaces associated to $\Gamma$ and $\Gamma'$ by gluing annuli to $f_1, \dots, f_n$ and $f'_1, \dots, f'_n$, respectively, and disks to all other faces are homeomorphic.
  \end{compactenum}
\end{corollary}

\begin{proof}
  To obtain the paired graph $\Gamma'$, we perform the following steps, as illustrated in Figure \ref{fig:kitaev_moduli_space_proof}:
  \begin{enumerate}[nolistsep, noitemsep]
    \item 
      Split every loop into two edges by creating a bivalent vertex.
    \item
      Double each edge for which the left and right face coincide.
    \item
      For every vertex $v \notin \left\{ v_1, \dots, v_n  \right\}$ create an adjacent face $f_v$ by doubling the first edge of $v$.
      Cyclically permute the ordering at $v$ and $f_v$ so that $(v, f_v)$ becomes a site.
    \item
      For every face $f \notin \left\{ f_1, \dots, f_n  \right\}$ that has not been created in step 3 
      split the first edge of $f$ into two by creating a bivalent vertex $v_f$.
      Cyclically permute the orderings of $f$ and $v_f$ so that $(v_f, f)$ becomes a site.
  \end{enumerate}

  \begin{figure}[h]
  \centering
  \begin{tikzpicture}[vertex/.style={circle, fill=black, inner sep=0pt, minimum size=2mm}, plain/.style={draw=none, fill=none}, scale=0.8]
    \begin{scope}[scale=0.9]

    \node at (-3.5, 0.5) {1.};
    \node at (-3.5, -3.5) {3.};
    \node at (7.5, 0.5) {2.};
    \node at (7.5, -3.5) {4.};

    \begin{scope}[shift={(-1,0)}]
      \coordinate (q) at (0,0);
      \coordinate (q1) at ($(q) + (0:0.375)$);
      \coordinate (q2) at ($(q) + (120:0.375)$);
      \coordinate (q3) at ($(q) + (240:0.375)$);

      \draw [-latex] (q1) to [out=-60, in=-30, looseness=1.3] ($(q)!2.5cm!(q3)$) to  [out=150, in =-180, looseness=1.3] (q2);
      \draw [-latex] (q3) to [out=-60, in=-30, looseness=1.3] ($(q)!1.85cm!(q3)$) to [out=150, in =-180, looseness=1.3] (q3);
      \draw [-latex] (q1) -- (q3);
      \draw [-latex] (q3) -- (q2);

      \draw [-latex] (0.9,-1) -- (1.9,-1);

      \coordinate (q) at ($(q) + (4,0)$);
      \coordinate (q1) at ($(q) + (0:0.375)$);
      \coordinate (q2) at ($(q) + (120:0.375)$);
      \coordinate (q3) at ($(q) + (240:0.375)$);

      \draw [-latex] (q1) -- (q3);
      \draw [-latex] (q1) to [out=-60, in=-30, looseness=1.3] ($(q)!2.5cm!(q3)$);
      \draw [-latex] (q3) to [out=-60, in=-30, looseness=1.3] ($(q)!1.85cm!(q3)$);
      \draw [-latex]  ($(q)!2.5cm!(q3)$) to [in=-45, out=-15, looseness=1.3] ($(q)!1.85cm!(q3)$);

      \draw [-latex] (q3) -- (q2);
      \draw [-latex] ($(q)!2.5cm!(q3)$) to  [out=150, in =-180, looseness=1.3] (q2);
      \draw [-latex] ($(q)!1.85cm!(q3)$) to [out=150, in =-180, looseness=1.3] (q3);
      \draw [-latex]  ($(q)!2.5cm!(q3)$) to [in=-195, out=-225, looseness=1.3] ($(q)!1.85cm!(q3)$);
    \end{scope}

    \begin{scope}[shift={(1.75,0)}]
      \coordinate (m) at (7.5,-2);
      \coordinate (m') at ($(m) + (0.325, -0.5)$);

      \draw [latex-] (m') to [out=0, in=270] ($(m) + (0.65, 0)$) -- ($(m) + (0.65, 2.5)$);
      \draw [latex-] (m') to [out=180, in=270] (m) -- ($(m) + (0,2.5)$);
      \draw [-latex] ($(m)+(0,2.5)$) -- ($(m)+(0.65,2.5)$);
      \draw [-latex] (m') to [in=210, out=180] ($(m) + (0.325,0)$);
      \draw ($(m) + (0.325,0)$) to [in=30,out=0] (m');

      \draw [-latex] ($(m) + (1.25, 1)$) -- ($(m) + (2.5, 1)$);

      \coordinate (m) at (11,-2);
      \coordinate (m') at ($(m) + (0.5, -0.5)$);

      \draw [latex-] ($(m) + (1, 2.5)$) -- ($(m) + (0.5, 2.5)$);
      \draw [latex-] (m') to [out=0, in=-45] ($(m) + (1, 2.5)$);
      \draw [latex-] ($(m') + (0,0.5)$) to [out=0, in=-45, looseness=0.9] ($(m) + (0.5, 2.5)$);
      \draw [-latex] ($(m) + (0.5,0)$) to [in=30,out=0] (m'); 

      \draw [-latex] ($(m)+(0,2.5)$) -- ($(m)+(0.5,2.5)$);
      \draw [latex-] (m') to [out=180, in=225] ($(m) + (0,2.5)$);
      \draw [latex-] ($(m') + (0, 0.5)$) to [out=180, in=225, looseness=0.9] ($(m) + (0.5,2.5)$);
      \draw [-latex] (m') to [in=210, out=180] ($(m) + (0.5,0)$);

      \node at ($(7.5,-2) + (0.325, 1)$) {$e$};
    \end{scope}

    \begin{scope}[shift={(-2.5, -5.5)}, scale=0.8]
    \coordinate (p) at (0,0);
    \coordinate (p1) at ($(p) + (36:0.553)$);
    \coordinate (p2) at ($(p) + (108:0.553)$);
    \coordinate (p3) at ($(p) + (180:0.553)$);
    \coordinate (p4) at ($(p) + (252:0.553)$);
    \coordinate (p5) at ($(p) + (-36:0.553)$);

    \coordinate (q) at ($(p) + (2,0) + (0.553,0) + (0.375,0)$);
    \coordinate (q1) at ($(q) + (0:0.375)$);
    \coordinate (q2) at ($(q) + (120:0.375)$);
    \coordinate (q3) at ($(q) + (240:0.375)$);

    \draw [-latex] (p1)--(q2);
    \draw [-latex] (q3)--(q2);
    \draw [-latex] (p5)--(p1);
    \draw [-latex] (p5)--(q3);

    \draw [-latex] (p2) -- (p1);
    \draw [-latex] ($(p1)!2cm!-90:(p2)$) -- (p1);
    \draw [-latex] ($(p2)!2cm!90:(p1)$) -- (p2);
    \draw [-latex] ($(p2)!2cm!90:(p1)$) -- ($(p1)!2cm!-90:(p2)$);

    \draw [-latex] (p3) -- (p2);
    \draw [-latex] ($(p2)!2cm!-90:(p3)$) -- (p2);
    \draw [-latex] ($(p3)!2cm!90:(p2)$) -- (p3);
    \draw [-latex] ($(p3)!2cm!90:(p2)$) -- ($(p2)!2cm!-90:(p3)$);

    \draw [-latex] (p3) -- (p4);
    \draw [-latex] (p3) -- ($(p3)!2cm!-90:(p4)$);
    \draw [-latex] (p4) -- ($(p4)!2cm!90:(p3)$);
    \draw [-latex] ($(p3)!2cm!-90:(p4)$) -- ($(p4)!2cm!90:(p3)$);

    \draw [-latex] (p5) -- (p4);
    \draw [-latex] ($(p4)!2cm!-90:(p5)$) -- (p4);
    \draw [-latex] ($(p5)!2cm!90:(p4)$) -- (p5);
    \draw [-latex] ($(p5)!2cm!90:(p4)$) -- ($(p4)!2cm!-90:(p5)$);

    \draw [decorate,decoration={snake, amplitude=0.5mm}] (p2) -- ($(p2) + (p)!-0.5cm!(p2)$);
    \end{scope}

    \draw [-latex] (0, -5.5) -- (1, -5.5);

    \begin{scope}[shift={(3, -5.5)}, scale=0.8]
    \coordinate (p) at (0,0);
    \coordinate (p1) at ($(p) + (36:0.553)$);
    \coordinate (p2) at ($(p) + (108:0.553)$);
    \coordinate (p3) at ($(p) + (180:0.553)$);
    \coordinate (p4) at ($(p) + (252:0.553)$);
    \coordinate (p5) at ($(p) + (-36:0.553)$);

    \coordinate (q) at ($(p) + (2,0) + (0.553,0) + (0.375,0)$);
    \coordinate (q1) at ($(q) + (0:0.375)$);
    \coordinate (q2) at ($(q) + (120:0.375)$);
    \coordinate (q3) at ($(q) + (240:0.375)$);

    \draw [-latex] (p1)--(q2);
    \draw [-latex] (q3)--(q2);
    \draw [-latex] (p5)--(p1);
    \draw [-latex] (p5)--(q3);

    \draw [-latex] (p2) -- (p1);
    \draw [-latex] ($(p1)!2cm!-90:(p2)$) -- (p1);
    \draw [-latex] ($(p2)!2cm!90:(p1)$) -- (p2);
    \draw [-latex] ($(p2)!2cm!90:(p1)$) -- ($(p1)!2cm!-90:(p2)$);

    \coordinate (p23) at ($(p3)!0.5!(p2)$);
    \draw [-latex] (p3) -- (p23);
    \draw [-latex] ($(p3)!2cm!90:(p23)$) to [bend right=30] (p3);
    \draw [latex-] (p23) to [bend right=-30] ($(p23)!2cm!-90:(p3)$);
    \draw [-latex] ($(p3)!2cm!90:(p23)$) -- ($(p23)!2cm!-90:(p3)$);

    \draw [latex-] (p2) -- (p23);
    \draw [-latex] ($(p2)!2cm!-90:(p23)$) to [bend right=-30] (p2);
    \draw [latex-] (p23) to [bend right=30] ($(p23)!2cm!90:(p2)$);
    \draw [latex-] ($(p2)!2cm!-90:(p23)$) -- ($(p23)!2cm!90:(p2)$);

    \draw [-latex] (p3) -- (p4);
    \draw [-latex] (p3) -- ($(p3)!2cm!-90:(p4)$);
    \draw [-latex] (p4) -- ($(p4)!2cm!90:(p3)$);
    \draw [-latex] ($(p3)!2cm!-90:(p4)$) -- ($(p4)!2cm!90:(p3)$);

    \draw [-latex] (p5) -- (p4);
    \draw [-latex] ($(p4)!2cm!-90:(p5)$) -- (p4);
    \draw [-latex] ($(p5)!2cm!90:(p4)$) -- (p5);
    \draw [-latex] ($(p5)!2cm!90:(p4)$) -- ($(p4)!2cm!-90:(p5)$);

    \draw [decorate,decoration={snake, amplitude=-0.5mm}] (p23) -- ($(p23)!0.75cm!180:(p)$);
    \draw [decorate,decoration={snake, amplitude=0.5mm}] (p23) -- ($(p23)!-0.4cm!180:(p)$);

    \end{scope}

    \begin{scope}[shift={(10,-5.5)}, scale=0.8]
      \coordinate (pl) at (0,-1);
      \coordinate (pr) at (0.65,-1);
      \coordinate (ql) at (0,1);
      \coordinate (qr) at (0.65,1);

      \draw [-latex] (pr) -- (qr);
      \draw [-latex] (pl) -- (ql);
      \draw [-latex] (pr) -- (pl);
      \draw [-latex] (qr) -- (ql);

      \coordinate (r) at ($(1.732, 0) + (0.65, 0)$);
      \coordinate (r1) at ($(r) + (90:2cm)$);
      \coordinate (r2) at ($(r) + (30:2cm)$);
      \coordinate (r3) at ($(r) + (-30:2cm)$);
      \coordinate (r4) at ($(r) + (-90:2cm)$);

      \coordinate (l) at (-1.376,0);
      \coordinate (l1) at ($(l) + (-108:1.701cm)$);
      \coordinate (l2) at ($(l) + (-180:1.701cm)$);
      \coordinate (l3) at ($(l) + (108:1.701cm)$);

      \draw [-latex] (pl) -- (l1);
      \draw [-latex] (l1) -- (l2);
      \draw [-latex] (l2) -- (l3);
      \draw [-latex] (l3) -- (ql);

      \draw [decorate,decoration={snake, amplitude=0.5mm}] (l2) -- ($(l2)!0.75cm!(l)$);

      \coordinate (ql') at ($(ql)!1cm!-90:(l3)$);
      \coordinate (qr') at ($(qr)!1cm!90:(r1)$);
      \coordinate (q') at (intersection of ql--ql' and qr--qr');

      \draw [-latex] (ql) -- (q');
      \draw [-latex] (l3) -- ($(l3) + (q') - (ql)$);
      \draw [-latex] ($(l3) + (q') - (ql)$) -- (q');

      \draw [-latex] (l3) -- ($(l3) + (ql)!1!72:(q') - (ql)$);
      \draw [-latex] (l2) -- ($(l2) + (ql)!1!72:(q') - (ql)$);
      \draw [-latex] ($(l2) + (ql)!1!72:(q') - (ql)$) -- ($(l3) + (ql)!1!72:(q') - (ql)$);

      \draw [-latex] (l2) -- ($(l2) + (ql)!1!144:(q') - (ql)$);
      \draw [-latex] (l1) -- ($(l1) + (ql)!1!144:(q') - (ql)$);
      \draw [-latex] ($(l1) + (ql)!1!144:(q') - (ql)$) -- ($(l2) + (ql)!1!144:(q') - (ql)$);

      \draw [-latex] (l1) -- ($(l1) + (ql)!1!216:(q') - (ql)$);
      \draw [-latex] (pl) -- ($(pl) + (ql)!1!216:(q') - (ql)$);
      \draw [-latex] ($(pl) + (ql)!1!216:(q') - (ql)$) -- ($(l1) + (ql)!1!216:(q') - (ql)$);
    \end{scope}

    \draw [-latex] (10.75, -5.5) -- (11.75, -5.5);

    \begin{scope}[shift={(15,-5.5)}, scale=0.8]
      \coordinate (pl) at (0,-1);
      \coordinate (pr) at (0.65,-1);
      \coordinate (ql) at (0,1);
      \coordinate (qr) at (0.65,1);

      \draw [-latex] (pr) -- (qr);
      \draw [-latex] (pl) -- (ql);
      \draw [-latex] (pr) -- (pl);
      \draw [-latex] (qr) -- (ql);

      \coordinate (r) at ($(1.732, 0) + (0.65, 0)$);
      \coordinate (r1) at ($(r) + (90:2cm)$);
      \coordinate (r2) at ($(r) + (30:2cm)$);
      \coordinate (r3) at ($(r) + (-30:2cm)$);
      \coordinate (r4) at ($(r) + (-90:2cm)$);

      \coordinate (l) at (-1.376,0);
      \coordinate (l1) at ($(l) + (-108:1.701cm)$);
      \coordinate (l2) at ($(l) + (-180:1.701cm)$);
      \coordinate (l3) at ($(l) + (108:1.701cm)$);

      \draw [-latex] (pl) -- (l1);
      \draw [-latex] (l1) -- (l2);
      \draw [-latex] (l2) -- ($(l2)!0.5!(l3)$);
      \draw [-latex] ($(l2)!0.5!(l3)$) -- (l3);
      \draw [-latex] (l3) -- (ql);

      \coordinate (ql') at ($(ql)!1cm!-90:(l3)$);
      \coordinate (qr') at ($(qr)!1cm!90:(r1)$);
      \coordinate (q') at (intersection of ql--ql' and qr--qr');

      \draw [-latex] (ql) -- (q');
      \draw [-latex] (l3) -- ($(l3) + (q') - (ql)$);
      \draw [-latex] ($(l3) + (q') - (ql)$) -- (q');

      \coordinate (l2') at ($(l2) + (ql)!1!72:(q') - (ql)$);
      \coordinate (l3') at ($(l3) + (ql)!1!72:(q') - (ql)$);
      \coordinate (l23) at ($(l2)!0.5!(l3)$);
      \coordinate (l23') at ($(l2')!0.5!(l3')$);
      \draw [-latex] (l3) -- (l3');
      \draw [-latex] (l2) -- (l2');
      \draw [-latex] (l2') -- (l23');
      \draw [latex-] (l3') -- (l23');
      \draw [-latex] (l23) to [bend left=45, looseness=1.2] (l23');
      \draw [-latex] (l23) to [bend right=45, looseness=1.2] (l23');

      \draw [decorate,decoration={snake, amplitude=0.5mm}] (l23) -- ($(l23)!-0.75cm!(l23')$);
      \draw [decorate,decoration={snake, amplitude=0.5mm}] (l23) -- ($(l23)!0.4cm!(l23')$);

      \draw [-latex] (l2) -- ($(l2) + (ql)!1!144:(q') - (ql)$);
      \draw [-latex] (l1) -- ($(l1) + (ql)!1!144:(q') - (ql)$);
      \draw [-latex] ($(l1) + (ql)!1!144:(q') - (ql)$) -- ($(l2) + (ql)!1!144:(q') - (ql)$);

      \draw [-latex] (l1) -- ($(l1) + (ql)!1!216:(q') - (ql)$);
      \draw [-latex] (pl) -- ($(pl) + (ql)!1!216:(q') - (ql)$);
      \draw [-latex] ($(pl) + (ql)!1!216:(q') - (ql)$) -- ($(l1) + (ql)!1!216:(q') - (ql)$);
    \end{scope}

    \end{scope}
  \end{tikzpicture}
  \caption{Examples for transformations 1. through 4. The edge $e$ in 2.  is incident to a univalent vertex and therefore both its sides are on the same face path}
  \label{fig:kitaev_moduli_space_proof}
  \end{figure}
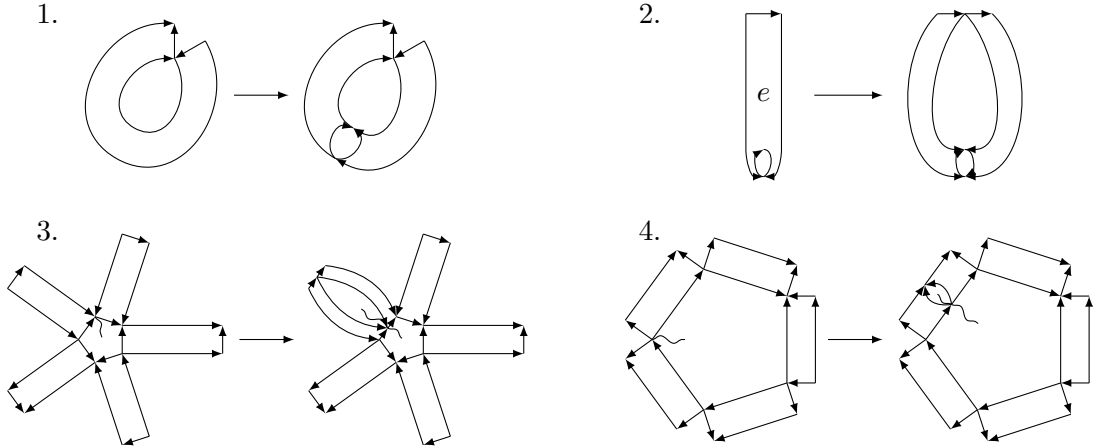

  The resulting graph $\Gamma'$ has no loops because of Step  1.
  Step 2. ensures that the faces left and right of every edge differ.
  The selected sites $(v_1, f_1), \dots, (v_n, f_n)$ are left intact and become the sites $(v_1', f_1'), \dots, (v_n', f_n')$ of the transformed graph.
  Steps 3. and 4. guarantee that all remaining vertices and faces can be combined into sites, so that $\Gamma'$ is a paired graph.

  The map $\psi: K' \to K$ is the map associated to gluing the pairs of edges created in Steps 1. through 4. back together in reverse order.
  It is Poisson by Lemma \ref{lemma:kitaev_gluing_trafos}, which proves (i).
  Statement (ii) follows from
  Proposition \ref{proposition:graph_trafos} (iii).
  
  We obtain the  oriented  surface associated to $\Gamma'$ by gluing annuli to $f_i' = f_i$ for all $i$ and disks to all other faces, which includes the faces created in Steps 2. and 3.
  Doubling an edge and gluing a disk to the new face $f_m$ and an annulus or disk to an adjacent face $f$ yields a homeomorphic topological space as only gluing an annulus or disk to $f$.
  Thus,  the oriented  surfaces obtained from $\Gamma$ and $\Gamma'$ by gluing annuli to $f_1, \dots, f_n$ and disks to all other faces are homeomorphic, which proves (iii).
\end{proof}

\subsection{Invariance and flatness}
\label{subsection:gauge_invariance_flatness}

In this section, we show how the vertex and face actions from Definition \ref{def:vertex_face_actions} interact with the vertex and face holonomies $\Hol^v, \Hol^f$ from Definition \ref{def:vertex_face_paths_and_holonomies}.
We prove commutation relations of vertex and face actions that are Poisson counterparts of the commutation relations for quantum vertex and face operators in Lemma \ref{lemma:quantum_kitaev_commutation_relations}.
We then prove that the Poisson bracket on $K$ induces a Poisson algebra structure on the set $\mathcal A(\Gamma, L)$ from Definition \ref{def:functions_on_flat_elements} (ii).

\begin{lemma}[Relations between actions and holonomies] Let $\gamma \in K, x \in G_-, \alpha \in G_+$.
  \label{lemma:kitaev_actions_and_holonomies}
  \begin{compactenum}
    \item
      For $v \neq v' \in V$ ($f \neq f' \in F$), the associated actions and holonomies satisfy:
      \[
        \Hol^v(\alpha \rhd_{v'} \gamma) = \Hol^v(\gamma) \qquad (\Hol^f(x \rhd_{f'} \gamma) = \Hol^f(\gamma)) \, .
      \]
    \item
      For all pairs $v \in V, f \in F$ with $v(f) \neq v$ and $f(v) \neq f$ (see Definition \ref{def:double_graph_basics}) we have:
      \[
        \Hol^v(x \rhd_f \gamma) = \Hol^v(\gamma) \qquad \Hol^f(\alpha \rhd_v \gamma) = \Hol^f(\gamma) \, .
      \]
    \item
      For each site $(v,f)$ we obtain:
      \begin{equation}
        \label{eq:holonomy_and_action_for_a_site}
        \Hol^{(v,f)} (x \rhd_f (\alpha \rhd_v \gamma)) = (x\alpha) \, \Hol^{(v,f)}(\gamma) \, (x\alpha)^{-1} \, .
      \end{equation}
  \end{compactenum}
\end{lemma}

\begin{proof}
  (i) and (ii): Direct computations using Lemma \ref{lemma:global_double_projections_properties}.

  For (iii), one can first assume that there are no loops at $v$ and that every edge is traversed at most once by the face path $p(f)$.
  In this case, (iii) follows from a computation using Lemma \ref{lemma:global_double_projections_properties}.

  In the general case, we transform the graph $\Gamma$ by introducing a bivalent vertex on every loop of $v$ and doubling every edge of $f$ that occurs twice in the face path $p(f)$, similar to Steps 1. and 2. in the proof of Corollary \ref{corollary:kitaev_transform_into_paired_graph}.
  Write $\Gamma'$ for the resulting graph and let $L'$ be the set of vertices and faces of $\Gamma'$ that have been created by this procedure.
  From Lemma \ref{lemma:kitaev_gluing_trafos} we obtain a corresponding Poisson map $\psi: K' \to K$ that implements gluing the new pairs of edges back together.

  Denote by $\Hol'^{(v,f)}$ the combined holonomy from Definition \ref{def:vertex_face_paths_and_holonomies} (iii) for the site $(v,f)$ in $\Gamma'$.
  We use Lemma \ref{lemma:kitaev_gluing_trafos} (iii) and (iv) to compute for $\gamma' \in K'_{L'}, x \in G_-, \alpha \in G_+$:
  \begin{equation*}
    \Hol^{(v,f)} (x \rhd_f (\alpha \rhd_v \psi(\gamma'))) \stackrel{\eqref{eq:kitaev_gluing_trafos_compatible_with_actions}}{=} \Hol^{(v,f)} (\psi(x \rhd_f (\alpha \rhd_v \gamma'))) \stackrel{\eqref{eq:kitaev_gluing_trafos_compatible_with_holonomies}}{=} \Hol'^{(v,f)} (x \rhd_f (\alpha \rhd_v \gamma')) \, .
  \end{equation*}
  In these equations we used that $G_+ \rhd_v K'_{L'} \subseteq K'_{L'}$ and $G_- \rhd_f K'_{L'} \subseteq K'_{L'}$, which follows from a direct computation using Lemma \ref{lemma:global_double_projections_properties}.
  As $\Gamma'$ has no loops at $v$ and no edge of $f$ occurs twice in its face path, we have already proven Statement (iii) for $\Gamma'$.
  So we can apply it to obtain:
  \begin{equation*}
    \Hol^{(v,f)} (x \rhd_f (\alpha \rhd_v \psi(\gamma'))) = (x \alpha) \, \Hol'^{(v,f)} (\gamma') \, (x \alpha)^{-1} \stackrel{\eqref{eq:kitaev_gluing_trafos_compatible_with_holonomies}}{=} (x \alpha) \, \Hol^{(v,f)} (\psi(\gamma')) \, (x \alpha)^{-1} \, .
  \end{equation*}
  As $\psi|_{K'_{L'}} : K'_{L'} \to K$ is surjective by Lemma \ref{lemma:kitaev_gluing_trafos} (v), this concludes the proof.
\end{proof}

\begin{lemma}[Stability of $K_L$ under vertex and face actions]
  \label{lemma:kitaev_flat_subspace_stable}
  Let $L \subseteq V \dot \cup F$.
  If $v \in V \cap L$ ($f \in F \cap L$) or $f(v) \notin L$ ($v(f) \notin L$), the subspace $K_L$ is stable under the vertex action $\rhd_v$ (face action $\rhd_f$) from Definition \ref{def:vertex_face_actions}:
  \[
    G_+ \rhd_v K_L = K_L \qquad (G_- \rhd_f K_L = K_L) \, .
  \]
\end{lemma}

\begin{proof}
  We only show the statement for the vertex action $\rhd_v$ as the proof for $\rhd_f$ is analogous.
  First we consider the case $v \in V \cap L$.
  We prove $G_+ \rhd_v (K_v \cap K_l) = (K_v \cap K_l)$ for all $l \in L$; as $v \in L$, this then implies $G_+ \rhd_v K_L = K_L$.

  To see that $G_+ \rhd_v K_v = K_v$, consider the face $f(v)$.
  If $(v, f(v))$ is a site, this follows from Equation \eqref{eq:holonomy_and_action_for_a_site} and the properties of the projection $\pi_-$ in Lemma \ref{lemma:global_double_projections_properties} because $\Hol^v = \pi_- \0 \Hol^{(v, f(v))}$.
  Otherwise, first cyclically permute the ordering of $f(v)$ so that $(v, f(v))$ becomes a site; then Equation \eqref{eq:holonomy_and_action_for_a_site} holds with respect to the new ordering.

  That $G_+ \rhd_v K_{v'} = K_{v'}$ for all other vertices $v' \neq v$ follows from Lemma \ref{lemma:kitaev_actions_and_holonomies} (i).

  It remains to prove $G_+ \rhd_v (K_v \cap K_{f'}) = (K_v \cap K_{f'})$ for $f' \in F$.
  If $f(v) = f'$, we can cyclically permute the ordering of the edge sides of $f'$ until $(v,f')$ becomes a site because the space $K_{f'}$ depends only on the cyclic ordering of $f'$.
  Then we can use the identity $\Hol^{f'} = \pi_+ \0 \Hol^{(v, f')}$ together with Lemma \ref{lemma:kitaev_actions_and_holonomies} (iii) to prove that $G_+ \rhd_v (K_v \cap K_{f'}) \subseteq K_{f'}$ holds.
  As we have already shown that $G_+ \rhd_v K_v = K_v$, this implies $G_+ \rhd_v (K_v \cap K_{f'}) = (K_v \cap K_{f'})$.

  If $f(v) \neq f'$ and there is another vertex $v'$ adjacent to $f'$, shift the ordering of $f'$ so that $v(f') = v'$.
  Then $v(f') \neq v$ and $f(v) \neq f'$, so that Lemma \ref{lemma:kitaev_actions_and_holonomies} (ii) applies and we have $G_+ \rhd_v K_{f'} = K_{f'}$.
  If $f(v) \neq f'$ and $v$ is the only vertex adjacent to $f'$, transform the graph $\Gamma$ by introducing a bivalent vertex $v'_m$ on an edge $e$ of $f'$, thus splitting it into the edges $e'_1$ and $e'_2$.
  From Lemma \ref{lemma:kitaev_gluing_trafos} we obtain the map $\psi_{v_m'}: K' \to K$ from the Poisson-Kitaev model $K'$ of the transformed graph $\Gamma'$ that implements gluing the edges $e'_1$ and $e'_2$ back together.
  In $\Gamma'$ we can shift the ordering of $f'$ so that $v(f') = v'_m$.
  Then we have $f(v) \neq f'$ and $v(f') \neq v$ in $\Gamma'$, so that $G_+ \rhd_v K'_{f'} = K'_{f'}$ by Lemma \ref{lemma:kitaev_actions_and_holonomies} (ii).
  To translate this into a statement for the graph $\Gamma$, let $\gamma \in K_{f'}$.
  By Lemma \ref{lemma:kitaev_gluing_trafos} (v) there is a $\gamma' \in K'_{f'} \cap K'_{v'_m}$ with $\psi_{v'_m}(\gamma') = \gamma$.
  For $\alpha \in G_+$ we obtain
  \[
    \Hol^{f'}( \alpha \rhd_v \gamma) = \Hol^{f'}(\alpha \rhd_v \psi_{v'_m}(\gamma')) \stackrel{\eqref{eq:kitaev_gluing_trafos_compatible_with_actions}}{=} \Hol^{f'} (\psi_{v'_m}(\alpha \rhd_v \gamma')) \, .
  \]
  Lemma \ref{lemma:kitaev_actions_and_holonomies} (i) implies $\alpha \rhd_v \gamma' \in K'_{v_m'}$, so we can apply \eqref{eq:kitaev_gluing_trafos_compatible_with_holonomies} to the right hand side to obtain 
  \[
    \Hol^{f'}( \alpha \rhd_v \gamma) = \Hol'^{f'}(\alpha \rhd_v \gamma') \, .
  \]
  For the graph $\Gamma'$ we have already shown that $G_+ \rhd_v K'_{f'} = K'_{f'}$, so the last term is equal to $1_{G_+}$ and we obtain $G_+ \rhd_v K_{f'} = K_{f'}$.

  We have thus shown that $G_+ \rhd_v (K_v \cap K_l) = K_v \cap K_l$ for all $l \in L$, which implies $G_+ \rhd_v K_L = K_L$.

  Now consider the case $v \notin L$ and $f(v) \notin L$.
  Lemma \ref{lemma:kitaev_actions_and_holonomies} (i) implies that $G_+ \rhd_v K_{v'} = K_{v'}$ for all vertices $v' \in V \cap L$.
  For $f' \in F \cap L$ we have $f(v) \neq f'$.
  We argue as before by shifting the ordering of $f'$ (after potentially creating a bivalent vertex adjacent to $f'$) so that $v(f') \neq v$.
  Then $G_+ \rhd_v K_{f'} = K_{f'}$ follows from Lemma \ref{lemma:kitaev_actions_and_holonomies} (ii).
\end{proof}

Next, we prove that the actions $\rhd_v, \rhd_f$ for a vertex $v$ and face $f$ are Poisson maps.
For this we use the graph transformations from Section \ref{subsection:graph_trafos_kitaev}.
To be able to do so, we first need a lemma that shows that the restriction of the Poisson bracket $\left\{ g, h \right\}$ to $K_L$ depends only on $g|_{K_L}$ and $h|_{K_L}$ if $g, h \in C^\infty(K, \R)^l_L$ for all $l \in L$.
This result will also be used in Proposition \ref{proposition:poisson_subalgebra_reduce_to_flat_subspace} to show that $\mathcal A(\Gamma, L)$ is a Poisson algebra.

To ensure that $K_L$ is a submanifold, we require that $\Gamma$ is connected and that there is a pair of a vertex and adjacent face that are not in $L \subseteq V \dot \cup F$.
Under these assumptions, we have:

\begin{lemma}[The Poisson bracket on $K_L$]~
  \label{lemma:locally_flat_poisson_submanifold}
  \begin{compactenum}
  \item
    The set $K_L \subseteq K$ is a submanifold.
  \item
    Let $g, g', h \in \bigcap_{l \in L} C^\infty(K, \R)^{l}_L$ with $g|_{K_L} = g'|_{K_L}$.
    The Poisson bracket satisfies:
    \[
      \left\{ g, h \right\}|_{K_L} = \left\{ g', h \right\}|_{K_L} \, .
    \]
  \end{compactenum}
\end{lemma}

\begin{proof}
  As a smooth manifold one has $K \cong G_-^{\x E} \x G_+^{\x E}$.
  We prove (i) by constructing an injective immersion $\phi: G_-^{\x E_-} \x G_+^{\x E_+} \to G_-^{\x E} \x G_+^{\x E}$ such that $K_L = \phi (G_-^{\x E_-} \x G_+^{\x E_+})$ for suitably chosen subsets $E_-, E_+ \subseteq E$.
  For this we parameterize $K_L$ by inductively implementing flatness at the vertices and faces in $L$.

  Let $v, f$ be a vertex and adjacent face that are not in $L$.
  With proper choice of orientation, there is an edge $e : v \to v'$ for some vertex $v'$ such that $f$ is left of $e$.
  Denote the face right of $e$ by $f'$.
  If $v' \in L$ ($f' \in L$), then for $\gamma \in K_L$  the value of $\Hol(f(e))$ ($\Hol(r(e))$)
  is uniquely determined by the values of $\pi_{e'}$ for the other edges $e'$ incident at $v'$ (adjacent to $f'$).
  Thus, flatness at $v'$ ($f'$) can be implemented by a parametrization of the $G_-$-component ($(G_+)$-component) of the copy of $G$ associated with $e$.
  Set $E_{done,1} := \left\{ e \right\}, V_{done, 1} := \left\{ v, v' \right\}$ and $F_{done, 1} := \left\{ f, f' \right\}$.

  Now suppose that for an $i \in \N$ we have sets $E_{done, i}, V_{done, i}, F_{done, i}$ such that for each $e' \in E_{done, i}$ the vertices $s(e'), t(e')$ are in $V_{done,i}$ and the faces left and right of $e'$ are in $F_{done, i}$.
  Suppose further that we can implement flatness for all vertices in $L \cap V_{done,i}$ and faces in $L \cap F_{done, i}$ by parametrizing the $G_-$- or $G_+$-components of some edges in $E_{done, i}$.
  These assumptions are true for $i=1$.

  Choose an edge $e \notin E_{done, i}$ such that, with proper choice of orientation of $e$, the next edge $e'$ in the cyclic ordering at $s(e)$ is in $E_{done, i}$.
  By assumption, we have $s(e'), t(e') \in V_{done, i}$ and the faces adjacent to $e'$ are in $F_{done, i}$.
  This implies that $s(e) \in V_{done, i}$ and that the face left of $e$ is in $F_{done, i}$.
  If $t(e) \in L \setminus V_{done,i}$, then no edges incident at $t(e)$ are in $E_{done,i}$.
  We can implement flatness at $t(e)$ by parametrizing $\Hol(f(e))$ in terms of the copies of $G$ associated with the other edges incident at $t(e)$.
  Likewise, if the face $f$ right of $e$ is in $L \setminus F_{done,i}$, no edges of $f$ are in $E_{done, i}$ and we obtain a parametrization of $\Hol(r(e))$.

  Set $E_{done,i+1} := E_{done, i} \cup \left\{ e \right\}, V_{done, i+1} := V_{done, i} \cup \left\{ t(e) \right\}$ and $F_{done, i+1} := F_{done, i} \cup \left\{ f \right\}$.
  Then for all $e' \in E_{done,i+1}$ one has $s(e), t(e) \in V_{done, i+1}$ and the faces left and right of $e'$ are in $F_{done, i+1}$.
  Furthermore, we can implement flatness for all vertices in $L \cap V_{done, i+1}$ and faces in $L \cap F_{done, i+1}$ by parametrizing the $G_-$- or $G_+$-components of the copies of $G$ associated to some edges in $E_{done, i+1}$.

  Because $\Gamma$ is connected, we can repeat this procedure until $E_{done, i+1} = E$.
  We obtain a parametrization of $K_L$ in terms of $G_-^{E_-} \x G_+^{E_+}$, where the set $E_-$ ($E_+$) consists of the edges that have not been used to implement flatness at a vertex (face).
  This proves (i).

  To prove (ii), note that for $\gamma \in K_L$ the tangent space $T_\gamma K_L$ is given by
  \[
    T_\gamma K_L = \left\{ X \in T_\gamma K \, | \, T\Hol^v \, X = T\Hol^f \, X = 0 \Forall v, f \in L \right\} \, .
  \]
  We show that for all $v \in V \cap L, f \in F \cap L$ and $h \in \bigcap_{l \in L} C^\infty(K, \R)^{l}_L$ the following equations hold for the Poisson bivector $w_K$ of $K$:
  \[
    (T \Hol^v \ox \, dh) \, w_K(\gamma) = 0 = (T \Hol^f \ox \, dh) \, w_K(\gamma) \quad \Forall \gamma \in K_L \, .
  \]
  This implies that $(\id \ox dh) \, w_K |_{K_L}$ is a vector field on $K_L$, so that $\left\{ g, h \right\}|_{K_L} = \left\{ g',h \right\}|_{K_L}$ follows from $g|_{K_L} = g'|_{K_L}$.

  Assume that $v \in V \cap L$ is a vertex without loops and that the edges $e_1 < \dots < e_m$ incident at $v$ are incoming.
  Set $d_i := \pi_{e_i}(\gamma)$ for $i = 1, \dots, m$.
  For $\gamma \in K_L$ the map $G_+ \to \R, \alpha \mapsto h(\alpha \rhd_v \gamma)$ is constant.
  Hence, we obtain
  \begin{equation}
    \label{eq:locally_flat_poisson_submanifold_proof0}
    0 = \langle d(h \circ (- \rhd_v \gamma)), \beta \rangle = \langle dh, - V (\beta) \rangle (\gamma)
  \end{equation}
  for all elements $\beta \in \g_+$, where $V (\beta) \in \Gamma(T K)$ is the vector field that generates the action $\rhd_v$.
  Define for $e \in E$ the inclusion map $\iota_e(\gamma) : G \to K$ by
  \begin{equation}
    \label{eq:inclusion_map}
    (\pi_{h} \0 \iota_e (\gamma)) (g) =
    \begin{cases}
      g & \text{if } h = e \\
      \pi_{h}(\gamma) & \text{otherwise.}
    \end{cases}
  \end{equation}
  By differentiating the map $(-\rhd_v \gamma)$ in Equation \eqref{eq:kitaev_vertex_action_easy} we obtain an explicit expression for $V(\beta)$
  \begin{equation}
    \label{eq:locally_flat_poisson_submanifold_proof1}
    V(\beta) (\gamma) = -\sum_{i =1}^m T (\iota_{e_i}(\gamma) \0 R_{d_i} \circ \pi_+ \circ R_{c_i(\gamma)}) \, \beta \, ,
  \end{equation}
  where $c_i(\gamma) = \Hol(p_{i-1}(v))(\gamma)$.
  Set $\tilde c_i (\gamma) := \Hol(p_{i-1}(v)^{-1} \0 p(v)) (\gamma)$, so that $c_i(\gamma) \tilde c_i(\gamma) = \Hol^v(\gamma) $.
  Consider the Poisson bivector  $w_{\He}$  from \eqref{eq:bivector_heisenberg_double}.
  The bivector of the product Poisson manifold $K =  G_\He ^{\x E}$ is given by $w_K(\gamma) = \sum_{e \in E} T(\iota_e(\gamma))^{\ox 2} \, w_{\He}  (\pi_e(\gamma))$.
  The map $\Hol^v : K \to G$ depends only on the copies of $ G_\He $ associated with the edges $e_1, \dots, e_m$, so that
  \[
    (T \Hol^v \ox \, dh) \, w_K (\gamma) = \sum_{i=1}^m (T\Hol^v \ox dh) \0 T\iota_{e_i}(\gamma)^{\ox 2} \, w_\He  (d_i) \, .
  \]
  As $\Hol^v(\gamma) = \pi_-(d_1) \cdots \pi_-(d_m)$, one has $\Hol^v \0 \, \iota_{e_i} = L_{c_i(\gamma)} \0 R_{\tilde c_{i+1}(\gamma)} \0 \pi_-$, which implies
  \[
    (T \Hol^v \ox \, dh) \, w_K (\gamma) = \sum_{i=1}^m T(L_{c_i(\gamma)} \circ R_{\tilde c_{i+1}(\gamma)} \circ \pi_-) \ox d(h \0 \iota_{e_i}(\gamma)) \, w_\He  (d_i) \, .
  \]
  Use Equation \eqref{eq:global_double_projections_r_matrix_i_2} from Lemma \ref{lemma:global_double_projections_r_matrix} (i) to obtain:
  \begin{equation*}
    (T \Hol^v \ox \, dh) \, w_K (\gamma) = - \sum_{i=1}^m \left( T(L_{c_i(\gamma)} \circ R_{\tilde c_{i+1}(\gamma)} \circ \pi_-) \ox d(h \0 \iota_{e_i}(\gamma)) \right) \0 TR_{d_i}^{\ox 2} \, r \, .
  \end{equation*}
  The identity $(R_{\tilde c_{i+1}(\gamma)} \0 \pi_- \0 R_{d_i})(x) = R_{\tilde c_i(\gamma)} (x)$ for $x \in G_-$ together with the fact that $r \in \g_- \ox \g_+$ implies
  \begin{align*}
    (T \Hol^v \ox \, dh) \, w_K (\gamma) =& - \sum_{i=1}^m \left( T(L_{c_i(\gamma)} \circ R_{\tilde c_{i}(\gamma)} ) \ox d(h \0 \iota_{e_i}(\gamma) \0 R_{d_i}) \right) \, r \\
    =& - \sum_{i=1}^m \left( (TR_{\Hol^v(\gamma)} \0 \Ad_{c_i(\gamma)}) \ox d(h \0 \iota_{e_i}(\gamma) \0 R_{d_i}) \right) \, r \, ,
  \end{align*}
  where we inserted the term $T(R_{c_i(\gamma)} \0 R_{c_i(\gamma)^{-1}})$ in the second equation and used $\Hol^v(\gamma) = c_i(\gamma) \tilde c_i(\gamma)$.
  Equation \eqref{eq:global_double_projections_r_matrix_ii_1} from Lemma \ref{lemma:global_double_projections_r_matrix} (ii) implies:
  \begin{align}
    (T \Hol^v \ox \, dh) \, w_K (\gamma) =& - \sum_{i=1}^m \left( TR_{\Hol^v(\gamma)} \ox d \left(h \0 \iota_{e_i}(\gamma) \0 R_{d_i} \0 \pi_+ \0 R_{c_i(\gamma)} \right) \right) \, r \nonumber \\
    \stackrel{\eqref{eq:locally_flat_poisson_submanifold_proof1}}{=} & \sum_{(r)} \langle dh, V(r_{(2)}) (\gamma) \rangle \, TR_{\Hol^v(\gamma)} \, r_{(1)} \stackrel{\eqref{eq:locally_flat_poisson_submanifold_proof0}}{=} 0 \label{eq:locally_flat_poisson_submanifold0} \, .
  \end{align}
  In the case with loops one proceeds in the same way, with the exceptions that loops occur twice in the holonomy $\Hol^v$ and one obtains two associated contributions to the vector field $V(\beta)$ in \eqref{eq:locally_flat_poisson_submanifold_proof1}.

  The proof of the statement $(T\Hol^f \ox dh) \, w_K (\gamma) = 0$ for $\gamma \in K_L, f \in F \cap L$ is analogous to that for vertex holonomies.
\end{proof}

\begin{proposition}
  \label{proposition:vertex_face_actions_are_poisson}
  The maps $\rhd_v : G_+ \x K \to K$ for $v \in V$ and $\rhd_f: G_- \x K \to K$ for $f \in F$ are Poisson.
\end{proposition}

\begin{proof}
  First let $v \in V$ be a vertex without loops and assume that the linearly ordered edges $e_1 < \dots < e_n$ are all incoming at $v$.
  Denote the product Poisson bivector of $G_+ \x K$ by $w_\pi$ and set $d_i := \pi_{e_i} (\gamma)$.
  To see that $\rhd_v : G_+ \x K \to K$ is Poisson, first note that the maps for the individual edges from Equation \eqref{eq:kitaev_vertex_action_easy}
  \[
    \pi_{e_i} \0 \rhd_v: G_+ \x K \to  G_\He  \qquad (\alpha, \gamma) \mapsto \pi_+(\alpha \, c_k (\gamma)) \, d_i  \quad \text{ with } \quad c_k(\gamma) = \Hol(p_{k-1}(v))(\gamma)
  \]
  are Poisson as a consequence of Lemma \ref{lemma:projectionLocallyPoisson} and Theorem \ref{theorem:heisenberg_double_inversion_poisson_actions} (ii).
  This implies
  \[
    T\pi_{e_i}^{\ox 2} \0 T(\rhd_v)^{\ox 2} \, w_\pi = w_\He  (\pi_{e_i} (\alpha \rhd_v \gamma)) = T\pi_{e_i}^{\ox 2} \, w_K (\alpha \rhd_v \gamma) \, ,
  \]
  where $w_K$ is the Poisson bivector of the product Poisson manifold $K =  G_\He ^{\x E}$.

  It remains to show that $(T\pi_{e_i} \ox T\pi_{e_j}) \0 T(\rhd_v)^{\ox 2} \, w_\pi = 0$ for $i \neq j$.
  Let $i < j$ and define
  \[
    c_{l,k} := \Hol(p_{l-1}(v))(\gamma)^{-1} \, \Hol(p_{k-1}(v))(\gamma)
  \]
  with the partial vertex path $p_k(v)$ from \eqref{eq:partial_vertex_path}.
  This is the value of the partial vertex holonomy from the $(k-1)$-th to the $l$-th edge end and satisfies $c_{1,k}(\gamma) = c_{k}(\gamma)$.
  We compute:
  \begin{align}
    & (T\pi_{e_i} \ox T\pi_{e_j}) \0 T(\rhd_v)^{\ox 2} \, w_\pi(\alpha, \gamma) = T(R_{d_i} \circ \pi_+ \circ R_{c_{1,i}}) \ox T(R_{d_j} \circ \pi_+ \circ R_{c_{1,j}}) \, w_{G_+} (\alpha) \nonumber \\
    \label{eq:kitaev_actions_poisson_proof0}
    & + \sum_{k=1}^{i-1} T(R_{d_i} \circ \pi_+ \circ L_{\alpha \, c_{1,k}} \circ R_{c_{k+1, i}} \circ \pi_-) \ox T(R_{d_j} \circ \pi_+ \circ L_{\alpha \, c_{1, k}} \circ R_{c_{k+1, j}} \circ \pi_-) \, w_\He  (d_k) \\
    & + TL_{\pi_+(\alpha \, c_{1, i})} \ox T(R_{d_j} \circ \pi_+ \circ L_{\alpha \, c_{1, i}} \circ R_{c_{i+1, j}} \circ \pi_-) \, w_\He   (d_i) \, . \nonumber
  \end{align}
  The first term here is obtained from the contribution of $G_+$ to $w_\pi (\alpha, d)$, while the other terms are associated with the edges $e_1, \dots, e_i$.
  As $G_+ \subseteq G$ is a Poisson-Lie subgroup of the quasi-triangular Poisson-Lie group $G$, we obtain the Poisson bivector $w_{G_+}$ of $G_+$ from Equation \eqref{eq:SklyaninBivector}.
  The $\Ad$-invariance of the symmetric component $r_s$ of the classical $r$-matrix of $G$ implies:
  \[
    w_{G_+}(\alpha) = w (\alpha) = (TL_{\alpha}^{\ox 2} - TR_{\alpha}^{\ox 2}) \, r_a = - (TL_{\alpha}^{\ox 2} - TR_{\alpha}^{\ox 2}) \, r_{21} \, .
  \]
  Lemma \ref{lemma:global_double_projections_properties} (ii) implies that $T(\pi_+ \circ R_{g}) \, y = 0 \Forall y \in \g_-, g \in G$.
  Applying this to the first term of \eqref{eq:kitaev_actions_poisson_proof0} and Equation \eqref{eq:global_double_projections_r_matrix_i_2} from Lemma \ref{lemma:global_double_projections_r_matrix} (i) to the other terms yields:
  \begin{align*}
    & (T\pi_{e_i} \ox T\pi_{e_j}) \0 T(\rhd_v)^{\ox 2} \, w_\pi(\alpha, \gamma) = \; - T(R_{d_i} \circ \pi_+ \circ R_{c_{1,i}}) \ox T(R_{d_j} \circ \pi_+ \circ R_{c_{1,j}}) \0 TL_{\alpha}^{\ox 2} \, r_{21} \\
    & + \sum_{k=1}^{i-1} T(R_{d_i} \circ \pi_+ \circ L_{\alpha \, c_{1,k}} \circ R_{c_{k+1, i}} \circ \pi_-) \ox T(R_{d_j} \circ \pi_+ \circ L_{\alpha c_{1, k}} \circ R_{c_{k+1, j}} \circ \pi_-) \0 TR_{d_k}^{\ox 2} \, r_{21} \\
    & + TL_{\pi_+(\alpha \, c_{1, i})} \ox T(R_{d_j} \circ \pi_+ \circ L_{\alpha \, c_{1, i}} \circ R_{c_{i+1, j}} \circ \pi_-) \0 TR_{d_i}^{\ox 2} \, r_{21} \, .
  \end{align*}
  Using $r_{21} \in \g_+ \ox \g_-$ and Lemma \ref{lemma:global_double_projections_properties} we can simplify this expression to:
  \begin{equation}
    \label{eq:kitaev_actions_poisson_proof1}
    \begin{split}
      & (T\pi_{e_i} \ox T\pi_{e_j}) \0 T(\rhd_v)^{\ox 2} \, w_\pi(\alpha, \gamma) = (TR_{d_i} \ox TR_{d_j}) \circ T(\pi_+ \circ L_\alpha)^{\ox 2} \\
      & \quad \Big[ - TR_{c_{1, i}} \ox TR_{c_{1, j}} + \sum_{k=1}^{i-1} T(L_{c_{1, k}} \circ R_{c_{k+1, i}} \circ \pi_- \circ R_{\pi_-(d_k)}) \ox T(L_{c_{1, k}} \circ R_{c_{k, j}}) \\
      & \quad \:\, + TL_{c_{1, i}} \ox T(L_{c_{1, i}} \circ R_{c_{i, j}}) \Big] \, r_{21} \, .
    \end{split}
  \end{equation}
  Let $k \in \left\{ 1, \dots, i-1 \right\}$.
  The equation $g \, \pi_-(d_k) = \pi_-(g \, \pi_-(d_k)) \, \pi_+(g \, \pi_-(d_k)) \Forall g \in G$ implies
  \begin{equation*}
    TR_{\pi_-(d_k)}^{\ox 2} \, r_{21} = \; T(\pi_- \circ R_{\pi_-(d_k)}) \ox TR_{\pi_-(d_k)} \, r_{21} + T( L_{\pi_-(d_k)} \circ \pi_+ \circ R_{\pi_-(d_k)}) \ox TR_{\pi_-(d_k)} \, r_{21} \, ,
  \end{equation*}
  and applying Equation \eqref{eq:global_double_projections_r_matrix_ii_1} from Lemma \ref{lemma:global_double_projections_r_matrix} (ii) to the second term yields:
  \begin{align*}
  & TR_{\pi_-(d_k)}^{\ox 2} \, r_{21} = \; T(\pi_- \circ R_{\pi_-(d_k)}) \ox TR_{\pi_-(d_k)} \, r_{21} + T L_{\pi_-(d_k)} \ox (TR_{\pi_-(d_k)} \0 \Ad_{\pi_-(d_k)}) \, r_{21} \\
    =& \; T(\pi_- \circ R_{\pi_-(d_k)}) \ox TR_{\pi_-(d_k)} \, r_{21} + T L_{\pi_-(d_k)}^{\ox 2} \, r_{21} \, .
  \end{align*}
  Apply the map $TL_{c_{1, k}}^{\ox 2} \0 (TR_{c_{k+1, i}} \ox TR_{c_{k+1, j}})$ to this equation to obtain
  \begin{align*}
    & TL_{c_{1, k}}^{\ox 2} \0 (TR_{c_{k, i}} \ox TR_{c_{k, j}}) \, r_{21} = \; T(L_{c_{1, k}} \0 R_{c_{k+1, i}} \0 \pi_- \circ R_{\pi_-(d_k)}) \ox T(L_{c_{1, k}} \0 R_{c_{k, j}}) \, r_{21} \\
    & + TL_{c_{1, k+1}}^{\ox 2} \0 (TR_{c_{k+1, i}} \ox TR_{c_{k+1, j}}) \, r_{21} \, .
  \end{align*}
  Thus, for the sum in \eqref{eq:kitaev_actions_poisson_proof1} we obtain
  \begin{align*}
    & \sum_{k=1}^{i-1} T(L_{c_{1, k}} \circ R_{c_{k+1, i}} \circ \pi_- \circ R_{\pi_-(d_k)}) \ox T(L_{c_{1, k}} \circ R_{c_{k, j}}) \, r_{21} \\
    =& \; \sum_{k=1}^{i-1} TL_{c_{1, k}}^{\ox 2} \0 (TR_{c_{k, i}} \ox TR_{c_{k, j}}) \, r_{21} - \sum_{k=1}^{i-1} TL_{c_{1, k+1}}^{\ox 2} \0 (TR_{c_{k+1, i}} \ox TR_{c_{k+1, j}}) \, r_{21} \\
    =& \; TR_{c_{1, i}} \ox TR_{c_{1, j}} \, r_{21} - TL_{c_{1, i}} \ox T(L_{c_{1, i}} \0 R_{c_{i, j}}) \, r_{21} \, .
  \end{align*}
  Insert this for the sum on the right hand side of \eqref{eq:kitaev_actions_poisson_proof1} to obtain
  \begin{equation*}
    (T\pi_{e_i} \ox T\pi_{e_j}) \0 T(\rhd_v)^{\ox 2} \, w_\pi(\alpha, \gamma) = 0 \, .
  \end{equation*}

  Now we relax the assumption that $v$ is a vertex without loops.
  So let $v \in V$ be a general vertex and $\Gamma'$ the graph obtained from introducing a bivalent vertex on every loop at $v$.
  Consider the Poisson map $\psi: K' \to K$ obtained from Lemma \ref{lemma:kitaev_gluing_trafos} that corresponds to gluing the split loops back together.
  Denote the set of newly introduced vertices by $L' \subseteq V'$.
  For $\gamma' \in K'_{L'}, \alpha \in G_+$ and $f_1, f_2 \in C^\infty(K, \R)$ we compute
  \begin{align}
    & \left\{ f_1 \0 \rhd_v, f_2 \0 \rhd_v \right\}_{G_+ \x K} (\alpha, \psi(\gamma')) \nonumber \\
    = & \left\{ f_1 \0 (- \rhd_v \psi(\gamma')), f_2 \0 (- \rhd_v \psi(\gamma')) \right\}_{G_+} (\alpha) + \left\{ f_1 \0 (\alpha \rhd_v -), f_2 \0 (\alpha \rhd_v -) \right\}_{K} (\psi(\gamma')) \nonumber \\
    =& \left\{ f_1 \0  (- \rhd_v \psi(\gamma')), f_2 \0  (- \rhd_v \psi(\gamma')) \right\}_{G_+} (\alpha) \label{eq:kitaev_actions_properties_proof0} \\
    & + \left\{ f_1 \0  (\alpha \rhd_v -) \0 \psi , f_2 \0  (\alpha \rhd_v -) \0 \psi \right\}_{K'} (\gamma') \nonumber \, ,
  \end{align}
  where we used the Poisson property of $\psi$ in the second equation.
  From Lemma \ref{lemma:kitaev_gluing_trafos} (iii) one obtains $\psi(\hat\alpha \rhd_v \gamma') = \hat\alpha \rhd_v \psi(\gamma') \Forall \hat\alpha \in G_+$ because $\gamma' \in K'_{L'}$.
  Apply this to \eqref{eq:kitaev_actions_properties_proof0} to obtain
  \begin{align}
    & \left\{ f_1 \0 \rhd_v, f_2 \0 \rhd_v \right\}_{G_+ \x K} (\alpha, \psi(\gamma')) = \left\{ f_1 \0  \psi \0 (- \rhd_v \gamma'), f_2 \0 \psi \0 (- \rhd_v \gamma') \right\}_{G_+} (\alpha) \nonumber \\
    &+ \left\{ f_1 \0  (\alpha \rhd_v -) \0 \psi , f_2 \0  (\alpha \rhd_v -) \0 \psi \right\}_{K'} (\gamma') \, . \label{eq:kitaev_actions_properties_proof1}
  \end{align}

  Let $g_i := f_i \0 (\alpha \rhd_v -) \0 \psi $ and $g_i' := f_i \0 \psi \0 (\alpha \rhd_v -)$ for $i =1,2$.
  By Lemma \ref{lemma:kitaev_gluing_trafos} (iii), the restrictions $g_i|_{K'_{L'}}$ and $g'_i|_{K'_{L'}}$ coincide.
  As $G_+ \rhd_{v'} K'_{L'} = K'_{L'}$ by Lemma \ref{lemma:kitaev_flat_subspace_stable}, the maps $g_i, g'_i$ also coincide on elements of the form $\hat\alpha \rhd_{v'} \hat \gamma'$ for $\hat\alpha \in G_+, \hat \gamma' \in K'_{L'}$.
  The map $\psi$ is invariant under $\rhd_{v'}$ for $v' \in L'$ by Lemma \ref{lemma:kitaev_gluing_trafos} (ii).
  Hence, both $g_i$ and $g'_i$ are elements of $\bigcap_{l' \in L'} C^\infty(K', \R)^{l'}_{L'}$.
  We can therefore apply Lemma \ref{lemma:locally_flat_poisson_submanifold} (ii) to substitute $g_i$ for $g'_i$ in \eqref{eq:kitaev_actions_properties_proof1}.
  We obtain
  \begin{align*}
    & \left\{ f_1 \0 \rhd_v, f_2 \0 \rhd_v \right\}_{G_+ \x K} (\alpha, \psi(\gamma')) \\
    =& \left\{ f_1 \0 \psi \0 (- \rhd_v \gamma'), f_2 \0 \psi \0 (- \rhd_v \gamma') \right\}_{G_+} (\alpha) + \left\{ f_1 \0 \psi \0 (\alpha \rhd_v -) , f_2 \0 \psi \0 (\alpha \rhd_v -) \right\}_{K'} (\gamma') \\
    =& \left\{ f_1 \0 \psi \0 \rhd_v, f_2 \0 \psi \0 \rhd_v \right\}_{G_+ \x K'} (\alpha, \gamma') = \left\{ f_1 , f_2 \right\}_{K} \0 \psi \0 \rhd_v \, (\alpha , \gamma') = \left\{ f_1 , f_2 \right\}_{K} \0 \rhd_v \, (\alpha ,\psi(\gamma')) \, ,
  \end{align*}
  where we used that $\rhd_v: G_+ \x K' \to K'$ and $\psi: K' \to K$ are Poisson in the third equation and Lemma \ref{lemma:kitaev_gluing_trafos} (iii) in the fourth.
  Because $\psi|_{K'_{L'}} : K'_{L'} \to K$ is surjective by Lemma \ref{lemma:kitaev_gluing_trafos} (v), this implies that $\rhd_v : G_+ \x K \to K$ is Poisson.
  The proof for $\rhd_f : G_- \x K \to K$ is analogous.
\end{proof}

In particular, Proposition \ref{proposition:vertex_face_actions_are_poisson} shows that $K$ is a Poisson $G_+$-space with respect to every vertex action and a Poisson $G_-$-space for every face action.
The maps $\rhd: G_+ \x G_- \to G_-, (\alpha, x) \mapsto \pi_-(\alpha x)$ and $\rhd' : G_- \x (G_+, -w_{G_+}) \to (G_+, -w_{G_+}), (x, \alpha) \mapsto \pi_+(\alpha x^{-1})$ are the dressing actions of $G_+$ and $G_-$ on each other \cite[Theorem 3.14]{luweinstein1990}.
These equip $G_-$ with the structure of a Poisson $G_+$-space and vice versa.
Vertex and face holonomies intertwine these actions:

\begin{proposition}[Holonomies and dressing transformations]
  For all $v \in V, f \in F$ the maps $\Hol^v : K \to G_-$ and $\Hol^f: K \to (G_+, -w_{G_+})$  from Definition \ref{def:vertex_face_paths_and_holonomies} are homomorphisms  of Poisson $G_+$-spaces and Poisson $G_-$-spaces, respectively.
\end{proposition}

\begin{proof}
  The projections $\pi_- :   G_\He  \to G_-$ and $\pi_+ :   G_\He  \to (G_+, -w_{G_+})$ are Poisson by Lemma \ref{lemma:projectionLocallyPoisson}.
  The inversion map $  G_\He  \to   G_\He $ is Poisson by Theorem \ref{theorem:heisenberg_double_inversion_poisson_actions} (ii).
  This implies that the maps $\Hol^v : K \to G_-, \, \Hol^f: K \to (G_+, -w_{G_+})$ are Poisson.
  It remains to show that they intertwine the respective group actions, that is:
  \begin{equation*}
    \Hol^v(\alpha \rhd_v \gamma) =  \alpha \rhd \Hol^v(\gamma) \qquad \Hol^f(x \rhd_f \gamma) =  x \rhd' \Hol^f(\gamma) \qquad \Forall x \in G_-, \alpha \in G_+, \gamma \in K \, .
  \end{equation*}
  We show this for the vertex holonomy $\Hol^v$ as the proof for $\Hol^f$ is analogous.
  Assume that $(v, f(v))$ is a site.
  Then this follows from Equation \eqref{eq:holonomy_and_action_for_a_site} and Lemma \ref{lemma:global_double_projections_properties} because $\Hol^v = \pi_- \0 \Hol^{(v,f(v))}$.
  Otherwise cyclically permute the ordering at $f(v)$ so that $(v, f(v))$ becomes a site, and use Equation \eqref{eq:holonomy_and_action_for_a_site} with respect to the new ordering.
\end{proof}

The vertex and face actions $\rhd_v$ and $\rhd_f$ can be viewed as the Poisson-Lie counterparts of the $H$- and $H^*$-module algebra structures for quantum Kitaev models from Theorem \ref{theorem:quantum_kitaev_module_algebra}, which are defined by vertex and face operators.
For each site the latter combine into a $D(H)$-module algebra structure.
We now show that a similar statement holds for the Poisson actions associated with a site.
This yields the following counterpart of Lemma \ref{lemma:quantum_kitaev_commutation_relations}:

\begin{proposition}[Commutation relations for vertex and face actions]~
  \label{proposition:kitaev_actions_properties}
  \begin{compactenum}
  \item
    For two vertices $v \neq v' \in V$ (two faces $f \neq f' \in F$) the actions $\rhd_v, \rhd_{v'}$ ($\rhd_f, \rhd_{f'}$) commute.
  \item
    For $v \in V, f \in F$ with $v(f) \neq v$ and $f(v) \neq f$ the actions $\rhd_v, \rhd_f$ commute.
  \item
    For each site $(v,f)$ we obtain a Poisson action
    \begin{equation}
      \label{eq:site_action}
      \rhd_{(v,f)} : G \x K \to K \qquad (g, \gamma) \mapsto \pi_-(g) \rhd_f (\pi_+(g) \rhd_v \gamma) \, .
    \end{equation}
  \end{compactenum}
\end{proposition}

\begin{proof}
  We only prove (iii) as the other statements are shown similarly.
  First note that the map $(\pi_+ \x \pi_-) \0 \Delta : G \to G_+ \x G_-$ is Poisson, where $\Delta: G \to G \x G$ is the diagonal map.
  This is shown by a computation using Lemma \ref{lemma:global_double_projections_properties} and the fact that the maps $\pi_\pm : G \to G_\pm$ are Poisson (Lemma \ref{lemma:projectionLocallyPoisson}).
  This implies that the map $\rhd_{(v,f)} : G \x K \to K$ is Poisson.

  We show that $\rhd_{(v,f)}$ is a group action, that is
  \begin{equation}
    \label{eq:kitaev_actions_properties_site_action_proof0}
    g_1 \rhd_{(v,f)} (g_2 \rhd_{(v,f)} \gamma) = (g_1 g_2) \rhd_{(v,f)} \gamma \qquad \Forall g_1, g_2 \in G, \gamma \in K \, .
  \end{equation}
  Because $g_i = \pi_-(g_i) \pi_+(g_i), i =1,2$, it suffices to show this for $g_i \in G_+ \cup G_-$.
  In the cases $g_1, g_2 \in G_-$ and $g_1, g_2 \in G_+$ Equation \eqref{eq:kitaev_actions_properties_site_action_proof0} follows from the fact that $\rhd_v, \rhd_f$ are group actions.
  For $g_1 \in G_-, g_2 \in G_+$ it follows from the definition of $\rhd_{(v,f)}$ in \eqref{eq:site_action}.
  It thus remains to show
  \begin{equation}
    \label{eq:kitaev_actions_properties_site_action_proof1}
    \alpha \rhd_v  (x \rhd_f \gamma) = \pi_-(\alpha x) \rhd_f (\pi_+(\alpha x) \rhd_v \gamma)
    \qquad \Forall x \in G_-, \alpha \in G_+, \gamma \in K \, .
  \end{equation}
  If there are no loops at $v$ and at most one edge side of any edge belongs to $f$, then Equation \eqref{eq:kitaev_actions_properties_site_action_proof1} follows from a direct computation using the properties of $\pi_\pm$ from Lemma \ref{lemma:global_double_projections_properties}.
  In the general case, transform the graph $\Gamma$ by adding a bivalent vertex on each loop of $v$ and doubling every edge whose left and right face is $f$.
  Denote by $(K', \Gamma')$ the Poisson-Kitaev model for the transformed graph $\Gamma'$ and by $L' \subseteq V' \dot \cup F'$ the set of vertices and edges created by this procedure.
  From Lemma \ref{lemma:kitaev_gluing_trafos} we obtain the map $\psi: K' \to K$ that corresponds to gluing back together the edges that we split into two.
  By the surjectivity of $\psi$ (Lemma \ref{lemma:kitaev_gluing_trafos} (v)), for $\gamma \in K$ there is a $\gamma' \in K'_{L'}$ with $\gamma = \psi(\gamma')$.
  Equation \eqref{eq:kitaev_actions_properties_site_action_proof1} holds for elements of $K'$ as $\Gamma'$ satisfies our previous assumptions.
  We compute:
  \begin{align*}
    & \alpha \rhd_v  (x \rhd_f \gamma) = \alpha \rhd_v ( x \rhd_f \psi(\gamma')) \stackrel{\eqref{eq:kitaev_gluing_trafos_compatible_with_actions}}= \psi (\alpha \rhd_v (x \rhd_f \gamma')) \stackrel{\eqref{eq:kitaev_actions_properties_site_action_proof1}}= \psi(\pi_-(\alpha x) \rhd_f (\pi_+(\alpha x) \rhd_v \gamma')) \\
    \stackrel{\eqref{eq:kitaev_gluing_trafos_compatible_with_actions}}=&  \pi_-(\alpha x) \rhd_f (\pi_+(\alpha x) \rhd_v \psi(\gamma')) = \pi_-(\alpha x) \rhd_f (\pi_+(\alpha x) \rhd_v \gamma) \, .
  \end{align*}
  While applying Equation \eqref{eq:kitaev_gluing_trafos_compatible_with_actions}, we used $G_+ \rhd_v K'_{L'} = G_- \rhd_f K'_{L'} = K'_{L'}$, which follows from Lemma \ref{lemma:kitaev_flat_subspace_stable} because $v(f) \notin L'$ and $f(v) \notin L'$.
\end{proof}

We show that the set $\mathcal A(\Gamma, L)$ of invariant functions on $K_L$ from Definition \ref{def:functions_on_flat_elements} (ii) inherits a Poisson bracket from $C^\infty(K, \R)$.
For this we use Lemma \ref{lemma:locally_flat_poisson_submanifold}, which requires that the graph $\Gamma$ is connected and that there is a vertex and adjacent face not contained in $L$.
We require the vertices and faces without a flatness condition to be paired into sites.

\begin{proposition}[Poisson algebra structure on $\mathcal A(\Gamma, L)$]
  \label{proposition:poisson_subalgebra_reduce_to_flat_subspace}
  Let $\Gamma$ be a connected doubly-ciliated ribbon graph with $n \geq 1$ sites $(v_1, f_1), \dots, (v_n, f_n)$ and 
  $L = (V \setminus \left\{ v_1, \dots, v_n \right\}) \dot \cup (F \setminus \left\{ f_1, \dots, f_n \right\})$.
  \begin{compactenum}
  \item
    The set $C^\infty(K, \R)^{inv}_L$ is a Poisson subalgebra of $C^\infty(K, \R)$.
  \item
    The Poisson bracket on $C^\infty(K, \R)^{inv}_L$ induces a Poisson bracket on the quotient $\mathcal A(\Gamma, L)$.
  \end{compactenum}
\end{proposition}

\begin{proof}
  Let $g_1, g_2 \in C^\infty(K, \R)^{inv}_L$ and $v \in V$.
  For $\alpha \in G_+, \gamma \in K_L$ we obtain from Proposition \ref{proposition:vertex_face_actions_are_poisson}
  \begin{align}
    \left\{ g_1, g_2 \right\}_K (\alpha \rhd_v \gamma) =& \left\{ g_1 \circ \rhd_v, g_2 \circ \rhd_v \right\}_{G_+ \x K} (\alpha, \gamma) \nonumber \\
    =& \left\{ g_1 \circ (\alpha \rhd_v -), g_2 \circ (\alpha \rhd_v -) \right\}_K (\gamma) \nonumber \\
    &+ \left\{ g_1 \circ (- \rhd_v \gamma), g_2 \circ (- \rhd_v \gamma) \right\}_{G_+} (\alpha) \nonumber \\
    =& \left\{ g_1 \circ (\alpha \rhd_v -), g_2 \circ (\alpha \rhd_v -) \right\}_K (\gamma) \, ,
    \label{eq:poisson_subalgebra_reduce_to_flat_subspace_proof0}
  \end{align}
  where we used that $g_i \0 (- \rhd_v \gamma)$ is constant in the third equation.
  For $i=1,2$ the identity $g_i \0 (\alpha \rhd_v -)|_{K_L} = g_i|_{K_L}$ holds as $g_i \in C^\infty(K, \R)^{inv}_L$.
  Lemma \ref{lemma:kitaev_flat_subspace_stable} implies that $G_+ \rhd_{v'} K_L = K_L$  for all $v' \in V$ because either $v' \in L$ or $v' = v_k$ for some $k \in \left\{ 1, \dots, n \right\}$, so that $f(v') = f_k \notin L$.
  Therefore, the maps $g_i \0 (\alpha \rhd_v -)$ and $g_i$ coincide on all elements of the form $\alpha \rhd_{v'} \gamma$ with $ \alpha \in G_+, \gamma \in K_L$.
  Likewise, they coincide on elements of the form $x \rhd_{f'} \gamma$ with $f' \in F, x \in G_-, \gamma \in K_L$.
  This implies $g_i \0 (\alpha \rhd_v -) \in C^\infty(K, \R)^{inv}_L$.
  We can thus apply Lemma \ref{lemma:locally_flat_poisson_submanifold} (ii) to \eqref{eq:poisson_subalgebra_reduce_to_flat_subspace_proof0} and obtain
  \begin{equation*}
    \left\{ g_1, g_2 \right\}_K (\alpha \rhd_v \gamma) = \left\{ g_1 , g_2 \right\}_K (\gamma) \, ,
  \end{equation*}
  so that $\left\{ g_1, g_2 \right\}_K \in C^\infty(K, \R)^{v}_L$.
  The proof for $\left\{ g_1, g_2 \right\}_K \in C^\infty(K, \R)^f_L$ for $f \in F$ is analogous.
  We conclude $\left\{ g_1, g_2 \right\} \in C^\infty(K, \R)^{inv}_L$, so that $C^\infty(K, \R)^{inv}_L$ is a Poisson subalgebra.

  To prove Statement (ii), it remains to show that the induced Poisson bracket on $\mathcal A(\Gamma, L)$ is well-defined.
  This follows from Lemma \ref{lemma:locally_flat_poisson_submanifold} (ii).
\end{proof}

\subsection{Vertex and face operators}
\label{subsection:vertex_face_operators}
We show an analogue of the commutation relations for vertex and face operators from Lemma \ref{lemma:quantum_kitaev_commutation_relations} for their Poisson counterparts from Definition \ref{def:vertex_face_operators}.
Then we prove that the derivation $\left\{ g, - \right\}$ for a vertex or face operator $g$ is related to the vector fields that generate the respective vertex or face \emph{action}.

First we need two technical lemmas.
Consider two different edges $e_1, e_2$ incident at a vertex $v$ such that $t(e_1) = v = s(e_2)$ and $i_s(e_2)$ comes directly after $i_t(e_1)$ in the ordering of the edge ends at $v$.
(If $e_1, e_2$ are no loops, this simply means that $e_1$ is incoming, $e_2$ outgoing and that $e_2$ is the next edge after $e_1$.)
Consider the paths  $p_1 := r(e_2) \0 r(e_1) ,\; p_2 := f(e_1) \0 b(e_2)^{-1}$. 
These are segments of a face path and the vertex path $p(v)$ (see Definition \ref{def:vertex_face_paths_and_holonomies}) that meet at $v$.

\begin{lemma}
  \label{lemma:kitaev_reidemeister_ii}
  For these paths and all $f_1, f_2 \in C^\infty(G)$ the following equation holds:
  \begin{equation}
    \label{eq:kitaev_reidemeister_ii}
    \left\{ f_1 \0 \Hol(p_1), f_2 \0 \Hol(p_2) \right\} = 0
  \end{equation}
\end{lemma}

\begin{proof}
  For $\gamma \in K$ set $d_i := \pi_{e_i}(\gamma), i = 1,2$.
  The holonomies along $p_1$ and $p_2$ are given by (see Equation \eqref{eq:kitaev_holonomy_functor}):
  \[
    \Hol(p_1)(\gamma) = \pi_+(d_2) \pi_+(d_1) \, , \qquad \Hol(p_2)(\gamma) = \pi_-(d_1) \pi_-(d_2^{-1}) \, .
  \]
  Using the inclusion map $\iota_e(\gamma) : G \to K$  from Equation \eqref{eq:inclusion_map}, the Poisson bivector of the product Poisson manifold $K= G_\He ^{\x E}$ can be written as $w_K (\gamma) = \sum_{e \in E} T\iota_e(\gamma)^{\ox 2}  \, w_\He  (\pi_e (\gamma))$ with the bivector $w_\He$ from \eqref{eq:bivector_heisenberg_double}.
  In the Poisson bracket \eqref{eq:kitaev_reidemeister_ii} only the terms associated with $e_1$ and $e_2$ contribute.
  We compute these terms.
  Let $\eta: G \to G$ be the inversion map.
  We obtain
  \begin{align*}
    & (T\Hol(p_1) \ox T\Hol(p_2)) \, w_K (\gamma) = (T(L_{\pi_+(d_2)} \0 \pi_+) \ox T(R_{\pi_-(d_2^{-1})} \0 \pi_-)) \, w_\He  (d_1) \\
    & + (T(R_{\pi_+(d_1)} \0 \pi_+) \ox T(L_{\pi_-(d_1)} \0 \pi_- \0 \eta)) \, w_\He  (d_2) \, ,
  \end{align*}
  where the first and second term is obtained from the contribution to $w_K$ associated with $e_1$ and $e_2$, respectively.
  Apply Equation \eqref{eq:global_double_projections_r_matrix_i_1} from Lemma \ref{lemma:global_double_projections_r_matrix} (i) to obtain:

  \begin{align*}
    & (T\Hol(p_1) \ox T\Hol(p_2)) \, w_K (\gamma) = (T(L_{\pi_+(d_2)} \0 \pi_+) \ox T(R_{\pi_-(d_2^{-1})} \0 \pi_-)) \, (TL_{d_1}^{\ox 2} \, r_{21}) \\
    & + (T(R_{\pi_+(d_1)} \0 \pi_+) \ox T(L_{\pi_-(d_1)} \0 \pi_- \0 \eta)) \, (TL_{d_2}^{\ox 2} \, r_{21}) \, .
    \end{align*}
    The computation rules for $\pi_\pm$ from Lemma \ref{lemma:global_double_projections_properties} together with the fact $r \in \g_- \ox \g_+$ imply:
    \begin{align*}
    & (T\Hol(p_1) \ox T\Hol(p_2)) \, w_K (\gamma) = \big[ T(L_{\pi_+(d_2)} \0 L_{\pi_+(d_1)}) \ox T(R_{\pi_-(d_2^{-1})} \0 L_{\pi_-(d_1)} \0 \pi_- \0 L_{\pi_+(d_1)}) \\
    & - T(R_{\pi_+(d_1)} \0 L_{\pi_+(d_2)}) \ox T(L_{\pi_-(d_1)} \0 R_{\pi_-(d_2^{-1})}) \big] \, r_{21} \\
    \stackrel{\eqref{eq:global_double_projections_r_matrix_ii_2}}=& \Big[ T(L_{\pi_+(d_2)} \0 L_{\pi_+(d_1)} \0 \Ad_{\pi_+(d_1)^{-1}}) \ox T(R_{\pi_-(d_2^{-1})} \0 L_{\pi_-(d_1)}) \\
    & - T(R_{\pi_+(d_1)} \0 L_{\pi_+(d_2)}) \ox T(L_{\pi_-(d_1)} \0 R_{\pi_-(d_2^{-1})}) \Big] \, r_{21} = 0 \, .
  \end{align*}
\end{proof}

\begin{lemma}
  \label{lemma:omegaHGbothsideszero}
  For any edge $e \in E$ and $f_1, f_2 \in C^\infty(G, \R)$ the holonomies along opposite edge ends or edge sides commute:
  \[
    \left\{ f_1 \0 \Hol(f(e)), f_2 \0 \Hol(b(e)) \right\} = \left\{ f_1 \0 \Hol(r(e)), f_2 \0 \Hol(l(e)) \right\} = 0 \, .
  \]
\end{lemma}

\begin{proof}
  This follows from $T\pi_- \ox T(\pi_- \0 \eta) \, w_\He  = T\pi_+ \ox T(\pi_+ \0 \eta) \, w_\He  = 0$, which can be shown by a direct computation using Lemma \ref{lemma:global_double_projections_properties} and the fact that $r \in \g_- \ox \g_+$.
\end{proof}

Now we are ready to prove an analogue of Lemma \ref{lemma:quantum_kitaev_commutation_relations} for Poisson-Kitaev models:

\begin{proposition}[Commutation relations for vertex and face operators] Let $g_1, g_2 \in C^\infty(G, \R)$.
  \label{proposition:kitaev_commutation_relations}
  \begin{enumerate}[nolistsep,noitemsep,label=(\roman*)]
    \item
      Vertex (face) operators for distinct vertices $v_1 \neq v_2$ (faces $f_1 \neq f_2$) Poisson commute:
      \[
        \left\{ g_1 \0 \Hol^{v_1}, g_2 \0 \Hol^{v_2} \right\} = 0 \qquad \left( \left\{ g_1 \0 \Hol^{f_1}, g_2 \0 \Hol^{f_2} \right\} = 0 \right) \, .
      \]
    \item
      For a vertex $v$ and face $f$ with $v(f) \neq v$ and $f(v) \neq f$ the associated operators commute:
      \[
        \left\{ g_1 \0 \Hol^{v}, g_2 \0 \Hol^{f} \right\} = 0 \, .
      \]
    \item
      Let $v \in V$ and $f \in F$ form a site $(v,f)$.
      The combined face and vertex holonomy from Definition \ref{def:vertex_face_paths_and_holonomies} (iii) induces an injective homomorphism of Poisson algebras
      \begin{equation}
        \label{eq:kitaev_site_operator_poisson_map}
        \Hol^{(v,f)*} :C^\infty( (G, w_{G^*}), \R) \to C^\infty(K, \R) \qquad g \mapsto g \0 \Hol^{(v,f)}
      \end{equation}
       with regard to the Poisson bivector $w_{G^*}$ on $G$ from Equation \eqref{eq:bivector_dual_poisson_lie_group}.
  \end{enumerate}
\end{proposition}

\begin{proof}
  We prove the statement for vertex operators in (i).
  Note that only those copies of $G$ associated with edges that connect $v_1$ with $v_2$ contribute to the Poisson bracket $\left\{ g_1 \0 \Hol^{v_1}, g_2 \0 \Hol^{v_2} \right\}$.
  For any such edge $e$ the map $g_1 \0 \Hol^{v_1}$ depends on $\Hol(f(e))$ and $g_2 \0 \Hol^{v_2}$ on $\Hol(b(e))$ or vice versa.
  Lemma \ref{lemma:omegaHGbothsideszero} implies that the contribution to $\left\{ g_1 \0 \Hol^{v_1}, g_2 \0 \Hol^{v_2} \right\}$ associated with $e$ vanishes.
  The proof for face operators is analogous.

  To see (ii), first assume that there are no loops at $v$ and that every edge is traversed at most once by the face path $p(f)$.
  Then edges with non-vanishing contribution to the Poisson bracket can be grouped into pairs $e_1 \neq e_2 \in E$ of consecutive edges on the face path $p(f)$. 
  With appropriate choice of orientation we have $t(e_1) = s(e_2) = v$.
  Then one argument of the Poisson bracket is a function of $\Hol (r_{e_2} \0 r_{e_1})$ and the other of $\Hol (f_{e_1} \0 b_{e_2}^{-1})$.
  By Lemma \ref{lemma:kitaev_reidemeister_ii} any such contribution vanishes.

  In the general case, we transform the graph $\Gamma$ by adding a bivalent vertex on every loop at $v$ and doubling any edge that occurs twice in the face path $p(f)$.
  Denote by $(K', \Gamma')$ the Poisson-Kitaev model for the transformed graph and the set of new vertices and faces by $L' \subseteq V' \dot \cup F'$.
  From Lemma \ref{lemma:kitaev_gluing_trafos} we obtain a Poisson map $\psi: K' \to K$ that corresponds to gluing the split edges back together.
  The map $\psi$ is invariant under actions at $l \in L'$ and satisfies by Equation \eqref{eq:kitaev_gluing_trafos_compatible_with_holonomies}:
  \[
    \Hol^v \0 \psi \big|_{K'_{L'}} = \Hol'^v \big|_{K'_{L'}}  \qquad \Hol^f \0 \psi \big|_{K'_{L'}} = \Hol'^f \big|_{K'_{L'}} \, .
  \]
  By Lemma \ref{lemma:kitaev_flat_subspace_stable}, the subspace $K'_{L'}$ is stable under the actions $\rhd_{l'}$ for $l' \in L'$ which implies that $\Hol'^v, \Hol'^f \in \bigcap_{l' \in L'} C^\infty(K, \R)^{l'}_{L'}$.
  Use the Poisson property of $\psi$ and Lemma \ref{lemma:locally_flat_poisson_submanifold} (ii) to compute
  \begin{align*}
    & \left\{ g_1 \0 \Hol^v, g_2 \0 \Hol^f \right\}_K \0 \psi \big|_{K'_{L'}} = \left\{ g_1 \0 \Hol^v \0 \psi, g_2 \0 \Hol^f  \0 \psi \right\}_{K'} \big|_{K'_{L'}} \\
    =& \left\{ g_1 \0 \Hol'^v, g_2 \0 \Hol'^f \right\}_{K'} \big|_{K'_{L'}} = 0 \, ,
  \end{align*}
  where we used Statement (ii) for the transformed graph $\Gamma'$.
  The map $\psi |_{K'_{L'}}$ is surjective by Lemma \ref{lemma:kitaev_gluing_trafos} (v), so this concludes the proof of (ii).

  Now we prove (iii).
  The map $\Hol^{(v,f)}$ is surjective, which implies the injectivity of $\Hol^{(v,f)*}$.
  It remains to show that $\Hol^{(v,f)} : K \to (G, w_{G^*})$ is a Poisson map.
  Let $(v,f)$ be a site and consider first the simpler case where there are no loops at $v$ and no edges whose left and right sides belong both to $f$.
  Applying $(T\Hol^{(v,f)})^{\ox 2}$ to $w_K$ yields:
  \begin{align}
    & (T \Hol^{(v,f)})^{\ox 2} \, w_K (\gamma) = \Big( (TR_{\Hol^f(\gamma)} \0 T\Hol^v)^{\ox 2} + (TL_{\Hol^v(\gamma)} \0 T\Hol^f)^{\ox 2} \nonumber \\
    & + (TR_{\Hol^f(\gamma)} \0 T\Hol^v) \ox  (TL_{\Hol^v(\gamma)} \0 T\Hol^f) \label{eq:kitaev_commutation_relations_proof-2} \\
    & + (TL_{\Hol^v(\gamma)} \0 T\Hol^f) \ox (TR_{\Hol^f(\gamma)} \0 T\Hol^v) \Big) \, w_K(\gamma) \label{eq:kitaev_commutation_relations_proof-1} \, .
  \end{align}
  We can assume that the only two edges that occur both in the vertex path $p(v)$ and the face path $p(f)$ are the first and last edge in the ordering of $v$ (which coincide with the first and last edge of $f$).
  To see this, recall that all other such edges form pairs $(e_1, e_2)$ where $e_2$ comes directly after $e_1$ on both $p(f)$ and $p(v)$.
  By Lemma \ref{lemma:kitaev_reidemeister_ii}, the pair $(e_1, e_2)$ does not contribute to \eqref{eq:kitaev_commutation_relations_proof-2} and \eqref{eq:kitaev_commutation_relations_proof-1}.

  Denote by $f_1 < \dots < f_m$ the edges of $f$ and the edges incident at $v$ by $v_1 < \dots < v_n$.
  We can assume without loss of generality that all $v_1, \dots, v_n$ are incoming at $v$ and that $f_2, \dots, f_{m-1}$ are oriented clockwise.
  This implies that $f_1=v_1$ is oriented counterclockwise and $f_m = v_n$ clockwise.
  Denote the elements associated to the edges by $d_{v_i} := \pi_{v_i} (\gamma)$ and $d_{f_j} := \pi_{f_j}(\gamma)$.
  Then
  \begin{align*}
    & \Hol^{(v,f)}(\gamma) = \pi_-(d_{v_1}) \cdots \pi_-({d_{v_n}}) \pi_+(d_{f_m}) \cdots \pi_+(d_{f_2}) \pi_+(d_{f_1}^{-1}) \\
    =&\; \pi_-(d_{v_1}) \pi_-(d_{v_2}) \cdots \pi_-({d_{v_{n-1}}}) \, d_{f_m} \, \pi_+(d_{f_{m-1}}) \cdots \pi_+(d_{f_2}) \pi_+(d_{v_1}^{-1}) = \pi_-(d_{v_1}) \, g(\gamma) \, \pi_+(d_{v_1}^{-1}) \, ,
  \end{align*}
  where $g : K \to  G_\He $ is given by $g(\gamma) := \pi_-(d_{v_2}) \cdots \pi_-({d_{v_{n-1}}}) \, d_{f_m} \, \pi_+(d_{f_{m-1}}) \cdots \pi_+(d_{f_2})$.
  It is Poisson by Lemma \ref{lemma:projectionLocallyPoisson} and Theorem \ref{theorem:heisenberg_double_inversion_poisson_actions} (ii).
  Therefore, it suffices to show that the map
  \[
    H :  G_\He  \x  G_\He  \to (G, w_{G^*}) \qquad (d_1, d_2) \mapsto \pi_-(d_1) \, d_2 \, \pi_+(d_1^{-1})
  \]
  is Poisson.
  Consider the $n$-fold product $ G_\He ^n$ equipped with the product Poisson structure.
  We can write $H = h \0 \Delta_{13}$ with the map  $\Delta_{13}:  G_\He ^2 \to  G_\He ^3, (d_1, d_2) \mapsto (d_1, d_2, d_1)$ and the map  $h:  G_\He ^{3}  \to  G_\He , (d_1, d_2, d_3) \mapsto \pi_-(d_1) \, d_2 \, \pi_+(d_3^{-1})$.
  We obtain for the bivector $w_{G_\He^2}$ of $G_\He^2$
  \begin{equation}
    \label{eq:kitaev_commutation_relations_proof-0.5}
    \begin{split}
      & TH^{\ox 2} w_{ G_\He ^2 } (d_1, d_2) = Th^{\ox 2} \Big( (T\iota_1^{\ox 2} + T\iota_3^{\ox 2}) \, w_\He   (d_1) + T\iota_2^{\ox 2} \, w_\He (d_2)  \\
      & + (T\iota_1 \ox T\iota_3 + T\iota_3 \ox T\iota_1) \, w_\He (d_1) \Big) \, ,
    \end{split}
  \end{equation}
  where for $i =1,2,3$ the map $\iota_i : G \to G^3$ is the inclusion of $G$ into the $i$-th copy of $G^3$ at the element $(d_1, d_2, d_1)$.
  One has $ (T\iota_1^{\ox 2} + T\iota_3^{\ox 2}) \, w_\He  (d_1) + T\iota_2^{\ox 2} \, w_\He   (d_2) = w_{G_\He^{3}}  (d_1, d_2, d_1)$ and the map $h$ is Poisson by Lemma \ref{lemma:projectionLocallyPoisson} and Theorem \ref{theorem:heisenberg_double_inversion_poisson_actions}.
  Applying this to \eqref{eq:kitaev_commutation_relations_proof-0.5} yields
  \begin{equation}
    \label{eq:kitaev_commutation_relations_proof0}
    TH^{\ox 2} w_{G_\He^{2}}  (d_1, d_2) = w_\He  (H(d_1, d_2)) + Th^{\ox 2} \, (T\iota_1 \ox T\iota_3 + T\iota_3 \ox T\iota_1) \, w_\He   (d_1) \, .
  \end{equation}
  Denote the inversion map on $G$ by $\eta$.
  We have
  \begin{align*}
    & Th^{\ox 2} \0 (T\iota_1 \ox T\iota_3 + T\iota_3 \ox T\iota_1) \, w_\He   (d_1) = \big( T(R_{d_2 \pi_+(d_1^{-1})} \0 \pi_-) \ox T(L_{\pi_-(d_1) d_2} \0 \pi_+ \0 \eta) \\
    & + T(L_{\pi_-(d_1) d_2} \0 \pi_+ \0 \eta) \ox T(R_{d_2 \pi_+(d_1^{-1})} \0 \pi_-) \big) \, w_\He  (d_1) \, .
  \end{align*}
  Lemma \ref{lemma:global_double_projections_r_matrix} (i), the computation rules for $\pi_\pm$ from Lemma \ref{lemma:global_double_projections_properties}, and $r \in \g_- \ox \g_+$ imply:
  \begin{align*}
    & Th^{\ox 2} \0 (T\iota_1 \ox T\iota_3 + T\iota_3 \ox T\iota_1) \, w_\He   (d_1) \\
    \stackrel{\eqref{eq:global_double_projections_r_matrix_i_2}}= & - (T(R_{d_2 \pi_+(d_1^{-1})} \0 \pi_-) \ox T(L_{\pi_-(d_1) d_2} \0 \pi_+ \0 \eta)) \0 TR_{d_1}^{\ox 2} \, r \\
    & + (T(L_{\pi_-(d_1) d_2} \0 \pi_+ \0 \eta) \ox T(R_{d_2 \pi_+(d_1^{-1})} \0 \pi_-) ) \0 TR_{d_1}^{\ox 2} \, r_{21} \\
    = & \, (TR_{\pi_-(d_1) d_2 \pi_+(d_1^{-1})} \ox TL_{\pi_-(d_1) d_2 \pi_+(d_1^{-1})}) \, r - (TL_{\pi_-(d_1) d_2 \pi_+(d_1^{-1})} \ox TR_{\pi_-(d_1) d_2 \pi_+(d_1^{-1})}) \, r_{21} \\
    = & \, (TR_{H(d_1, d_2)} \ox TL_{H(d_1, d_2)}) \, r - (TL_{H(d_1,d_2)} \ox TR_{H(d_1, d_2)}) \, r_{21} \, .
  \end{align*}
  Inserting this into \eqref{eq:kitaev_commutation_relations_proof0} and comparing with $w_{G^*}$ from Equation \eqref{eq:bivector_dual_poisson_lie_group} proves the claim:
  \[
    TH^{\ox 2} w_{G_\He^{2}}  (d_1, d_2) = w_{G^*} (H(d_1, d_2)) \, .
  \]

  In the general case, we again transform $\Gamma$ by adding bivalent vertices to loops at $v$ and doubling edges where both sides belong to $f$. From Lemma \ref{lemma:kitaev_gluing_trafos} we obtain the map $\psi: K' \to K$ that is associated with gluing the split edges back together.
  Then we argue as in the proof of (ii).
\end{proof}

\begin{remark}
  \label{remark:kitaev_operator_algebra_dual_of_double}
  The homomorphism $\Hol^{(v,f)*} : C^\infty((G, w_{G^*}), \R) \to C^\infty(K, \R)$ of Poisson algebras from \eqref{eq:kitaev_site_operator_poisson_map} is an analogue of the algebra homomorphism $\tau: D(H) \to \End_\C(H^{\ox E})$ from Lemma \ref{lemma:quantum_kitaev_commutation_relations} (iv).
  The Poisson algebra $C^\infty(K, \R)$ is an analogue of the endomorphism algebra $\End_\C(H^{\ox E})$ for a quantum Kitaev model and the Poisson algebra $C^\infty((G, w_{G^*}), \R)$ takes the role of the algebra structure on $D(H)$.

  In fact, by Lemma \ref{lemma:dual_poisson_lie_group} there is a Poisson map $G^* \to (G, w_{G^*})$ that is a diffeomorphism near the units.
  The Poisson algebra structure on $C^\infty(G^*, \R)$ for the dual $G^*$ of the double Poisson-Lie group $G$ is a Poisson counterpart of the algebra structure on the Drinfeld double $D(H)$.
\end{remark}

By Equation \eqref{eq:quantum_kitaev_right_action}, the $D(H)$-right module algebra structure at a site $(v,f)$ is obtained from the vertex and face operators for $v$ and $f$, respectively.
We show an analogue of this relation by expressing the derivation $\left\{ g, - \right\}$ for a vertex or face operator $g$ in terms of the vector fields obtained from the respective vertex or face action.
For $\beta \in \g_+, y \in \g_-$ the vector fields generating the vertex action at $v$ and face action at $f$ are given by
\[
  V_v(\beta)(\gamma) := -T(- \rhd_v \gamma) \, \beta \qquad V_f(y)(\gamma) := -T(- \rhd_f \gamma) \, y \qquad \Forall \gamma \in K \, .
\]

\begin{proposition}[Vector fields for operators and actions.]
  Let $v \in V, f \in F$ and $g \in C^\infty(G, \R)$.
  The vector fields generated by the vertex operator $g \0 \Hol^v$ and face operator $g \0 \Hol^f$ satisfy
  \begin{align}
    \left\{ g \0 \Hol^v , - \right\} (\gamma) &= \sum_{(r)} \langle dg, TR_{\Hol^v (\gamma)} \, r_{(1)} \rangle \, V_v (r_{(2)})(\gamma) \label{eq:kitaev_vector_field_vertex_operator}\\
    \left\{ g \0 \Hol^f , - \right\} (\gamma) &= - \sum_{(r)} \langle dg, TL_{\Hol^f (\gamma)} \, r_{(2)} \rangle \, V_f (r_{(1)})(\gamma) \label{eq:kitaev_vector_field_face_operator} \, ,
  \end{align}
  where $r \in \g_- \ox \g_+$ is the classical $r$-matrix of the double Poisson-Lie group $G$.
\end{proposition}

\begin{proof}
  Equation \eqref{eq:kitaev_vector_field_vertex_operator} has been shown in Equation \eqref{eq:locally_flat_poisson_submanifold0} as part of the proof for Lemma \ref{lemma:locally_flat_poisson_submanifold}.
  The proof for Equation \eqref{eq:kitaev_vector_field_face_operator} is analogous.
\end{proof}

\begin{corollary}
  All invariant functions $h \in C^\infty(K, \R)^{inv}$ Poisson-commute with all vertex and face operators:
  \begin{equation}
    \label{eq:corollary_commutation_with_operators_equivalent_to_invariant}
    \left\{ h, g \0 \Hol^v \right\} = \left\{ h, g \0 \Hol^f \right\} = 0 \qquad \Forall g \in C^\infty(G, \R) \, , v \in V \, , f \in F \, .
  \end{equation}
  If $G$ is connected, then Equation \eqref{eq:corollary_commutation_with_operators_equivalent_to_invariant} for a function $h \in C^\infty(K, \R)$ in turn implies invariance under all vertex and face actions.
\end{corollary}

\begin{proof}
  Equations \eqref{eq:corollary_commutation_with_operators_equivalent_to_invariant} follow from Equations \eqref{eq:kitaev_vector_field_vertex_operator} and \eqref{eq:kitaev_vector_field_face_operator}.
  To prove the converse statement, note that the $r$-matrix of $G$ induces linear isomorphisms $\g_-^* \to \g_+: \varphi \mapsto (\varphi \ox \id) \, r $ and $\g_+^* \to \g_-: \varphi \mapsto (\varphi \ox \id) \, r_{21}$.
  By Equation, \eqref{eq:kitaev_vector_field_vertex_operator} the identity $\left\{ h, g \0 \Hol^v \right\} = 0 \Forall g \in C^\infty(G, \R)$ thus implies $\langle dh, V_v(\beta) \rangle = 0 \Forall \beta \in g_+$, so that $h \in C^\infty(K, \R)^v$.
  The argument for faces $f \in F$ is similar.
\end{proof}

\subsection{Poisson-Kitaev models and Fock-Rosly spaces}
\label{subsection:poisson_kitaev_models_fock_rosly_spaces}

In this section we show that a Poisson-Kitaev model on the graph $\Gamma$ is Poisson-isomorphic to a Fock-Rosly space on $\Gamma$ if all vertices and faces can be paired into sites.
This isomorphism will be used in Section \ref{subsection:relation_with_moduli_spaces} to prove that the Poisson algebra $\mathcal A(\Gamma, L)$ is isomorphic to the canonical Poisson algebra of functions on the moduli space of flat $G$-bundles.

In this section we assume that $\Gamma$ is a \emph{paired} doubly ciliated ribbon graph (see Definition \ref{def:double_graph_basics} (iv)).

We consider the Poisson-Kitaev model $K =  G_\He ^{\x E}$ for the graph $\Gamma$ and the Fock-Rosly space $\FR = (G^{\x E}, w_{\FR})$ from Definition \ref{def:fock_rosly_space}.
The latter is equipped with the Poisson bivector $w_{\FR}$ from Equation \eqref{eq:fockrosly_vertex_bivector} where we set $r(v) = r$ for all $v \in V$ and $r$ is the classical $r$-matrix of $G$.

We construct a Poisson isomorphism $\Phi: K\to \FR$.
It is obtained from the holonomies of certain paths associated to the edges of $\Gamma$.
More specifically, we associate to an edge $e: v_s \to v_t$ the path $p(e)$ in the thickening $\Gamma_D$ (see Definition \ref{def:thickening}) that starts at the cilium of $v_s$, turns counterclockwise around $v_s$, then along $e$ and clockwise along $v_t$ until it meets the cilium of $v_t$.
This path is depicted in Figure \ref{fig:trafo_kitaev_to_fr}.

Suppose that $b(e)$ is the $j_s$-th edge end in the linear ordering at $v_s$ and $f(e)$ the $j_t$-th edge end at $v_t$.
The path $p(e)$ is given by
\begin{equation}
  \label{eq:kitaev_fock_rosly_iso_path_for_an_edge}
  p(e) :=  p_f(e) \0 f(e) \0 r(e) \0 p_b(e) \quad \text{with} \quad p_b(e) := p_{j_s-1}(v_s)^{-1} \quad p_f(e):= p_{j_t-1}(v_t) \, ,
\end{equation}
where $p_{j_s -1}(v_s)$ and $p_{j_t -1}(v_t)$ are partial vertex paths as defined by Equation \eqref{eq:partial_vertex_path}.
With the help of the holonomy functor $\Hol : \G(\Gamma_D) \to C^\infty(K, G)$ from \eqref{eq:kitaev_holonomy_functor} we define the map $\Phi: K \to \FR$ by
\begin{equation}
  \label{eq:kitaev_fock_rosly_iso}
  \pi_e \0 \Phi := \Hol (p(e)) \, .
\end{equation}

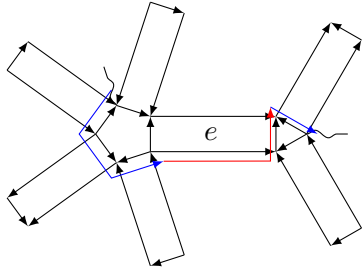
\begin{figure}[h]
\centering
\begin{tikzpicture}[vertex/.style={circle, fill=black, inner sep=0pt, minimum size=2mm}, plain/.style={draw=none, fill=none}, scale=0.8]
  \begin{scope}[scale=0.9]
    
  \coordinate (p) at (0,0);
  \coordinate (p1) at ($(p) + (36:0.553)$);
  \coordinate (p2) at ($(p) + (108:0.553)$);
  \coordinate (p3) at ($(p) + (180:0.553)$);
  \coordinate (p4) at ($(p) + (252:0.553)$);
  \coordinate (p5) at ($(p) + (-36:0.553)$);

  \coordinate (q) at ($(p) + (2,0) + (0.553,0) + (0.375,0)$);
  \coordinate (q1) at ($(q) + (0:0.375)$);
  \coordinate (q2) at ($(q) + (120:0.375)$);
  \coordinate (q3) at ($(q) + (240:0.375)$);

  \draw [-latex] (p1)--(q2);
  \draw [-latex] (q3)--(q2);
  \draw [-latex] (p5)--(p1);
  \draw [-latex] (p5)--(q3);
  \node at ($(p)!1.553cm!(q)$) {$e$};

  \draw [-latex] (q1) -- (q2);
  \draw [-latex] (q1) -- ($(q1)!2cm!-90:(q2)$);
  \draw [-latex] (q2) -- ($(q2) + (q1)!2cm!-90:(q2) - (q1)$);
  \draw [-latex] ($(q1)!2cm!-90:(q2)$) -- ($(q2) + (q1)!2cm!-90:(q2) - (q1)$);

  \draw [-latex] (q1) -- (q3);
  \draw [-latex] ($(q1)!2cm!90:(q3)$) -- (q1);
  \draw [-latex] ($(q3)!2cm!-90:(q1)$) -- (q3);
  \draw [-latex] ($(q1)!2cm!90:(q3)$) -- ($(q3)!2cm!-90:(q1)$);

  \draw [decorate,decoration={snake, amplitude=0.5mm}] (q1) -- ($(q1)+(0.75,0)$);

  \draw [-latex] (p2) -- (p1);
  \draw [-latex] ($(p1)!2cm!-90:(p2)$) -- (p1);
  \draw [-latex] ($(p2)!2cm!90:(p1)$) -- (p2);
  \draw [-latex] ($(p2)!2cm!90:(p1)$) -- ($(p1)!2cm!-90:(p2)$);

  \draw [-latex] (p3) -- (p2);
  \draw [-latex] ($(p2)!2cm!-90:(p3)$) -- (p2);
  \draw [-latex] ($(p3)!2cm!90:(p2)$) -- (p3);
  \draw [-latex] ($(p3)!2cm!90:(p2)$) -- ($(p2)!2cm!-90:(p3)$);

  \draw [-latex] (p3) -- (p4);
  \draw [-latex] (p3) -- ($(p3)!2cm!-90:(p4)$);
  \draw [-latex] (p4) -- ($(p4)!2cm!90:(p3)$);
  \draw [-latex] ($(p3)!2cm!-90:(p4)$) -- ($(p4)!2cm!90:(p3)$);

  \draw [-latex] (p5) -- (p4);
  \draw [-latex] ($(p4)!2cm!-90:(p5)$) -- (p4);
  \draw [-latex] ($(p5)!2cm!90:(p4)$) -- (p5);
  \draw [-latex] ($(p5)!2cm!90:(p4)$) -- ($(p4)!2cm!-90:(p5)$);

  \draw [decorate,decoration={snake, amplitude=0.5mm}] (p2) -- ($(p2) + (p)!0.75cm!(p2)$);

  \draw [-latex,color=blue] ($(p2) + (p)!0.3cm!(p2)$) -- ($(p3) + (p)!0.3cm!(p3)$) -- ($(p4) + (p)!0.3cm!(p4)$) -- ($(p5) + (p)!0.3cm!(p5)$);
  \draw [-latex,color=red] ($(p5) + (p)!0.3cm!(p5)$) -- ($(q3) + (q)!0.2cm!(q3) - (q)$) -- ($(q2) + (q)!0.2cm!(q2) - (q)$);
  \draw [-latex,color=blue] ($(q2) + (q)!0.2cm!(q2) - (q)$) -- ($(q1) + (q)!0.2cm!(q1) - (q)$);

  \end{scope}
\end{tikzpicture}
\caption{The paths $p_b(e)$ and $p_f(e)$ are shown in blue, the path $f(e) \circ r(e)$ in red}
\label{fig:trafo_kitaev_to_fr}
\end{figure}

This map can be viewed as the Poisson-Lie counterpart of the isomorphism of $D(H)$-module algebras in \cite{meusburger16} that was used to relate the endomorphism algebra $\End_\C(H^{\otimes E})$ of a Kitaev model to the quantum moduli algebra.

We show that for each site $(v,f)$ the map $\Phi: K \to \FR$ is an isomorphism of Poisson $G$-spaces with respect to the Poisson actions $\rhd_{(v,f)}: G \x K \to K$ from \eqref{eq:site_action} and $\rhd_v^{\FR} : G \x \FR \to \FR$ from \eqref{eq:fock_rosly_vertex_action}.
We also prove that it intertwines the holonomies along the combined vertex and face path $p(v) \0 p(f)$ with respect to the holonomy functors $\Hol$ from \eqref{eq:kitaev_holonomy_functor} and $\Hol_{\FR}$ from \eqref{eq:fock_rosly_holonomy_functor}.

\begin{theorem}[Poisson-Kitaev models and Fock-Rosly spaces]
  \label{theorem:KitaevFockRoslyIso}
  ~
  \begin{compactenum}
  \item
    The map $\Phi: K \to \FR$ is a Poisson isomorphism.
  \item
    For each site $(v, f)$ the map $\Phi$ is a homomorphism  of Poisson $G$-spaces with respect to the Poisson actions $\rhd_{(v, f)} : G \x K \to K$ and $\rhd_v^{\FR} : G \x \FR \to \FR$:
    \begin{equation}
      \label{eq:KitaevFockRoslyIso_compatible_site_actions}
      g \rhd^{FR}_v \Phi(\gamma) = \Phi ( \pi_-(g) \rhd_{f} ( \pi_+(g) \rhd_v \gamma) ) \qquad \Forall g \in G, \gamma \in K \, .
    \end{equation}
  \item
    The map $\Phi$ intertwines the holonomies associated with any site $(v, f)$:
    \begin{equation}
      \label{eq:KitaevFockRoslyIso_compatible_holonomies}
      \Hol_{\FR}(p(v) \0 p(f)) \0 \Phi = \Hol(p(v) \0 p(f)) \, .
    \end{equation}
  \end{compactenum}
\end{theorem}

\begin{proof}
  \textbf{Statement (i):} We construct an inverse $\Psi: \FR \to K$ of $\Phi$.
  For an edge $e \in E$, let $f_l, f_r$ be the faces to the left and right of $e$, respectively.
  Suppose that $r(e)$ is the $k_r$-th edge side in the linear ordering at $f_r$ and $l(e)$ the $k_l$-th one at $f_l$.
  Consider the partial face paths $p_l(e) := p_{k_l-1}(f_l)$ and $p_r(e) := p_{k_r-1}(f_r)$ from Equation \eqref{eq:partial_face_path} as illustrated in Figure \ref{fig:trafo_fr_to_kitaev}.
  We define $\Psi$ with the help of the holonomy functor $\Hol_{\FR}$ from \eqref{eq:fock_rosly_holonomy_functor} by:
  \begin{equation}
    \label{eq:iso_kitaev_fock_rosly_inverse}
    (\pi_e \0 \Psi) (\gamma) := \pi_-(\Hol_{\FR} (p_l(e))(\gamma))^{-1} \, \pi_e (\gamma) \, \pi_-(\Hol_{\FR}(p_r(e))(\gamma)) \, .
  \end{equation}

  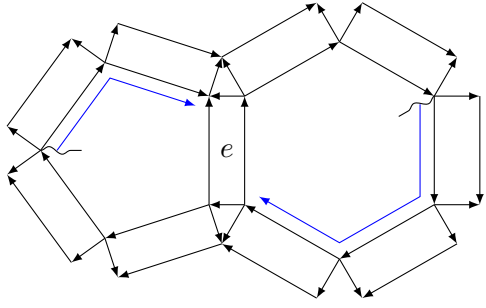
\begin{figure}[h]
  \centering
  \begin{tikzpicture}[vertex/.style={circle, fill=black, inner sep=0pt, minimum size=2mm}, plain/.style={draw=none, fill=none}, scale=0.8]
    \begin{scope}[scale=0.9]

    \coordinate (pl) at (0,-1);
    \coordinate (pr) at (0.65,-1);
    \coordinate (ql) at (0,1);
    \coordinate (qr) at (0.65,1);
    \node at (0.325,0) {$e$};

    \draw [-latex] (pr) -- (qr);
    \draw [-latex] (pl) -- (ql);
    \draw [-latex] (pr) -- (pl);
    \draw [-latex] (qr) -- (ql);

    \coordinate (r) at ($(1.732, 0) + (0.65, 0)$);
    \coordinate (r1) at ($(r) + (90:2cm)$);
    \coordinate (r2) at ($(r) + (30:2cm)$);
    \coordinate (r3) at ($(r) + (-30:2cm)$);
    \coordinate (r4) at ($(r) + (-90:2cm)$);

    \draw [-latex] (qr) -- (r1);
    \draw [-latex] (r1) -- (r2);
    \draw [-latex] (r2) -- (r3);
    \draw [-latex] (r3) -- (r4);
    \draw [-latex] (r4) -- (pr);

    \draw [decorate,decoration={snake, amplitude=0.5mm}] (r2) -- ($(r2)!0.75cm!(r)$);

    \draw [-latex,color=blue] ($(r2)!0.3cm!(r)$) -- ($(r3)!0.3cm!(r)$) -- ($(r4)!0.3cm!(r)$) -- ($(pr)!0.3cm!(r)$);

    \coordinate (l) at (-1.376,0);
    \coordinate (l1) at ($(l) + (-108:1.701cm)$);
    \coordinate (l2) at ($(l) + (-180:1.701cm)$);
    \coordinate (l3) at ($(l) + (108:1.701cm)$);

    \draw [-latex] (pl) -- (l1);
    \draw [-latex] (l1) -- (l2);
    \draw [-latex] (l2) -- (l3);
    \draw [-latex] (l3) -- (ql);

    \draw [decorate,decoration={snake, amplitude=0.5mm}] (l2) -- ($(l2)!0.75cm!(l)$);

    \draw [-latex,color=blue] ($(l2)!0.3cm!(l)$) -- ($(l3)!0.3cm!(l)$) -- ($(ql)!0.3cm!(l)$);

    \coordinate (ql') at ($(ql)!1cm!-90:(l3)$);
    \coordinate (qr') at ($(qr)!1cm!90:(r1)$);
    \coordinate (q') at (intersection of ql--ql' and qr--qr');

    \draw [-latex] (qr) -- (q');
    \draw [-latex] (r1) -- ($(r1) + (q') - (qr)$);
    \draw [-latex] (q') -- ($(r1) + (q') - (qr)$);

    \draw [-latex] (r1) -- ($(r1) + (qr)!1!-60:(q') - (qr)$);
    \draw [-latex] (r2) -- ($(r2) + (qr)!1!-60:(q') - (qr)$);
    \draw [-latex] ($(r1) + (qr)!1!-60:(q') - (qr)$) -- ($(r2) + (qr)!1!-60:(q') - (qr)$);

    \draw [-latex] (r2) -- ($(r2) + (qr)!1!-120:(q') - (qr)$);
    \draw [-latex] (r3) -- ($(r3) + (qr)!1!-120:(q') - (qr)$);
    \draw [-latex] ($(r2) + (qr)!1!-120:(q') - (qr)$) -- ($(r3) + (qr)!1!-120:(q') - (qr)$);
  
    \draw [-latex] (r3) -- ($(r3) + (qr)!1!-180:(q') - (qr)$);
    \draw [-latex] (r4) -- ($(r4) + (qr)!1!-180:(q') - (qr)$);
    \draw [-latex] ($(r3) + (qr)!1!-180:(q') - (qr)$) -- ($(r4) + (qr)!1!-180:(q') - (qr)$);

    \draw [-latex] (r4) -- ($(r4) + (qr)!1!-240:(q') - (qr)$);
    \draw [-latex] (pr) -- ($(pr) + (qr)!1!-240:(q') - (qr)$);
    \draw [-latex] ($(r4) + (qr)!1!-240:(q') - (qr)$) -- ($(pr) + (qr)!1!-240:(q') - (qr)$);

    \draw [-latex] (ql) -- (q');
    \draw [-latex] (l3) -- ($(l3) + (q') - (ql)$);
    \draw [-latex] ($(l3) + (q') - (ql)$) -- (q');

    \draw [-latex] (l3) -- ($(l3) + (ql)!1!72:(q') - (ql)$);
    \draw [-latex] (l2) -- ($(l2) + (ql)!1!72:(q') - (ql)$);
    \draw [-latex] ($(l2) + (ql)!1!72:(q') - (ql)$) -- ($(l3) + (ql)!1!72:(q') - (ql)$);

    \draw [-latex] (l2) -- ($(l2) + (ql)!1!144:(q') - (ql)$);
    \draw [-latex] (l1) -- ($(l1) + (ql)!1!144:(q') - (ql)$);
    \draw [-latex] ($(l1) + (ql)!1!144:(q') - (ql)$) -- ($(l2) + (ql)!1!144:(q') - (ql)$);

    \draw [-latex] (l1) -- ($(l1) + (ql)!1!216:(q') - (ql)$);
    \draw [-latex] (pl) -- ($(pl) + (ql)!1!216:(q') - (ql)$);
    \draw [-latex] ($(pl) + (ql)!1!216:(q') - (ql)$) -- ($(l1) + (ql)!1!216:(q') - (ql)$);
      
    \end{scope}
  \end{tikzpicture}
  \caption{The paths $p_l (e)$ (left) and $p_r(e)$ (right) for the edge $e$}
  \label{fig:trafo_fr_to_kitaev}
  \end{figure}

  That $\Psi$ is the inverse to $\Phi$ can be checked directly using the properties of the projections $\pi_\pm$ in Lemma \ref{lemma:global_double_projections_properties}.

  It remains to show that $\Phi: K \to \FR$ is Poisson.
  To simplify presentation, we prove the equation
  \begin{equation}
    \left\{ f_1 \0 \Phi, f_2 \0 \Phi \right\}_K = \left\{ f_1 , f_2 \right\}_{FR} \circ \Phi
    \label{eq:PhiPoisson}
  \end{equation}
  for functions of the form $f_i = \hat f_i \0 \pi_{e_i}$ with edges $e_1, e_2 \in E$ and $\hat f_1, \hat f_2 \in C^\infty(G, \R)$.
  It implies $T\pi_{e_1} \ox T\pi_{e_2} \, (T\Phi^{\ox 2} \, w_K) = T\pi_{e_1} \ox T\pi_{e_2} \, (w_{\FR} \0 \Phi)$ for all edges $e_1, e_2$ and hence that $\Phi$ is Poisson.

  First we prove Equation \eqref{eq:PhiPoisson} for graphs $\Gamma$ which satisfy the additional assumption that for all vertices $v$ any non-trivial path $p: v \to v$ in $\Gamma$ traverses at least four edges.
  This excludes loops, pairs of vertices connected by more than one edge, and triangles.
  Recall from Remark \ref{remark:edge_reversal} and Proposition \ref{proposition:fockrosly_graph_trafos}
  that a change of the orientation of an edge $e$ amounts to an inversion of the element $\pi_e(\gamma)$ in both the Poisson-Kitaev model $K$ and the Fock-Rosly space $\FR$.
  In both cases this is a Poisson isomorphism and one can check easily that it commutes with $\Phi$.
  Thus, we can choose the orientations of any edges and only  have to consider the following cases:
  \begin{enumerate}[nolistsep, noitemsep, label=(\alph*)]
    \item
      $e_1 = e_2$.
    \item
      The edges $e_1$ and $e_2$ have no vertex in common.
    \item
      $s(e_1) = s(e_2)$ with $e_1 < e_2$ at $s(e_1)$ and $t(e_1) \neq t(e_2)$.
  \end{enumerate}

  \textbf{Case (a):}
  Let $e := e_1 = e_2$.
  The map $\pi_e \0 \Phi : K \to  G_\He $ is Poisson by Lemma \ref{lemma:projectionLocallyPoisson} (ii) and Theorem \ref{theorem:heisenberg_double_inversion_poisson_actions} (ii).
  Equation \eqref{eq:fockrosly_vertex_bivector} implies that $T\pi_e^{\ox 2} \, w_{\FR} = w_\He  \0 \pi_e$.
  We obtain:
  \begin{align*}
    \left\{ f_1 \0 \Phi, f_2 \0 \Phi \right\}_K &= \left\{ \hat f_1 \0 \pi_e \0 \Phi, \hat f_2 \0 \pi_e \0 \Phi \right\}_K = \left\{ \hat f_1, \hat f_2 \right\}_{ G_\He } \0 \pi_e \0 \Phi \\
    &= \left\{ \hat f_1 \0 \pi_e, \hat f_2 \0 \pi_e \right\}_{\FR} \0 \Phi = \left\{ f_1, f_2 \right\}_{\FR} \0 \Phi \, .
  \end{align*}

  \textbf{Case (b):}
  If the edges $e_1, e_2$ have no common vertex, then $\left\{ f_1 , f_2 \right\}_{FR} = 0$ by Equation \eqref{eq:fockrosly_vertex_bivector}.
  The only edges $e$ that potentially contribute to $\left\{ f_1 \circ \Phi, f_2 \circ \Phi \right\}_K$ connect the edges $e_1$ and $e_2$.
  With proper choice of orientation for $e_1$ and $e_2$, the edge ends of $e$ are traversed by the paths $p_b(e_i)$ for $i=1,2$, but not by $p_f(e_i)$ because our assumptions on $\Gamma$ imply that the vertices $s(e_i)$ and $t(e_i)$ are only connected by $e_i$.
  Then $f_1$ is a function of $\Hol(b(e))$ and $f_2$ of $\Hol(f(e))$ or vice versa. 
  Lemma \ref{lemma:omegaHGbothsideszero} implies $\left\{ f_1 \0 \Phi, f_2 \0 \Phi \right\}_K = 0 = \left\{ f_1, f_2 \right\}_{\FR} \0 \Phi$.

  \textbf{Case (c):}
  Let $v := s(e_1) = s(e_2)$.
  We can assume that there are no edges in $\Gamma$ that connect $t(e_1)$ with $t(e_2)$ as otherwise $e_1, e_2$ would be part of a triangle.
  This situation is pictured in Figure \ref{fig:trafo_kitaev_to_fr_proof} along with the paths $p(e_i) = p_f(e_i) \0 f(e_i) \0 r(e_i) \0 p_b(e_i)$ for $i=1,2$.

  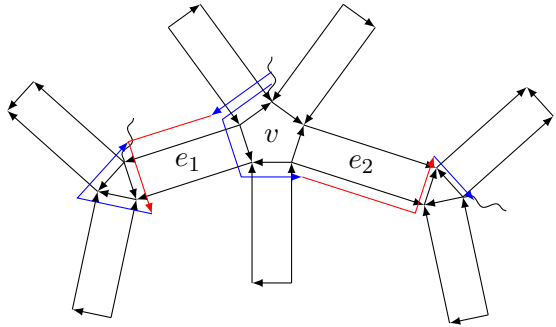
\begin{figure}[h]
  \centering
  \begin{tikzpicture}[vertex/.style={circle, fill=black, inner sep=0pt, minimum size=2mm}, plain/.style={draw=none, fill=none}, scale=0.8]
    \begin{scope}[rotate=-18]
    \node (p) at (0,0) {$v$};
    \coordinate (p1) at ($(p) + (36:0.553)$);
    \coordinate (p2) at ($(p) + (108:0.553)$);
    \coordinate (p3) at ($(p) + (180:0.553)$);
    \coordinate (p4) at ($(p) + (252:0.553)$);
    \coordinate (p5) at ($(p) + (-36:0.553)$);

    \coordinate (q) at ($(p) + (2,0) + (0.553,0) + (0.375,0)$);
    \coordinate (q1) at ($(q) + (0:0.375)$);
    \coordinate (q2) at ($(q) + (120:0.375)$);
    \coordinate (q3) at ($(q) + (240:0.375)$);

    \draw [-latex] (p1)--(q2);
    \draw [-latex] (q3)--(q2);
    \draw [-latex] (p5)--(p1);
    \draw [-latex] (p5)--(q3);
    \node at ($(p)!1.553cm!(q)$) {$e_2$};

    \draw [-latex] (q1) -- (q2);
    \draw [-latex] (q1) -- ($(q1)!2cm!-90:(q2)$);
    \draw [-latex] (q2) -- ($(q2) + (q1)!2cm!-90:(q2) - (q1)$);
    \draw [-latex] ($(q1)!2cm!-90:(q2)$) -- ($(q2) + (q1)!2cm!-90:(q2) - (q1)$);

    \draw [-latex] (q1) -- (q3);
    \draw [-latex] ($(q1)!2cm!90:(q3)$) -- (q1);
    \draw [-latex] ($(q3)!2cm!-90:(q1)$) -- (q3);
    \draw [-latex] ($(q1)!2cm!90:(q3)$) -- ($(q3)!2cm!-90:(q1)$);

    \draw [decorate,decoration={snake, amplitude=0.5mm}] (q1) -- ($(q1)+(0.75,0)$);

    \draw [-latex] (p2) -- (p1);
    \draw [-latex] ($(p1)!2cm!-90:(p2)$) -- (p1);
    \draw [-latex] ($(p2)!2cm!90:(p1)$) -- (p2);
    \draw [-latex] ($(p2)!2cm!90:(p1)$) -- ($(p1)!2cm!-90:(p2)$);

    \draw [-latex] (p3) -- (p2);
    \draw [-latex] ($(p2)!2cm!-90:(p3)$) -- (p2);
    \draw [-latex] ($(p3)!2cm!90:(p2)$) -- (p3);
    \draw [-latex] ($(p3)!2cm!90:(p2)$) -- ($(p2)!2cm!-90:(p3)$);

    \draw [-latex] (p3) -- (p4);
    \draw [-latex] (p3) -- ($(p3)!2cm!-90:(p4)$);
    \draw [-latex] (p4) -- ($(p4)!2cm!90:(p3)$);
    \draw [-latex] ($(p3)!2cm!-90:(p4)$) -- ($(p4)!2cm!90:(p3)$);
    \node at ($(p3)!0.5!(p4) + (p3)!1cm!-90:(p4) - (p3)$) {$e_1$};

    \draw [-latex] (p5) -- (p4);
    \draw [-latex] ($(p4)!2cm!-90:(p5)$) -- (p4);
    \draw [-latex] ($(p5)!2cm!90:(p4)$) -- (p5);
    \draw [-latex] ($(p5)!2cm!90:(p4)$) -- ($(p4)!2cm!-90:(p5)$);

    \draw [decorate,decoration={snake, amplitude=0.5mm}] (p2) -- ($(p2) + (p)!0.75cm!(p2)$);

    \coordinate (r1) at ($(p3)!2cm!-90:(p4)$);
    \coordinate (r2) at ($(p4)!2cm!90:(p3)$);
    \coordinate (r3) at ($(r2)!0.65cm!60:(r1)$);
    \coordinate (r) at ($0.333*(r1) + 0.333*(r2) + 0.333*(r3)$);

    \draw [-latex] (r2) -- (r3);
    \draw [latex-] (r2) -- ($(r2)!2cm!90:(r3)$);
    \draw [latex-] (r3) -- ($(r3)!2cm!-90:(r2)$);
    \draw [-latex] ($(r2)!2cm!90:(r3)$) -- ($(r3)!2cm!-90:(r2)$);

    \draw [-latex] (r1) -- (r3);
    \draw [-latex] (r1) -- ($(r1)!2cm!-90:(r3)$);
    \draw [-latex] (r3) -- ($(r3)!2cm!90:(r1)$);
    \draw [-latex] ($(r1)!2cm!-90:(r3)$) -- ($(r3)!2cm!90:(r1)$);

    \draw [decorate,decoration={snake, amplitude=0.5mm}] (r1) -- ($(r1)!0.75cm!150:(r2)$);

    \draw [-latex,color=blue] ($(p2) + (p)!0.3cm!(p2)$) -- ($(p3) + (p)!0.3cm!(p3)$) -- ($(p4) + (p)!0.3cm!(p4)$) -- ($(p5) + (p)!0.3cm!(p5)$);
    \draw [-latex,color=red] ($(p5) + (p)!0.3cm!(p5)$) -- ($(q3) + (q)!0.2cm!(q3) - (q)$) -- ($(q2) + (q)!0.2cm!(q2) - (q)$);
    \draw [-latex,color=blue] ($(q2) + (q)!0.2cm!(q2) - (q)$) -- ($(q1) + (q)!0.2cm!(q1) - (q)$);

    \draw [-latex,color=blue] ($(p2) + (p)!0.5cm!(p2)$) -- ($(p3) + (p)!0.5cm!(p3)$);
    \draw [-latex,color=red] ($(p3) + (p)!0.5cm!(p3)$) -- ($(r1) + (r)!0.35cm!(r1) -(r)$) -- ($(r2) + (r)!0.35cm!(r2) -(r)$);
    \draw [-latex,color=blue] ($(r2) + (r)!0.35cm!(r2) -(r)$) -- ($(r3) + (r)!0.35cm!(r3) -(r)$) -- ($(r1) + (r)!0.35cm!(r1) -(r)$);
    \end{scope}
  \end{tikzpicture}
  \caption{The paths $p_b(e_i)$ and $p_f(e_i)$ (blue) and the paths $f(e_i) \circ r(e_i)$ (red)}
  \label{fig:trafo_kitaev_to_fr_proof}
  \end{figure}
  First we consider the case when $e_1$ is the first edge at $v$ and $e_2$ comes directly after $e_1$ with respect to the linear ordering, so that $p_b(e_1) = 1_{s(b(e_1))}$ is the identity morphism of $s(b(e_1))$ and $p_b(e_2) = b(e_1)$.
  Additionally, we require that $p_f(e_1) = 1_{t(f(e_1))}$ and $p_f(e_2) = 1_{t(f(e_2))}$.
  Then one has
  \[
    (\pi_{e_1} \0 \Phi) (\gamma) = \pi_{e_1}(\gamma) \qquad (\pi_{e_2} \0 \Phi)(\gamma) = \pi_{e_2}(\gamma) \, \pi_-(\pi_{e_1}(\gamma)^{-1})^{-1} \, .
  \]
  Thus, we have to prove that the following map is Poisson
  \begin{equation}
    \label{eq:theorem_KitaevFockRoslyIso_proof-2}
    h :  G_\He ^2 \to (G^2, w_{\FR}') \qquad (d_1, d_2) \mapsto (d_1, d_2 \, \pi_-(d_1^{-1})^{-1}) \, .
  \end{equation}
  Here, $ G_\He ^2$ is equipped with the bivector $w_\pi$ of the product Poisson structure and $w_{\FR}'$ is the Poisson bivector from Equation \eqref{eq:fockrosly_vertex_bivector} for the case when the graph consists only of the two edges $e_1: v \to t(e_1)$ and $e_2 : v \to t(e_2)$ with $e_1 < e_2$ at $v$ and $t(e_1) \neq t(e_2)$.
  Denote by $\pi_1, \pi_2 : G^2 \to G$ the projections on the components that correspond to $e_1$ and $e_2$.
  Then $w_{\FR}'$ is explicitly given by
  \begin{equation}
    \label{eq:theorem_KitaevFockRoslyIso_proof-1}
    \begin{split}
      (T\pi_i \ox T\pi_i) \, w'_{\FR} (d_1, d_2) &= - (TL_{d_i}^{\ox 2} + TR_{d_i}^{\ox 2}) \, r_a = w_\He   (d_i) \qquad i=1,2 \\
      (T\pi_1 \ox T\pi_2) \, w'_{\FR} (d_1, d_2) &=  TL_{d_1} \ox TL_{d_2} \, r_{21} \, ,
    \end{split}
  \end{equation}
  and the bivector $w_\pi$ by
  \begin{equation*}
    (T\pi_1 \ox T\pi_2) \, w_{\pi} (d_1, d_2) = 0 \qquad (T\pi_i \ox T\pi_i) \, w_{\pi} (d_1, d_2) = w_\He   (d_i) \qquad i = 1,2 \, .
  \end{equation*}
  We decompose $h$ into Poisson maps.
  For this we write $h$ as
  \begin{equation}
    \label{eq:theorem_KitaevFockRoslyIso_proof-0.5}
    h = (\id_{G} \x (\mu \circ ( \id_{G} \x f) \0 \tau)) \0 ( \Delta \x \id_{G}) \, ,
  \end{equation}
  where $\Delta: G \to G \x G, d \mapsto (d,d)$ is the diagonal map, $\tau: G^2 \to G^2, (d_1, d_2) \mapsto (d_2, d_1)$ the flip, $\mu: G \x G \to G$ the multiplication and $f := \eta \0 \pi_- \0 \eta$ with the inversion map $\eta : G \to G$.

  Let $\pi_i' : G^3 \to G, i =1,2,3$ be the projections on the components of $G^3$ and $\iota_i': G \to G^3, i=1,2,3$ the inclusion maps at the element $(\Delta \x \id_G) (d_1, d_2) = (d_1, d_1, d_2)$.
  They satisfy:
  \[
    \pi_i' \0 \iota_j' =
    \begin{cases}
      \id_G & \text{if } i=j \\
      (g \mapsto \pi_i' (d_1, d_1, d_2)) & \text{otherwise.}
    \end{cases}
  \]
  Denote by $w_\pi'$ the Poisson bivector of the product Poisson manifold $ G_\He ^3$.
  Then one has
  \begin{equation}
    \label{eq:theorem_KitaevFockRoslyIso_proof0}
    T(\Delta \x \id_G)^{\ox 2} \; w_{\pi} (d_1, d_2) = \, w_\pi' (d_1, d_1, d_2) + (T\iota_1' \ox T\iota_2' + T\iota_2' \ox T\iota_1') \, w_\He   (d_1) \, .
  \end{equation}
  This implies for the map $h$ from Equation \eqref{eq:theorem_KitaevFockRoslyIso_proof-0.5}:
  \begin{equation}
    \label{eq:theorem_KitaevFockRoslyIso_proof0.5}
    \begin{split}
      & Th^{\ox 2} \, w_{\pi} (d_1, d_2) = \, T(\id_G \x (\mu \circ (\id_{ G } \x f) \0 \tau) )^{\ox 2} \, w_\pi' (d_1, d_1, d_2) \\
      & + T(\id_G \x (\mu \circ (\id_{ G } \x f) \0 \tau) )^{\ox 2} \0 (T\iota_1' \ox T\iota_2' + T\iota_2' \ox T\iota_1') \, w_\He  (d_1) \, .
    \end{split}
  \end{equation}
  By Theorem \ref{theorem:heisenberg_double_inversion_poisson_actions} (ii) and Lemma \ref{lemma:projectionLocallyPoisson} (ii) the following map is Poisson
  \[
    \mu \circ ( \id_G \x f) \0 \tau:  G_\He ^2 \to  G_\He  \qquad (d_1, d_2) \mapsto d_2 \, \pi_-(d_1^{-1})^{-1} \, ,
  \]
  and applying this fact to the first term on the right hand side of \eqref{eq:theorem_KitaevFockRoslyIso_proof0.5} yields:
  \begin{equation}
    \label{eq:theorem_KitaevFockRoslyIso_proof1}
    \begin{split}
      & Th^{\ox 2} \, w_{\pi} (d_1, d_2) = \, w_\pi(h(d_1, d_2)) \\
      & + T(\id_G \x (\mu \0 (\id_{ G } \x f) \0 \tau))^{\ox 2} \0 (T\iota_1' \ox T\iota_2' + T\iota_2' \ox T\iota_1') \, w_\He  (d_1) \, .
    \end{split}
  \end{equation}
  Apply the projection $T\pi_i^{\ox 2}$ and use \eqref{eq:theorem_KitaevFockRoslyIso_proof-1} to obtain:
  \begin{equation}
    \label{eq:theorem_KitaevFockRoslyIso_proof_h_is_Poisson1}
    T\pi_i^{\ox 2} ( Th^{\ox 2} \, w_{\pi} (d_1, d_2) ) = T\pi_i^{\ox 2} (w_\pi(h(d_1, d_2))) = T\pi_i^{\ox 2} \, w'_{\FR} (h(d_1, d_2)) \, .
  \end{equation}
  By the anti-symmetry of the Poisson bivectors $w_\pi$ and $w_{\FR}'$ it only remains to show the following equation to prove that $h$ is Poisson:
  \[
    (T\pi_1 \ox T\pi_2) \, ( Th^{\ox 2} \, w_{\pi} (d_1, d_2) ) = (T\pi_1 \ox T\pi_2) \, w'_{\FR} (h(d_1, d_2)) \, .
  \]
  Applying the projection $T\pi_1 \ox T\pi_2$ to Equation \eqref{eq:theorem_KitaevFockRoslyIso_proof1} yields
  \begin{align*}
    & (T\pi_1 \ox T\pi_2) \0 Th^{\ox 2} \, w_\pi (d_1, d_2) = \, (T\id_G \ox T(L_{d_2} \0 f)) \, w_\He  (d_1) \\
    = & \, (T(\eta \0 \eta) \ox T(L_{d_2} \0 \eta \0 \pi_- \0 \eta)) \, w_\He  (d_1) \, ,
  \end{align*}
  where we used the identity $f = \eta \0 \pi_- \0 \eta$ and replaced $\id_G$ by $\eta \0 \eta$ in the second step.
  The inversion $\eta:  G_\He  \to  G_\He $ is Poisson by Theorem \ref{theorem:heisenberg_double_inversion_poisson_actions} (ii), so that
  \begin{align*}
    & (T\pi_1 \ox T\pi_2) \0 Th^{\ox 2} \, w_\pi (d_1, d_2) = (T\eta \ox T(L_{d_2} \0 \eta \0 \pi_-)) \, w_\He  (d_1^{-1}) \\
    \stackrel{\eqref{eq:global_double_projections_r_matrix_i_2}}=& (T\eta \ox T(L_{d_2} \0 \eta \0 \pi_-)) \0 TR_{d_1^{-1}}^{\ox 2} \, r_{21} \, .
  \end{align*}
  Use $r_{21} \in \g_+ \ox \g_-$ and the computation rules from Lemma \ref{lemma:global_double_projections_properties} to obtain
  \begin{equation}
    \label{eq:theorem_KitaevFockRoslyIso_proof_h_is_Poisson2}
    (T\pi_1 \ox T\pi_2) \0 Th^{\ox 2} \, w_\pi (d_1, d_2)  = (TL_{d_1} \ox TL_{d_2 \, \pi_-(d_1^{-1})^{-1}}) \, r_{21}  \stackrel{\eqref{eq:theorem_KitaevFockRoslyIso_proof-1}}{=} (T\pi_1 \ox T\pi_2) \, w'_{\FR} (h(d_1,d_2)) \, .
  \end{equation}
  This shows that $h: G_\He ^2 \to (G^2, w'_{\FR})$ is Poisson.

  Now we consider Case (c) without our earlier assumptions,  as pictured in Figure \ref{fig:trafo_kitaev_to_fr_proof}.
  This means that the paths $p_b(e_1), p_f(e_1)$ and $p_f(e_2)$ are not necessarily identity morphisms and there may be edges between $e_1$ and $e_2$ in the linear ordering of $v = s(e_1) = s(e_2)$.
  We have to show that the following map is Poisson
  \begin{equation*}
    \Phi': K \to (G^2, w'_{\FR}) \qquad \gamma \mapsto (\pi_{e_1}(\Phi(\gamma)), \pi_{e_2}(\Phi(\gamma))) \, ,
  \end{equation*}
  where $w'_{\FR}$ is again the Fock-Rosly bivector for a graph that consists only of the edges $e_1 : v \to t(e_1)$ and $e_2: v \to t(e_2)$ with $e_1 < e_2$ at $v$ and $t(e_1) \neq t(e_2)$.

  For this we decompose $\Phi' = \phi_4 \0 \phi_3 \0 \phi_2 \0 \phi_1$ into Poisson maps $\phi_i, i=1,\dots,4$.

  Let $e_{s,1} < \dots < e_{s,m} < e_1$ be the edges at $v = s(e_1) = s(e_2)$ that come before $e_1$ in the linear ordering at $v$.
  Denote the edges at $t(e_1)$ that come before $e_1$ with regard to the ordering at $t(e_1)$ by $e_{t,1} < \dots < e_{t, n} < e_1$.
  To the edges $e_{s, 1}, \dots, e_{s, m}, e_1, e_2, e_{t, 1}, \dots, e_{t, n}$ we associate the product Poisson manifold $ G_\He ^m \x  G_\He ^2 \x  G_\He ^n$.
  Define the path
  \[
    p(e_2)' := p_f(e_2) \0 f(e_2) \0 r(e_2) \0 p_b(e_2) \0 (b(e_1) \0 p_b(e_1))^{-1} \, .
  \]
  This path is the same as $p(e_2)$ from Equation \eqref{eq:kitaev_fock_rosly_iso_path_for_an_edge} for the shifted linear ordering at $v$ that corresponds to a cilium at $t(b(e_1))$.
  The path $p(e_2)'$ and the edges $e_{s, 1}, \dots, e_{s, m}$ and $e_{t, 1}, \dots, e_{t, n}$ are illustrated in Figure \ref{fig:trafo_kitaev_to_fr_proof1}.

  \begin{figure}[h]
  \centering
  \begin{tikzpicture}[vertex/.style={circle, fill=black, inner sep=0pt, minimum size=2mm}, plain/.style={draw=none, fill=none}, scale=0.8]
    \begin{scope}[rotate=-18]
    \node (p) at (0,0) {$v$};
    \coordinate (p1) at ($(p) + (36:0.553)$);
    \coordinate (p2) at ($(p) + (108:0.553)$);
    \coordinate (p3) at ($(p) + (180:0.553)$);
    \coordinate (p4) at ($(p) + (252:0.553)$);
    \coordinate (p5) at ($(p) + (-36:0.553)$);

    \coordinate (q) at ($(p) + (2,0) + (0.553,0) + (0.375,0)$);
    \coordinate (q1) at ($(q) + (0:0.375)$);
    \coordinate (q2) at ($(q) + (120:0.375)$);
    \coordinate (q3) at ($(q) + (240:0.375)$);

    \draw [-latex] (p1)--(q2);
    \draw [-latex] (q3)--(q2);
    \draw [-latex] (p5)--(p1);
    \draw [-latex] (p5)--(q3);
    \node at ($(p)!1.553cm!(q)$) {$e_2$};

    \draw [-latex] (q1) -- (q2);
    \draw [-latex] (q1) -- ($(q1)!2cm!-90:(q2)$);
    \draw [-latex] (q2) -- ($(q2) + (q1)!2cm!-90:(q2) - (q1)$);
    \draw [-latex] ($(q1)!2cm!-90:(q2)$) -- ($(q2) + (q1)!2cm!-90:(q2) - (q1)$);

    \draw [-latex] (q1) -- (q3);
    \draw [-latex] ($(q1)!2cm!90:(q3)$) -- (q1);
    \draw [-latex] ($(q3)!2cm!-90:(q1)$) -- (q3);
    \draw [-latex] ($(q1)!2cm!90:(q3)$) -- ($(q3)!2cm!-90:(q1)$);

    \draw [decorate,decoration={snake, amplitude=0.5mm}] (q1) -- ($(q1)+(0.75,0)$);

    \draw [-latex] (p2) -- (p1);
    \draw [-latex] ($(p1)!2cm!-90:(p2)$) -- (p1);
    \draw [-latex] ($(p2)!2cm!90:(p1)$) -- (p2);
    \draw [-latex] ($(p2)!2cm!90:(p1)$) -- ($(p1)!2cm!-90:(p2)$);

    \draw [-latex, color=green] (p3) -- (p2);
    \draw [-latex, color=green] ($(p2)!2cm!-90:(p3)$) -- (p2);
    \draw [-latex, color=green] ($(p3)!2cm!90:(p2)$) -- (p3);
    \draw [-latex, color=green] ($(p3)!2cm!90:(p2)$) -- ($(p2)!2cm!-90:(p3)$);
    \node at ($(p)!1.4cm!36:(p2)$) {$e_{s,1}$};

    \draw [-latex] (p3) -- (p4);
    \draw [-latex] (p3) -- ($(p3)!2cm!-90:(p4)$);
    \draw [-latex] (p4) -- ($(p4)!2cm!90:(p3)$);
    \draw [-latex] ($(p3)!2cm!-90:(p4)$) -- ($(p4)!2cm!90:(p3)$);
    \node at ($(p3)!0.5!(p4) + (p3)!1cm!-90:(p4) - (p3)$) {$e_1$};

    \draw [-latex] (p5) -- (p4);
    \draw [-latex] ($(p4)!2cm!-90:(p5)$) -- (p4);
    \draw [-latex] ($(p5)!2cm!90:(p4)$) -- (p5);
    \draw [-latex] ($(p5)!2cm!90:(p4)$) -- ($(p4)!2cm!-90:(p5)$);

    \draw [decorate,decoration={snake, amplitude=0.5mm}] (p2) -- ($(p2) + (p)!0.75cm!(p2)$);

    \coordinate (r1) at ($(p3)!2cm!-90:(p4)$);
    \coordinate (r2) at ($(p4)!2cm!90:(p3)$);
    \coordinate (r3) at ($(r2)!0.65cm!60:(r1)$);
    \coordinate (r) at ($0.333*(r1) + 0.333*(r2) + 0.333*(r3)$);

    \draw [-latex, color=orange] (r2) -- (r3);
    \draw [latex-, color=orange] (r2) -- ($(r2)!2cm!90:(r3)$);
    \draw [latex-, color=orange] (r3) -- ($(r3)!2cm!-90:(r2)$);
    \draw [-latex, color=orange] ($(r2)!2cm!90:(r3)$) -- ($(r3)!2cm!-90:(r2)$);
    \node at ($(r)!1.2cm!60:(r1)$) {$e_{t,1}$};

    \draw [-latex, color=orange] (r1) -- (r3);
    \draw [-latex, color=orange] (r1) -- ($(r1)!2cm!-90:(r3)$);
    \draw [-latex, color=orange] (r3) -- ($(r3)!2cm!90:(r1)$);
    \draw [-latex, color=orange] ($(r1)!2cm!-90:(r3)$) -- ($(r3)!2cm!90:(r1)$);
    \node at ($(r)!1.2cm!60:(r3)$) {$e_{t,2}$};

    \draw [decorate,decoration={snake, amplitude=0.5mm}] (r1) -- ($(r1)!0.75cm!150:(r2)$);

    \draw [-latex,color=blue] ($(p4) + (p)!0.3cm!(p4)$) -- ($(p5) + (p)!0.3cm!(p5)$);
    \draw [-latex,color=blue] ($(p5) + (p)!0.3cm!(p5)$) -- ($(q3) + (q)!0.2cm!(q3) - (q)$) -- ($(q2) + (q)!0.2cm!(q2) - (q)$);
    \draw [-latex,color=blue] ($(q2) + (q)!0.2cm!(q2) - (q)$) -- ($(q1) + (q)!0.2cm!(q1) - (q)$);

    \end{scope}
  \end{tikzpicture}
  \caption{The path $p(e_2)'$ is pictured in blue, the edges $e_{s, 1}, \dots, e_{s, m}$ in green and $e_{t, 1}, \dots, e_{t, n}$ in orange}
  \label{fig:trafo_kitaev_to_fr_proof1}
  \end{figure}
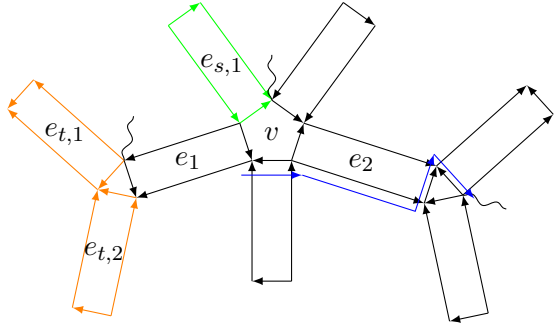

  The map $\phi_1 : K \to  G_\He ^{m} \x  G_\He ^2 \x  G_\He ^n$ assigns the holonomy along $p(e_2)'$ to the copy of $ G_\He $ that corresponds to $e_2$:
  \begin{equation}
    \label{eq:iso_kitaev_fock_rosly_phi1}
    \pi_e \0 \phi_1 =
    \begin{cases}
      \Hol(p(e_2)') & \text{if } e = e_2 \\
      \pi_e & \text{otherwise} \, .
    \end{cases}
  \end{equation}
  It is Poisson by Lemma \ref{lemma:projectionLocallyPoisson} (ii) and Theorem \ref{theorem:heisenberg_double_inversion_poisson_actions} (ii).

  The map $\phi_2:  G_\He ^m \x  G_\He ^2 \x  G_\He ^n \to  G_\He ^m \x (G^2, w'_{\FR}) \x  G_\He ^n$ is defined by
  \begin{equation}
    \label{eq:iso_kitaev_fock_rosly_phi2}
    \phi_2 := \id_{G^m} \x h \x \id_{G^n}
  \end{equation}
  with $h:  G_\He ^2 \to (G^2, w'_{\FR})$ from Equation \eqref{eq:theorem_KitaevFockRoslyIso_proof-2}.
  We have shown that $h$ is Poisson in Equations \eqref{eq:theorem_KitaevFockRoslyIso_proof_h_is_Poisson1} and \eqref{eq:theorem_KitaevFockRoslyIso_proof_h_is_Poisson2}.
  This implies that $\phi_2$ is a Poisson map.

  Denote by $\pi_s : K \to  G_\He ^m$ the projection on the components associated with the edges $e_{s,1}, \dots, e_{s,m}$ incident at $v$ and by $\pi_{t}: K \to  G_\He ^n$ the projection associated with the edges $e_{t,1}, \dots, e_{t,n}$ incident at $t(e_1)$.
  Let $hol_b :  G_\He ^m \to G_-$ be the map that is uniquely defined by $hol_b \0 \pi_s = \Hol(p_b(e_1))$ with the path $p_b(e_1)$ from \eqref{eq:kitaev_fock_rosly_iso_path_for_an_edge}.
  Similarly, let $hol_{f}:  G_\He ^n \to G_-$ be the unique map with $hol_{f} \0 \pi_{t} = \Hol(p_f(e_1))$.
  The map $\phi_3 :  G_\He ^m \x (G^2, w'_{\FR}) \x  G_\He ^n \to G_- \x (G^2, w'_{\FR}) \x G_-$ is defined by
  \begin{equation}
    \label{eq:iso_kitaev_fock_rosly_phi3}
    \phi_3 = (\eta_{G_-} \0 hol_b) \x \id_{G^2} \x hol_f \, ,
  \end{equation}
  where $\eta_{G_-}: G_- \to G_-$ is the inversion map.
  By Lemma \ref{lemma:projectionLocallyPoisson} (ii) and Theorem \ref{theorem:heisenberg_double_inversion_poisson_actions} (ii), the maps $hol_f:  G_\He ^n \to G_-$ and $\eta_{G_-} \0 hol_b:  G_\He ^m \to G_-$ are Poisson.
  Therefore, $\phi_3$ is Poisson.

  We define $\phi_4: G_- \x (G^2, w'_{\FR}) \x G_- \to (G^2, w'_{\FR})$ by
  \begin{equation}
    \label{eq:iso_kitaev_fock_rosly_phi4}
    (x_b, d_1, d_2, x_f) \mapsto (x_f d_1 x_b^{-1}, d_2 x_b^{-1}) \, .
  \end{equation}
  Recall that $(G^2, w'_{\FR})$ is the Fock-Rosly space for the graph that has only the edges $e_1: v \to t(e_1)$ and $e_2 : v \to t(e_2)$ with $t(e_1) \neq t(e_2)$ and $e_1 < e_2$ in the ordering at $v$.
  Using the associated vertex actions $\rhd_v^{\FR}, \rhd_{t(e_1)}^{\FR} : G \x (G^2, w'_{\FR})$ at $v$ and $t(e_1)$ from Equation \eqref{eq:fock_rosly_vertex_action}, we can write
  \[
    \phi_4 (x_f, d_1, d_2, x_b) = x_b \rhd_v^{\FR} (x_f \rhd_{t(e_1)}^{\FR} (d_1, d_2)) \, .
  \]
  The actions $\rhd_v^{\FR}, \rhd_{t(e_1)}$ are Poisson by Proposition \ref{proposition:fockrosly_poisson_g_space}, and therefore so is $\phi_4$.

  One can check easily that $\Phi' = \phi_4 \0 \phi_3 \0 \phi_2 \0 \phi_1$ by inserting the definitions of $\phi_i, i=1,\dots, 4$ from Equations \eqref{eq:iso_kitaev_fock_rosly_phi1}, \eqref{eq:iso_kitaev_fock_rosly_phi2}, \eqref{eq:iso_kitaev_fock_rosly_phi3} and \eqref{eq:iso_kitaev_fock_rosly_phi4}.
  Hence, the map $\Phi'$ is Poisson.

  Now consider the general case, where a path $p: v \to v$ in $\Gamma$ may traverse any number of edges. 
  We transform the graph $\Gamma$ by introducing a bivalent vertex $v_m$ on every edge $e$.
  This splits $e$ into the two edges $e_1: s(e) \to v_m$ and $e_2: v_m \to t(e)$.
  Choose the linear ordering at $v_m$ such that $e_2 > e_1$.
  By repeating this procedure one more time, we obtain a graph $\Gamma'$ so that any non-trivial closed path in $\Gamma'$ traverses at least four edges. 
  Denote the transformed Poisson-Kitaev model by $(K', \Gamma')$ and the set of newly created vertices by $L' \subseteq V'$. 
  From Lemma \ref{lemma:kitaev_gluing_trafos} we obtain a Poisson map $\psi : K' \to K$ that corresponds to gluing the split edges back together.

  The resulting graph $\Gamma'$ is not paired anymore.
  However, the map $\Phi$ from Equation \eqref{eq:kitaev_fock_rosly_iso} can be defined for any doubly ciliated ribbon graph.
  So far our proof for its Poisson property did not require a paired graph, either.
  Denote by $\FR'$ the Fock-Rosly space for $G$  associated to $\Gamma'$, where we set $r(v)=r$ for all vertices $v$ in Equation \eqref{eq:fockrosly_vertex_bivector}. 
  As $\Gamma'$ fulfils our previous assumptions,  we thus obtain the Poisson map $\Phi' : K' \to \FR'$ defined by \eqref{eq:kitaev_fock_rosly_iso}.

  There is a map $\psi_{\FR} : \FR' \to \FR$ that implements gluing the split edges of $\Gamma'$ back together given by Equation \eqref{eq:fock_rosly_gluing_edges}.
  It is Poisson by Proposition \ref{proposition:fockrosly_graph_trafos}.
  A direct computation shows that the equation
  \begin{equation}
  \label{eq:trafo_kitaev_to_fr_proof_gluing}
    (\Phi \0 \psi)(\gamma') = (\psi_{\FR} \0 \Phi')(\gamma')
  \end{equation}
  holds for all $\gamma' \in K'_{L'}$, where $K'_{L'}$ is the subset of $K'$ consisting of elements flat at $L'$.

  Let $f_1, f_2 \in C^\infty(\FR, \R)$ and $\gamma \in K$.
  As the map $\psi|_{K'_{L'}} : K'_{L'} \to K$ is surjective by Lemma \ref{lemma:kitaev_gluing_trafos} (v), we can choose $\gamma' \in K'_{L'}$ with $\psi(\gamma')=\gamma$.
  We compute
  \begin{align}
    & \left\{ f_1, f_2 \right\}_{\FR} \0 \Phi (\gamma) = \left\{ f_1, f_2 \right\}_{\FR} \0 (\Phi \0 \psi) (\gamma') \nonumber \\
    \stackrel{\eqref{eq:trafo_kitaev_to_fr_proof_gluing}}=& \left\{ f_1, f_2 \right\}_{\FR} \0 (\psi_{\FR} \0 \Phi') (\gamma') = \left\{ f_1 \0 \psi_{\FR} \0 \Phi', f_2 \0 \psi_{\FR} \0 \Phi' \right\}_{K'} (\gamma') \, , \label{eq:trafo_kitaev_to_fr_proof_gluing1}
  \end{align}
  where we used that $\psi_{\FR} : \FR' \to \FR$ and $\Phi': K' \to \FR'$ are Poisson in the last equation.
  By Lemma \ref{lemma:kitaev_gluing_trafos} (ii), one has $f_i \0 \Phi \0 \psi \in C^\infty(K', \R)^v_{L'}$ for all $v \in L'$ and $i=1,2$.
  Because $\psi_{\FR} \0 \Phi'$ coincides with $\Phi \0 \psi$ on $K'_{L'}$ and $K'_{L'}$ is stable under the vertex action $\rhd_v$ for $v \in L'$ by Lemma \ref{lemma:kitaev_flat_subspace_stable}, the maps $f_i \0 \psi_{\FR} \0 \Phi'$ are also in $C^\infty(K', \R)^v_{L'}$ for $i=1,2$.
  We can therefore apply Lemma \ref{lemma:locally_flat_poisson_submanifold} (ii) to the right hand side of \eqref{eq:trafo_kitaev_to_fr_proof_gluing1} to obtain
  \begin{align*}
    & \left\{ f_1, f_2 \right\}_{\FR} \0 \Phi (\gamma) = \left\{ f_1 \0 \Phi \0 \psi, f_2 \0 \Phi \0 \psi \right\}_{K'} (\gamma') \\
    =& \left\{ f_1 \0 \Phi, f_2 \0 \Phi \right\}_K (\psi(\gamma')) = \left\{ f_1 \0 \Phi, f_2 \0 \Phi \right\}_K (\gamma) \, ,
  \end{align*}
  where we used that $\psi : K' \to K$ is Poisson in the second equation and $\psi(\gamma') = \gamma$ in the third.

  \textbf{Statement (ii):}
  Equation \eqref{eq:KitaevFockRoslyIso_compatible_site_actions} is equivalent to the equations
  \begin{align}
    \label{eq:KitaevFockRoslyIso_actions_compatible_proof0}
    \alpha \rhd_v^{\FR} \Phi(\gamma) &= \Phi(\alpha \rhd_v \gamma) \qquad \Forall \alpha \in G_+, \gamma \in K \\
    x \rhd_v^{\FR} \Phi(\gamma) &= \Phi(x \rhd_f \gamma) \qquad \Forall x \in G_-, \gamma \in K \, .
    \label{eq:KitaevFockRoslyIso_actions_compatible_proof1}
  \end{align}
  As $\Gamma$ is paired, there are no loops at $v$ and we can assume that all edges are incoming.
  Then $\rhd_v$ is given by Formula \eqref{eq:kitaev_vertex_action_easy} and Equation \eqref{eq:KitaevFockRoslyIso_actions_compatible_proof0} can be checked directly using Lemma \ref{lemma:global_double_projections_properties}.
  To see \eqref{eq:KitaevFockRoslyIso_actions_compatible_proof1}, use the fact that there are no edges that occur twice in the face path $p(f)$ to assume that all edges of $f$ are oriented clockwise.
  Then $\rhd_f$ is given by \eqref{eq:kitaev_face_action_easy} and \eqref{eq:KitaevFockRoslyIso_actions_compatible_proof1} follows from the properties of $\pi_\pm$ in Lemma \ref{lemma:global_double_projections_properties}.

  \textbf{Statement (iii)}:
  This follows from a direct computation using Lemma \ref{lemma:global_double_projections_properties}.
\end{proof}

\begin{remark}~
  \begin{compactenum}
  \item
    As $K$ is the product Poisson manifold $K =  G_\He ^{\x E}$, the pullback along the transformation $\Phi : K \to \FR$ from Theorem \ref{theorem:KitaevFockRoslyIso} decouples the contributions of different edges to the Poisson bivector $w_{\FR}$ of $\FR$.
    However, while the gauge transformation $\rhd^{\FR}_v : G \x \FR \to \FR$ at a vertex is given by the diagonal action, the corresponding action $\rhd_{(v,f)} : G \x K \to K$ involves all  edges associated with the site $(v,f)$.
  \item
    Theorem \ref{theorem:KitaevFockRoslyIso} requires a global double Poisson-Lie group $G$.
    However, a weaker version of Theorem \ref{theorem:KitaevFockRoslyIso} can be formulated for general double Poisson-Lie groups.
    In the general case, the projections $\pi_\pm$ from Lemma \ref{lemma:projectionLocallyPoisson} are only defined on an open neighbourhood of the unit $U \subseteq G$.
    We can choose $U$ small enough so that Formula \eqref{eq:kitaev_fock_rosly_iso} is well-defined on $U^{\x E}$ and use this formula to define the map $\Phi': U^{\x E} \to G^{\x E}$.
    For a sufficiently small open neighbourhood of the unit $U' \subseteq G^{\x E}$ we can also define the map $\Psi': U'^{\x E} \to G^{\x E}$ by Formula \eqref{eq:iso_kitaev_fock_rosly_inverse}.
    Then $O := \Phi'^{-1}(U'^{\x E})$ is an open neighbourhood of the unit and a direct computation using a local version of Lemma \ref{lemma:global_double_projections_properties} shows that $\Psi' (\Phi'(x)) = x$ for all $x \in O$.
    In particular, the tangent map $T_x \Phi'$ is a linear isomorphism, so that $\Phi'|_O: O \to G^{\x E}$ is a local diffeomorphism.
    Thus, $\Phi'|_O$ is open and the corestriction $\Phi'|_O^{\Phi'(O)} : O \to \Phi'(O)$ is a diffeomorphism onto the open neighbourhood $\Phi'(O)$ of the unit.
    The proof for the Poisson property of $\Phi$ can also be applied to $\Phi' : U^{\x E} \to G^{\x E}$ (using local versions of the lemmas that were utilized).
    It  follows that $\Phi'|_O^{\Phi' (O)}$ is a Poisson isomorphism with respect to the product Poisson bivector $w_K$ on $O$ and the bivector $w_{\FR}$ from \eqref{eq:fockrosly_vertex_bivector} on $\Phi'(O)$.

    Most of our constructions for Poisson-Kitaev models are based on the properties of the projections $\pi_\pm$ from Lemma \ref{lemma:projectionLocallyPoisson}.
    Therefore, one could generalize Poisson-Kitaev models to (general) double Poisson-Lie groups $G$, but many definitions, such as the holonomy functor or vertex and face actions, would have to be replaced by local versions.
  \end{compactenum}
\end{remark}

\subsection{Relation with moduli spaces of flat \texorpdfstring{$G$}{G}-bundles for surfaces with boundary}
\label{subsection:relation_with_moduli_spaces}

In Section \ref{subsection:graph_trafos_kitaev}, we introduced graph transformations that allow us to transform the graph $\Gamma$ of a Poisson-Kitaev model $(K, \Gamma)$ into a paired doubly-ciliated ribbon graph $\Gamma'$.
By Theorem \ref{theorem:KitaevFockRoslyIso}, the Poisson manifold $K'$ associated with the paired graph $\Gamma'$ is isomorphic to the Fock-Rosly space $\FR'$ on $\Gamma'$.
In \cite{fockrosly98}, Fock and Rosly have shown that the Poisson structure on $\FR'$ induces the canonical symplectic structure on the moduli space of flat $G$-bundles $\Hom(\pi_1(S'),G)/G$ for the oriented  surface $S'$ obtained by gluing annuli to the faces of $\Gamma'$.
We use this to relate the Poisson algebra $\mathcal A(\Gamma, L)$ from Definition \ref{def:functions_on_flat_elements} (ii) to moduli spaces of flat $G$-bundles.

Recall from Proposition \ref{proposition:fockrosly_goldman} that the subalgebra of invariant functions on $\FR'$ is isomorphic to the canonical Poisson algebra of functions on $\Hom(\pi_1(S'), G)/G$ defined by the non-degenerate symmetric component $r_s$ of the classical $r$-matrix of $G$.
In our situation, $r$ is the classical $r$-matrix of the double Lie bialgebra $\g$ from Theorem \ref{theorem:double_lie_bialgebra}. 

Let $\Gamma$ be a connected doubly-ciliated ribbon graph and $(v_1, f_1), \dots, (v_n, f_n)$ pairwise distinct sites  with $n \geq 1$.
Consider the oriented  surface $S$ obtained by gluing annuli to $f_1, \dots, f_n$ and disks to all other faces.
Denote the set of flat vertices and faces by $L := (V \setminus\left\{ v_1, \dots, v_n \right\}) \dot \cup (F \setminus \left\{  f_1, \dots, f_n \right\})$.
By Proposition \ref{proposition:poisson_subalgebra_reduce_to_flat_subspace} (ii), the set $\mathcal A(\Gamma, L) = C^\infty(K, \R)^{inv}_L / \sim$ inherits a Poisson algebra structure from the Poisson subalgebra $C^\infty(K, \R)^{inv}_L \subseteq C^\infty(K, \R)$.
We obtain:

\begin{theorem}[Poisson-Kitaev models and moduli spaces of flat $G$-bundles]
  \label{theorem:kitaev_moduli_space}
  The Poisson algebra $\mathcal A(\Gamma, L)$ is isomorphic to the Poisson algebra of functions on the moduli space of flat $G$-bundles $\Hom(\pi_1(S), G)/G$ for the bilinear form dual to $r_s$.
\end{theorem}

\begin{proof}
  By Corollary \ref{corollary:kitaev_transform_into_paired_graph}, there is a paired doubly ciliated ribbon graph $\Gamma'$ together with $n$ sites $(v_1', f_1'), \dots, (v_n', f_n')$ such that the Poisson algebras $\mathcal A(\Gamma, L)$ and $\mathcal A(\Gamma', L')$ are isomorphic, where $L' = (V' \setminus\left\{ v_1', \dots, v_n' \right\}) \dot \cup (F' \setminus \left\{ f_1', \dots, f_n' \right\})$.
  The oriented  surface obtained by gluing annuli to $f_1', \dots, f'_n$ and disks to all other faces of $\Gamma'$ is homeomorphic to $S$.
  
  From Theorem \ref{theorem:KitaevFockRoslyIso} we obtain a Poisson isomorphism $\Phi: K' \to \FR'$, where $\FR'$ is the Fock-Rosly space for the Poisson-Lie group $G$ and the graph $\Gamma'$.
  The element $r(v)$ in Equation \eqref{eq:fockrosly_vertex_bivector} for the Poisson bivector $w_{\FR}$ is given by the classical $r$-matrix of $G$ for all $v \in V'$.
  Define the subspace of $\FR'$ of elements flat at $L'$ by
  \[
    \FR'_{L'} := \left\{ \gamma \in \FR' \, \mid \, \Hol_{\FR}(p(f)) (\gamma) = 1_G \Forall f \in L' \cap F' \right\} \, .
  \]
  Consider the set of functions invariant on $\FR'_{L'}$:
  \[
    C^\infty(\FR', \R)^{inv}_{L'} := \left\{ f \in C^\infty(\FR', \R) \, \mid \, f(g \rhd^{\FR}_v \gamma) = f(\gamma) \Forall \gamma \in \FR'_{L'}, v \in V', g \in G \right\} \, .
  \]
  Similarly to $\mathcal A(\Gamma', L')$ from Definition \ref{def:functions_on_flat_elements} (ii), we define the quotient
  \[
    \mathcal A_{\FR} (\Gamma', L') := C^\infty(\FR', \R)^{inv}_{L'} / \sim
  \]
  with respect to the relation $h_1 \sim h_2 \Leftrightarrow h_1 |_{\FR'_{L'}} = h_2 |_{\FR'_{L'}}$.

  Note that $\Hol_{\FR}(p(f)) = \Hol_{\FR}(p(v) \0 p(f))$ for each site $(v,f)$ (see Equation \eqref{eq:fock_rosly_holonomy_functor}).
  By Equation \eqref{eq:KitaevFockRoslyIso_compatible_holonomies}, one has $\Phi(K'_{L'}) = \FR'_{L'}$ for the set $K'_{L'}$ from Definition \ref{def:flatness}.
  At each site $(v,f)$ the map $\Phi$ intertwines the $G$-actions $\rhd_{(v,f)}$ and $\rhd_{v}^{\FR}$ (Equation \eqref{eq:KitaevFockRoslyIso_compatible_site_actions}).
  Therefore, $\Phi$ induces a bijection $\Phi^*: C^\infty(\FR', \R)^{inv}_{L'} \to C^\infty(K', \R)^{inv}_{L'}, h \mapsto h \0 \Phi$ for the Poisson subalgebra $C^\infty(K', \R)^{inv}_{L'}$ from Definition \ref{def:functions_on_flat_elements} (i).
  Because $\Phi(K'_{L'}) = \FR'_{L'}$, this bijection satisfies $\Phi^*(h_1)|_{K'_{L'}} = \Phi^{*}(h_2)|_{K'_{L'}}$ if $h_1|_{\FR'_{L'}} = h_2|_{\FR'_{L'}}$.
  It thus induces a bijection
  \[
    \Phi^*_{/\sim} : \mathcal A_{\FR}(\Gamma', L') \to \mathcal A(\Gamma', L') \qquad [h] \mapsto [\Phi^*(h)] \, .
  \]
  One can equip $\mathcal A_{\FR}(\Gamma', L')$ with a Poisson bracket such that the bijection $\Phi^*_{/\sim}$ becomes an isomorphism of Poisson algebras.
  Because $\Phi: K' \to \FR'$ is a Poisson isomorphism, this Poisson bracket on $\mathcal A_{\FR}(\Gamma', L')$ is the one induced by the Poisson bracket on $\FR'$.

  We have shown that the Poisson algebra $\mathcal A(\Gamma, L)$ is isomorphic to $\mathcal A_{\FR}(\Gamma', L')$.
  It remains to show that $\mathcal A_{\FR}(\Gamma', L')$ is isomorphic to the canonical Poisson algebra of functions on $\Hom(\pi_1(S),G)/G$.
  For this we show that $\mathcal A_{\FR}(\Gamma', L')$ is isomorphic to the Poisson subalgebra $C^\infty(\FR'', \R)^{inv}$ of invariant functions on the Fock-Rosly space $\FR''$ for a suitable graph $\Gamma''$, and that the oriented  surface $S''$ obtained by gluing annuli to the faces of $\Gamma''$ is homeomorphic to $S$.
  By Proposition \ref{proposition:fockrosly_goldman}, the Poisson algebra $C^\infty(\FR'', \R)^{inv}$ is isomorphic to the Poisson algebra of functions on the moduli space $\Hom(\pi_1(S''), G)/G$ for the bilinear form dual to $r_s$, which then proves Theorem \ref{theorem:kitaev_moduli_space}.

  We construct the graph $\Gamma''$ from $\Gamma'$ by merging the faces in $F' \cap L'$ into adjacent faces until only the selected faces $f_1', \dots, f_n'$ are left.
  More specifically, we transform $\Gamma'$ as follows.

  Choose a flat face $f_{flat, 1} \in F' \cap L'$ and choose an edge $e_1 \in E$ adjacent to $f_{flat, 1}$ such that the faces to the left and right of $e_1$ differ.
  Such an edge $e_1$ exists, since otherwise $\Gamma'$ would have only one face or several connected components in contradiction to our assumptions.
  Remove the edge $e_1$, thus merging the face $f_{flat, 1}$ into a different, adjacent face $f$.
  Denote the transformed graph by $\Gamma'_1$.
  The topological spaces obtained by gluing a disc to $\Gamma'$ along $f_{flat,1}$ and an annulus (disc) along $f$, and by gluing an annulus (disc) to $\Gamma'_1$ along $f$ are homeomorphic.
  Thus, the oriented  surface obtained from $\Gamma'_1$ by gluing annuli to $f'_1, \dots, f'_n$ and disks to the faces in $(F' \cap L') \setminus \left\{ f_{flat, 1} \right\}$ is homeomorphic to $S$.

  Then choose another face $f_{flat, 2}$ in $(F' \cap L') \setminus \left\{ f_{flat, 1} \right\}$.
  The graph $\Gamma'_1$ remains connected, and by the same argument as for $e_1$, there is an edge $e_2$ adjacent to $f_{flat,2}$ and some other face.
  Remove $e_{2}$ and thus $f_{flat, 2}$ from $\Gamma'_1$ to obtain a graph $\Gamma'_2$ with face set $F' \setminus \left\{ f_{flat, 1}, f_{flat, 2} \right\}$.
  The oriented  surface obtained from $\Gamma'_2$ by gluing annuli to $f'_1, \dots, f'_n$ and disks to the faces in $(F' \cap L') \setminus \left\{ f_{flat,1}, f_{flat, 2} \right\}$ is again homeomorphic to $S$.
  Repeat this procedure for all remaining faces in $(F' \cap L') \setminus \left\{ f_{flat, 1}, f_{flat, 2} \right\}$.

  The resulting graph $\Gamma''$ is a subgraph of $\Gamma'$ with $n$ faces, and the oriented  surface $S''$ obtained from $\Gamma''$ by gluing annuli to its faces is homeomorphic to $S$.
  Denote the associated Fock-Rosly space by $\FR''$.

  To the transformation $\Gamma' \to \Gamma''$ we associate a Poisson map $\FR' \to \FR''$ as follows.
  Denote by $\FR'_i$ the Fock-Rosly space associated to the graph $\Gamma'_i$, which we obtain by removing the edges $e_1, \dots, e_i$ from $\Gamma'$.
  Set $\FR'_0 := \FR'$ and let $\phi_i: \FR_{i-1}' \to \FR'_i$ be the map from Equation \eqref{eq:fock_rosly_erasing_edge} associated with erasing the edge $e_i$.
  There is a right inverse $\psi_i: \FR'_i \to \FR'_{i-1}$ of $\phi_i$ that is uniquely defined by
  \begin{equation}
    \label{eq:kitaev_moduli_space_proof0}
    \Hol_{\FR}(f_{flat, i}) \0 \psi_i = (\gamma \mapsto 1_G) \, .
  \end{equation}
  It satisfies $\pi_e \0 \psi_i = \pi_e$ for all edges $e$ of $\Gamma'_{i-1}$ other than $e_i$.
  This implies
  \begin{equation}
    \label{eq:kitaev_moduli_space_proof1}
    (\psi_i \0 \phi_i) (\gamma) = \gamma \qquad \Forall \gamma \in \FR'_{i-1} \quad \text{with} \quad \Hol_{\FR}(f_{flat, i})(\gamma) = 1_G \, .
  \end{equation}
  A direct computation shows:
  \begin{align}
    \label{eq:kitaev_moduli_space_proof2}
    \Hol_{\FR}(f) (\psi_i(\gamma)) &= \Hol_{\FR}(f)(\gamma) \qquad \Forall \gamma \in \FR_{i}', f \in F' \setminus \left\{ f_{flat, 1}, \dots, f_{flat, i} \right\} \\
    \label{eq:kitaev_moduli_space_proof3}
    \psi_i(g \rhd_v^{\FR} \gamma) &= g \rhd_v^{\FR} \psi_i(\gamma) \qquad\, \Forall v \in V', g \in G, \gamma \in \FR'_{i} \, .
  \end{align}
  By composing the maps $\phi_i$, we obtain the map $\phi: \FR' \to \FR''$ that satisfies
  \[
    \pi''_e (\phi(\gamma)) = \pi'_e(\gamma) \qquad \Forall e \in E'' \, ,
  \]
  where $\pi''_e, \pi'_e$, respectively, are the projections on the component of $\FR'', \FR'$ that is associated with the edge $e$.
  By Proposition \ref{proposition:fockrosly_graph_trafos}, this is a Poisson map and it intertwines the Poisson actions $\rhd_v^{\FR}$ for $v \in V' = V''$.
  It therefore induces a homomorphism of Poisson algebras
  \[
    \phi^*_{/\sim}: C^\infty(\FR'', \R)^{inv} \to \mathcal A_{\FR}(\Gamma', L') \qquad f \mapsto [f \0 \phi] \, .
  \]
  By composing the right inverses $\psi_i$ we obtain a right inverse $\psi: \FR'' \to \FR'$ of $\phi$.
  Equations \eqref{eq:kitaev_moduli_space_proof0} and \eqref{eq:kitaev_moduli_space_proof2} imply that $\psi(\FR'') \subseteq \FR'_{L'}$.
  Together with Equation \eqref{eq:kitaev_moduli_space_proof3} this implies that the map $\psi$ induces a map
  \[
    \psi^*_{/\sim}: \mathcal A_{\FR}(\Gamma', L') \to C^\infty(\FR'', \R)^{inv} \qquad [f] \mapsto f \0 \psi \, .
  \]
  A direct computation shows that the map $\phi_i : \FR'_{i-1} \to \FR'_i$ associated with erasing the edge $e_i$ satisfies for all $\gamma \in \FR'_{i-1}$ with $\Hol_{\FR}(f_{flat,i}) = 1_G$:
  \[
    \Hol_{\FR}(f)(\phi_i (\gamma)) = \Hol_{\FR}(f)(\gamma) \qquad \Forall f \in F' \setminus \left\{ f_{flat, 1} , \dots, f_{flat, i} \right\} \, .
  \]
  Therefore, one has for all $i$ and $\gamma \in \FR'_{L'}$
  \[
    \Hol_{\FR}(f) (\phi_i \0 \dots \0 \phi_1 (\gamma) ) = 1_G \qquad \Forall f \in (F' \cap L') \setminus\left\{ f_{flat, 1}, \dots, f_{flat, i} \right\} \, .
  \]
  Together with Equation \eqref{eq:kitaev_moduli_space_proof1} this implies for all $i$ and $\gamma \in \FR'_{L'}$
  \[
    (\psi_1 \0 \dots \0 \psi_i \0 \phi_i \0 \dots \0 \phi_1) (\gamma) = \gamma \, .
  \]
  This implies $\psi \0 \phi |_{\FR'_{L'}} = \id_{\FR'_{L'}}$ and the map $\psi^*_{/\sim}$ is an inverse to $\phi^*_{/\sim}$.
  This concludes the proof.
\end{proof}

\begin{remark}~
  \label{remark:to_theorem_kitaev_moduli_space}
  \begin{compactenum}
  \item
    Theorem \ref{theorem:kitaev_moduli_space} implies in particular that the Poisson algebra $\mathcal A(\Gamma, L)$ does not depend on the graph $\Gamma$, but only on the homeomorphism class of the oriented  surface $S$ obtained by gluing annuli to $f_1, \dots, f_n$ and disks to all other faces.
  \item
    By Remark \ref{remark:excitations}, the Poisson algebra $\mathcal A(\Gamma,L)$ corresponds to a subalgebra of the endomorphism algebra of the space $\mathcal L_L \subseteq H^{\ox E}$ from \eqref{eq:quantum_kitaev_protected_space_with_excitations}, which describes a quantum Kitaev model with excitations at the sites $(v_1, f_1), \dots, (v_n, f_n)$.
    Theorem \ref{theorem:kitaev_moduli_space} thus relates the moduli space $\Hom(\pi_1(S), G)/G$ to a Kitaev model with excitations, and the excitations correspond to the boundary components of $S$. 
  \item
    Alekseev and Malkin have introduced a decoupling transformation for the moduli space of flat $H$-bundles for a semi-simple quasi-triangular Poisson-Lie group $H$ in \cite{alekseevmalkin95}.
    Their construction is based on a set of generators of the fundamental group $\pi_1(S)$ of the compact oriented surface $S$. 
    They showed that the symplectic structure on the moduli space $\Hom(\pi_1(S), H)/H$ can be obtained from the product Poisson manifold
    \[
      (H^*)^{n} \x \He(H)^{g}  \, .
    \]
    It consists of a copy of the dual Poisson-Lie group $H^*$ for each of the $n$ boundary components of $S$ and $g$ copies of the Heisenberg double $\He(H)$ of $H$  (Definition \ref{def:classical_heisenberg_double}), where $g$ is the genus of $S$ \cite[Theorem 2]{alekseevmalkin95}.
    If $H= G$ is a global double Poisson-Lie group, then we obtain the symplectic structure on $\Hom(\pi_1(S), G)/G$ from the product Poisson manifold $K= G_\He ^{\x E}$ by Theorem \ref{theorem:kitaev_moduli_space}.
    This can be viewed as  a generalization of the decoupling transformation from \cite{alekseevmalkin95} to paired doubly ciliated ribbon graphs.
    Note that we do not require $G$ to  be semi-simple, but that $G$ is a global double. 

    That each boundary component of $S$ is associated with a copy of $G^*$ can also be seen for Poisson-Kitaev models.
    Let $(v,f)$ be one of the selected sites $(v_1, f_1), \dots, (v_n, f_n)$ from Theorem \ref{theorem:kitaev_moduli_space} that corresponds to a boundary component of the surface $S$.
    In Proposition \ref{proposition:kitaev_commutation_relations} (iii) we have shown that the combined holonomy around the site $(v,f)$ is a Poisson map $\Hol^{(v,f)} : K \to (G, w_{G^*})$.
    By Lemma \ref{lemma:dual_poisson_lie_group} there is a Poisson map $G^* \to (G, w_{G^*})$ that is a local diffeomorphism in a neighbourhood of the unit $1_{G^*}$.
  \end{compactenum}
\end{remark}


\urlstyle{same}
\bibliographystyle{abbrv}

\bibliography{bibliography}

\end{document}